\documentclass[a4paper,UKenglish,cleveref,numberwithinsect]{lipics-v2021}
\pdfoutput=1 %uncomment to ensure pdflatex processing (mandatatory e.g. to submit to arXiv)
\hideLIPIcs  %uncomment to remove references to LIPIcs series (logo, DOI, ...), e.g. when preparing a pre-final version to be uploaded to arXiv or another public repository

\usepackage{xcolor}
\usepackage[title]{appendix}
\usepackage{url}  %Required
\usepackage{tikz}
\usetikzlibrary{arrows.meta}
\usepackage{hyperref}
\hypersetup{
    colorlinks=true,       
    linkcolor=blue,        
    citecolor=red,         
}
\usepackage{algorithm}
\usepackage{algorithmicx}
\usepackage[noend]{algpseudocode}

\algrenewcommand\algorithmicrequire{\textbf{Input:}}
\algrenewcommand\algorithmicensure{\textbf{Output:}}

\newenvironment{varalgorithm}[1]
  {\algorithm\renewcommand{\thealgorithm}{#1}}
  {\endalgorithm}

\makeatletter
\newcommand{\algmargin}{\the\ALG@thistlm}
\makeatother
\algnewcommand{\parState}[1]{\State%
    \parbox[t]{\dimexpr\linewidth-\algmargin}{\strut\hangindent=\algorithmicindent \hangafter=1 #1\strut}}
    
\definecolor{lightblue}{RGB}{119,170,221} 
\definecolor{orange}{RGB}{238,136,102}
\definecolor{pear}{RGB}{187,204,51} 
\definecolor{lightcyan}{RGB}{153,221,255} 
\definecolor{lightyellow}{RGB}{238,221,136} 
\definecolor{palegrey}{RGB}{221,221,221}
\definecolor{pink}{RGB}{255,170,187}
\definecolor{olive}{RGB}{170,170,0}
\definecolor{mint}{RGB}{68,187,153}
\definecolor{darkgrey}{RGB}{80,80,80}  

\colorlet{FPTcolor}{mint}
\colorlet{inPcolor}{lightcyan}
\colorlet{NPhardcolor}{orange}
\colorlet{XPcolor}{lightyellow}

\newcommand\new[1]{#1}

\usepackage{comment}
\usepackage{mathtools}
\newcommand{\R}{\mathbb{R}}
\newcommand{\N}{\mathbb{N}}
\newcommand{\NP}{\mathsf{NP}}
\newcommand{\symmdiff}{\mathbin{\triangle}}
\newcommand{\ol}[1]{\overline{#1}}
\newcommand{\dist}[1][]{\mathsf{dist}_{s^{#1},t^{#1}}}
\newcommand{\init}{\textup{init}}
\newcommand{\inc}{\textup{inc}}
\newcommand{\vsel}{\textup{vert}}
\newcommand{\esel}{\textup{edge}}

\newcommand{\FF}{\mathcal{F}}

\newcommand{\LL}{\mathcal{L}}

\newcommand{\PP}{\mathcal{P}}

\newcommand{\CC}{\mathcal{C}}

\newcommand{\EO}{\mathrm{EO}}
\newcommand{\fixed}{\mathsf{fixed}}
\newcommand{\cont}{\mathsf{contained}}
\newcommand{\SR}{\textsc{Robust Submodular Minimizer}}
\newcommand{\ASR}{\textsc{Anchored Submodular Minimizer}}
\def\InstSR{I_\textup{RSM}}
\def\InstCS{I_\textup{RS}}
\def\shortASR{ASM}
\newcommand{\MBDCF}{\textsc{Multi-Budgeted Directed Cut with Forbidden Edges}}
\newcommand{\MBDC}{\textsc{Multi-Budgeted Directed Cut}}
\newcommand{\centralset}{\textsc{Robust Separation}}
\newcommand{\MBMC}{\textsc{Most Balanced Minimum Cut}}
\newcommand{\BMC}{\textsc{Perfectly Balanced Minimum Cut}}
\def\tt{\texttt{true}}
\def\ff{\texttt{false}}
\def\IH{$\mathcal{IH}$}

\title{Parameterized Complexity of Submodular Minimization under Uncertainty}

\author{Naonori Kakimura}{Department of Mathematics, Keio University, Japan}{kakimura@math.keio.ac.jp}{https://orcid.org/0000-0002-3918-3479}{Supported by JSPS KAKENHI Grant Numbers JP22H05001, JP	23K21646, and JP21H03397, Japan and JST ERATO Grant Number JPMJER2301, Japan.}

\author{Ildik\'o Schlotter}{HUN-REN Centre for Economic and Regional Studies, Budapest, Hungary \and Budapest University of Technology and Economics, Budapest, Hungary}{schlotter.ildiko@krtk.hun-ren.hu}{https://orcid.org/0000-0002-0114-8280}{Supported by the Hungarian Academy of Sciences under its J\'anos Bolyai Research Scholarship and its Momentum Programme (LP2021-2).}

\authorrunning{N. Kakimura and I. Schlotter} 

\Copyright{Naonori Kakimura and Ildik\'o Schlotter} %TODO mandatory, please use full first names. LIPIcs license is "CC-BY";  http://creativecommons.org/licenses/by/3.0/

\ccsdesc[500]{Theory of computation~Submodular optimization and polymatroids}
\ccsdesc[500]{Theory of computation~Fixed parameter tractability}
\ccsdesc[500]{Mathematics of computing~Combinatorial optimization} %TODO mandatory: Please choose ACM 2012 classifications from https://dl.acm.org/ccs/ccs_flat.cfm 

\keywords{Submodular minimization, optimization under uncertainty, parameterized complexity, cut function.} %TODO mandatory; please add comma-separated list of keywords

\category{} %optional, e.g. invited paper

\relatedversiondetails[linktext={Published in \emph{Proc.\ of the 19th Scandinavian Symposium on Algorithm Theory (SWAT 2024)},  LIPIcs, volume 294, pp.\ 30:1-30:17, 2024.}]{Conference Version}{https://doi.org/10.4230/LIPIcs.SWAT.2024.30} %linktext and cite are optional

\acknowledgements{The authors are grateful to the organizers of the 14th Emléktábla Workshop which took place in Vác, Hungary in 2023, and provided a great opportunity for our initial discussions. \new{We are also grateful for the anonymous referees of SWAT 2024 for their helpful comments.}}%optional

\nolinenumbers %uncomment to disable line numbering

\EventEditors{John Q. Open and Joan R. Access}
\EventNoEds{2}
\EventLongTitle{42nd Conference on Very Important Topics (CVIT 2016)}
\EventShortTitle{CVIT 2016}
\EventAcronym{CVIT}
\EventYear{2016}
\EventDate{December 24--27, 2016}
\EventLocation{Little Whinging, United Kingdom}
\EventLogo{}
\SeriesVolume{42}
\ArticleNo{23}
\begin{document}

\maketitle

%% Abstract
\begin{abstract}
This paper studies the computational complexity of a robust variant of a two-stage submodular minimization problem that we call \SR{}.
In this problem, we are given $k$ submodular functions~$f_1,\dots,f_k$ over a set family~$2^V$, which represent $k$ possible scenarios in the future when we will need to find an optimal solution for one of these scenarios, i.e., a minimizer for one of the functions. 
The present task is to find a set $X \subseteq V$ that is close to \emph{some} optimal solution for each $f_i$ in the sense that some minimizer of~$f_i$ can be obtained from $X$ by adding/removing at most $d$ elements for a given integer $d \in \mathbb{N}$.
The main contribution of this paper is to provide a complete computational map of this problem with respect to parameters~$k$ and~$d$, which reveals a tight complexity threshold for both parameters:
\begin{itemize}
\item \SR{} can be solved in polynomial time when $k \leq 2$, but is $\NP$-hard if $k$ is a constant with $k \geq 3$.
\item \SR{} can be solved in polynomial time when $d=0$, but
is $\NP$-hard if $d$ is a constant with $d \geq 1$.
\item \SR{} is fixed-parameter tractable when parameterized by~$(k,d)$.
\end{itemize}
We also show that if some submodular function $f_i$ has a polynomial number of minimizers, then the problem becomes fixed-parameter tractable when parameterized by $d$.
On the other hand, the problem remains $\mathsf{W}[1]$-hard parameterized by $k$ even if each function $f_i$ has \new{at most~$|V|$} minimizers.
We remark that all our hardness results hold even if each submodular function is given by a cut function of a directed graph.
\end{abstract}

\section{Introduction}
\label{sec:intro}
This paper proposes a two-stage robust optimization problem under uncertainty.
Suppose that we want to find a minimum cut on a directed graph under uncertainty.
The uncertainty here is represented by $k$ directed graphs $G_1, \dots, G_k$ on the same vertex set $V \cup \{s,t\}$.
%The uncertainty here is represented by $k$ directed graphs $G_1, \dots, G_k$ with edge costs $c_1, \dots, c_k$ on the same vertex set $V$ and the same two distinct vertices $s$ and $t$.
That is, we have $k$ possible scenarios of graph realizations in the future.
At the moment, we want to choose an $(s,t)$-cut in advance, so that after the graph is revealed, we will be able to obtain a minimum $(s,t)$-cut in the graph with small modification.
Therefore, our aim is to find an $(s,t)$-cut that is close to some minimum $(s,t)$-cut in each graph $G_i$ for $i=1,\dots, k$.

Let us formalize this problem.
For a vertex set $X$ in a directed graph~$G={(V \cup\{s,t\},E)}$, 
the \emph{cut function} $f: 2^V\to \mathbb{Z}$ % is 
\new{takes} the number of out-going edges from~$X$ 
\new{as its value}.
%%\[
%f(X) = \sum_{e\in \delta (X)} c(e).
%\]
Let us denote the family of minimum $(s,t)$-cuts in~$G$ by $\mathcal{C}_{s,t}(G)$, that is,
$\CC_{s,t}(G)=\{{Y\subseteq V}:{ f(Y)\leq f(Y')} \ \forall Y'\subseteq V\}$.
Given directed graphs $G_1,\dots,G_k$ over a common vertex set~${V \cup \{s,t\}}$, we want to find a subset $X \subseteq V$ and sets~$Y_i \in \mathcal{C}_{s,t}(G_i)$
for each $i\in [k]$ that minimizes  $\max_{i\in [k]} |X \symmdiff Y_i|$
%\begin{equation}
%  \text{minimize} \qquad \max_{i\in [k]} |X \symmdiff Y_i| \label{eq:prob0}
%\end{equation}
where $\symmdiff$ stands for symmetric difference and  $[k]$ denotes $\{1,\dots, k\}$ for any positive integer~$k$.

We study a natural generalization of this problem where, instead of the cut functions of directed graphs which are known to be submodular~\cite{schrijver-book}, we consider arbitrary submodular set functions over some non-empty finite set~$V$.
Let $f_1,\dots,f_k \colon 2^V \to \R$ be $k$ submodular functions.
Let $\arg\min f_i=\{Y\subseteq V: f_i(Y)\leq f_i(Y')~\forall Y'\subseteq V\}$  refer to the set of minimizers of $f_i$.
%$\arg\min f_i$ refer to the set of minimizers of $f_i$, i.e., $\arg\min f_i=\{Y\subseteq V: f_i(Y)\leq f_i(Y')~\forall Y'\subseteq V\}$.
We want to find a subset $X \subseteq V$ and sets $Y_i \in \arg\min f_i$
for all $i\in [k]$ that
\begin{linenomath*}
\begin{equation*}
  \text{minimize} \qquad \max_{i\in [k]} |X \symmdiff Y_i|. \label{eq:prob0}
\end{equation*}
\end{linenomath*}
We call the decision version of this problem 
%\emph{submodular reassignment problem}.
\SR{}.
\begin{center}
\fbox{ 
\parbox{13cm}{
\begin{tabular}{l}\SR{}:  \end{tabular} \\
\begin{tabular}{p{1cm}p{11cm}}
Input: & A finite set~$V$, submodular functions~$f_1,\dots,f_k:2^V \rightarrow \R$, and an integer~$d \in \N$. \\
Task: & Find a set $X \subseteq V$ such that for each $i \in [k]$ there exists $Y_i \in \arg\min f_i$ with $|X \symmdiff Y_i| \leq d$, or detect if no such set exists.
\end{tabular}
}}
\end{center}
%When $k=1$, this is equivalent to the efficiently solvable submodular function minimization problem \cite{LeeSW15},
%in which we are given a single submodular function $f\colon 2^V \to \R$ and
%want to find a set $X \subseteq V$ in $\arg\min f$.
%

We remark that the min-sum variant of the problem, that is, the problem obtained by replacing the condition $\max_{i\in [k]} |X \symmdiff Y_i| \leq d$ with $\sum_{i\in [k]} |X \symmdiff Y_i| \leq d$, was introduced by Kakimura et al.~\cite{KakimuraKKO22}, who
showed that it can be solved in polynomial time.
%via a max-flow algorithm, while the min-max problem is $\NP$-complete.

\subsection{Our Contributions and Techniques}

The contribution of this paper is to reveal the complete computational complexity of \SR{} with respect to the parameters~$k$ and~$d$, 
\new{as well as parameters~$\min_{i \in [k]} |\arg\min f_i|$ and $\max_{i \in [k]} |\arg\min f_i|$.

To phrase this more precisely,
let us define a \emph{parameterized restriction}\footnote{See Section~\ref{sec:prelim} for the concepts from parameterized complexity theory that are most relevant to us, and for some references to further background.} of \SR\  as a problem obtained from \SR\ by restricting 
each of the four values~$k$, $d$, $\min_{i \in [k]} |\arg\min f_i|$, and $\max_{i \in [k]} |\arg\min f_i|$
as either (i) \emph{unbounded}, (ii) a \emph{parameter}, or (iii) a \emph{constant}, with the additional possibility for declaring either the values 
$\min_{i \in [k]} |\arg\min f_i|$ and $\max_{i \in [k]} |\arg\min f_i|$ to be (iv) \emph{polynomially bounded in~$|V|$}.
Our results completely classify every parameterized restriction of \SR\ as either 
(a) fixed-parameter tractable, 
(b) $\mathsf{W}[1]$-hard but in~$\mathsf{XP}$, or 
(c) para-$\NP$-hard with the given parameterization; in the case when there are no parameters, we classify the problem as either (d) polynomial-time solvable or (e) $\NP$-hard.
See Figure~\ref{fig-decision-diagram} at the end of the paper for a decision diagram that can be used to determine the computational complexity of each parameterized restriction of \SR.}

%We also discuss the case when the submodular functions have only polynomially many minimizers. 
\new{Below we highlight our main results;
a more detailed summary  can be found in Table~\ref{tab:summary}.}
 
%We also provide an algorithm for the case when one of the submodular functions has only polynomially many minimizers. Our results are as follows:
\begin{enumerate}
\item \SR{} can be solved in polynomial time when $k \leq 2$ (Theorem~\ref{thm:k=2}), but is $\NP$-hard if $k$ is a constant with $k \geq 3$ (Corollary~\ref{cor:k-atleast-3}).
\item \SR{} can be solved in polynomial time when $d=0$ (Observation~\ref{obs:d=0}), but 
is $\NP$-hard if $d$ is a constant with $d \geq 1$ (Theorem~\ref{cor:SR-NPhard-d-atleast1}).
\item \SR{} is fixed-parameter tractable when parameterized by~$(k,d)$.
\item If the size of $\arg\min f_i$ for some $i\in [k]$ is polynomially bounded, then \SR{} is fixed-parameter tractable when parameterized by~$d$~(Theorem~\ref{thm:unique}).
On the other hand, the problem is W[1]-hard parameterized by $k$ even when the size of $\arg\min f_i$ for each $i\in [k]$ is polynomially bounded~(Theorem~\ref{thm:W1hard-k}).
%\item \SR{} is $\NP$-hard even if $k$ is a constant with $k \geq 3$.
\end{enumerate}

%Let us describe the details of our results.
%The cases when $k=1$ or $d=0$ can be easily handled.
\new{Let us now describe our techniques.}
When $k=1$, \SR{} is equivalent to the efficiently solvable submodular function minimization problem~\cite{LeeSW15}, in which we are given a single submodular function $f\colon 2^V \to \R$ and
want to find a set $X \subseteq V$ in $\arg\min f$.
It is not difficult to observe that \SR{} for $d=0$ can also be solved in polynomial time by computing a minimizer of the submodular function $\sum_{i=1}^k f_i$; see Section~\ref{sec:alg-poly}.

The rest of our positive results are based on 
Birkhoff's representation theorem on distributive lattices~\cite{Birkhoff37} that allows us to maintain the family of minimizers of a submodular function in a compact way.
Specifically, 
even though the number of minimizers may be exponential in the input size,
we can represent all  minimizers as a family of cuts in a directed acyclic graph with polynomial size.
As we show in Section~\ref{sec:alg-poly}, we can use this representation to solve an instance~$I$ of \SR{} with $k=2$ by constructing a directed graph with two distinct vertices, $s$ and $t$, in which a minimum $(s,t)$-cut yields a solution for~$I$.
More generally, 
Birkhoff's compact representation allows us to reduce  \SR{} for arbitrary~$k$ to the so-called \MBDC{} problem, solvable by an algorithm due to Kratsch et al.~\cite{KratschLMPW20}, leading to a fixed-parameter tractable algorithm for the parameter~$(k,d)$.
We note that a similar construction was used to show that the min-sum variant of the problem is polynomial-time solvable~\cite{KakimuraKKO22}.

Section~\ref{sec:hardness} contains our $\NP$-hardness results for the cases when either $d$ is a constant at least~$1$, or $k$ is a constant at least~$3$. \new{For the former result}, we present a reduction from an intermediate problem that may be of independent interest: in this problem, we are given $k$ set families~$\FF_1,\dots,\FF_k$ over a universe~$V$ containing two distinguished elements, $s$ and~$t$, with each $\FF_i$ containing pairwise disjoint subsets of~$V$; the task is to find a set~$X \subseteq V$ containing~$s$ but not~$t$ that has a bounded distance from each family~$\FF_i$ for a specific distance measure.

In Section~\ref{sec:polynomialminimizers}, we consider the case when some or all of the $k$ submodular functions have only polynomially many minimizers.
As mentioned in~\cite{KakimuraKKO22}, \SR{} is $\NP$-hard even when each submodular function $f_i$ has a unique minimizer.
In fact, the problem is equivalent to the \textsc{Closest String} problem over a binary alphabet, shown to be $\NP$-hard under the name \textsc{Minimum Radius} by Frances and Litman~\cite{DBLP:journals/mst/FrancesL97}.
Since the \textsc{Closest String} problem is fixed-parameter tractable with parameter $k$, so is the problem when each submodular function has a constant number of minimizers by a brute force approach~(Proposition~\ref{prop:few-minimizers}).
However, if each submodular function has $|V|$ minimizers, Theorem~\ref{thm:W1hard-k} says that the problem is W[1]-hard parameterized by $k$.
% On the other hand,
\new{By contrast}, 
for the case when $|\arg\min f_i|$ is polynomially bounded for \textit{some}~$i \in [k]$, we present a fixed-parameter tractable algorithm parameterized only by $d$. 
Our algorithm guesses a set in~$\arg\min f_i$ and uses it as an ``anchor,'' then solves the problem recursively by the bounded search-tree technique.

\subsection{Related Work}\label{sec:relatedwork}

\SR{} is related to the \textit{robust recoverable} combinatorial optimization problem, introduced by Liebchen et al.~\cite{LiebchenLMS09}.
It is a framework of mathematical optimization that allows recourse in decision-making to deal with uncertainty.
In this framework, we are given a problem instance with some scenarios and a recovery bound $d$, and the task is to find a feasible solution $X$~(the first-stage solution) to the instance that can be transformed to a feasible solution $Y_i$~(the second-stage solutions) in each scenario $i$ respecting the recovery bound~(e.g., $|X\symmdiff Y_i|\leq d$ for each $i$).
The cost of the solution is usually evaluated by the sum of the cost of $X$ and the sum of the costs of $Y_i$'s.
Robust recoverable versions have been studied for a variety of standard combinatorial optimization problems.
Examples include the shortest path problem~\cite{Busing12}, the assignment problem~\cite{HLW21}, the travelling salesman problem~\cite{abs-2111-09691},
and others~\cite{HradovichKZ17, LachmannLW21, LendlPT22}.
The setting was originally motivated from the situation where the source of uncertainty was the cost function which changes in the second stage.
We consider another situation dealing with \emph{structural uncertainty}, where some unknown set of input elements can be interdicted~\cite{DouradoMPRPA15,ItoKKKO22}.
Recently, a variant of robust recoverable problems has been studied where certain operations are allowed in the second stage~\cite{HommelsheimMMP}.

\emph{Reoptimization} is another concept related to \SR{}.
In general reoptimization, 
we are given an instance~$I$ of a combinatorial optimization problem and an optimal solution~$X$ for~$I$.
Then, for a slightly modified instance~$I'$ of the problem, we need to make a small change to $X$ so that the resulting solution~$X'$ is an optimal~(or a good approximate) solution to the modified instance~$I'$.
Reoptimization has been studied for several combinatorial optimization problems such as the minimum spanning tree problem~\cite{BoriaP10}, the traveling salesman problem~\cite{Monnot15}, and the Steiner tree problem~\cite{BockenhauerFHMSS12}.

%Closest to our work is the paper by Kakimura et al.~\cite{KakimuraKKO22}
%who have investigated the min-sum variant of \SR{} for which they proposed a polynomial-time algorithm based on flow techniques.

\section{Preliminaries}
\label{sec:prelim}

\new{
In this section we introduce all the notation we need, and pinpoint the known results that are the most relevant for our study.}

\paragraph*{Graphs and Cuts}
Given a directed graph~$G=(V,E)$, we write~$uv$ for an edge pointing from~$u$ to~$v$. For  a subset~$X \subseteq V$ of vertices in~$G$, let $\delta_G(X)$ denote the set of edges leaving~$X$. 
If $G$ is an undirected graph, then $\delta_G(X)$ for some set~$X$ of vertices denotes the set of edges with exactly one endpoint in~$X$.
We may simply write $\delta (X)$ if the graph is clear from the context.

For two vertices~$s$ and~$t$ in a directed or undirected graph $G=(V,E)$, an \emph{$(s,t)$-cut} is a set~$X$  of vertices such that $s \in X$ but $t \notin X$. 
A \emph{minimum $(s,t)$-cut} in~$G$ is an $(s,t)$-cut~$X$ that minimizes $|\delta(X)|$.
Given a cost function $c\colon E\to \R_+ \cup \{+\infty\}$ on the edges of~$G$ %\textcolor{red}{cost or capacity?}
where $\R_+$ is the set of all non-negative real numbers, 
the \emph{(weighted) cut function} $\kappa_G:2^V \rightarrow \R_+ \cup \{+\infty\}$ is defined by
\begin{equation}
  \kappa_G (X) = \sum_{e \in \delta (X)} c(e).
  \label{eq:cutfnct}
\end{equation}
%The value $\kappa_G (X)$ is called the \emph{cost} of a cut.
%
A \emph{minimum-cost $(s,t)$-cut} is an $(s,t)$-cut~$X$ that minimizes $\kappa_G(X)$.
%The \textsc{minimum $(s,t)$-cut} problem is to find an $(s,t)$-cut with minimum cost. 

%\paragraph*{Submodular Functions}

%\subsection{Distributive Lattices}
\paragraph*{Distributive Lattices}
In this paper, we will make use of properties of finite distributive lattices on a ground set~$V$.
%To ease the notation, we minimize the use of terminology in poset theory, and we stick to the terminology in sets and graphs.

%For our purpose, 
A \emph{distributive lattice} is a set family $\LL \subseteq 2^V$ that
is closed under union and intersection, that is, $X, Y \in \LL$ implies $X\cup Y \in \LL$ and
$X\cap Y \in \LL$.
Then $\LL$ is a partially ordered set with respect to set inclusion $\subseteq$,
and has a unique minimal element and a unique maximal element.

\emph{Birkhoff's representation theorem} is a useful tool for studying distributive 
lattices.

\begin{theorem}[Birkhoff's representation theorem~\cite{Birkhoff37}]\label{thm:Birkhoff}
Let $\LL \subseteq 2^V$ be a distributive lattice.
Then there exists a partition of $V$ into $U_0, U_1, \dots, U_b, U_\infty$, where
$U_0$ and $U_\infty$ can possibly be empty, such that the following hold:
\begin{bracketenumerate}
    \item[\textup{(1)}] Every set in $\LL$ contains $U_0$.
    \item[\textup{(2)}] Every set in $\LL$ is disjoint from $U_\infty$.
    \item[\textup{(3)}] For every set $X \in \LL$, there exists a set $J \subseteq [b]$ of
indices such that $X = U_0\cup \bigcup_{j \in J} U_j$.
    \item[\textup{(4)}] There exists a directed acyclic graph $G(\LL)$ that has the following properties.
    \begin{alphaenumerate}
        \item The vertex set is $\{U_0, U_1, \dots, U_b\}$.
        \item $U_0$ is a unique sink\footnote{A \textit{sink} is a vertex of out-degree zero.} of $G(\LL)$.
        \item For a non-empty set $Z$ of vertices in $G(\LL)$, $Z$ has no out-going edge if and only if
$\bigcup_{U_j\in Z} U_j\in \LL$.
%$\bigcup_{j\in J} U_j\in \LL$, where $J=\{j: U_j\in Z\}$.
%\textcolor{red}{Isn't this the same as $\bigcup Z \in \LL$?$\rightarrow$ fixed}
    \end{alphaenumerate}
\end{bracketenumerate}
\end{theorem}

For a distributive lattice $\LL \subseteq 2^V$, 
we call the directed acyclic graph $G(\LL)$ above
a \emph{compact representation} of $\LL$.
Note that the size of $G(\LL)$ is $O(|V|^2)$ while
$|\LL|$ can be as large as $2^{|V|}$.
%This justifies the use of word ``compact.'' The graph $G(\LL)$ is unique if we further forbid a transitive arc.

%\begin{figure}[t]
%\centering
%\includegraphics[width=8cm]{birkhoff.pdf}
%\caption{Example for Birkhoff's representation theorem.
%The left is a distributive lattice $\LL$ on $U=\{1,\dots,8\}$
%(shown by its Hasse diagram), and the
%right is the directed graph $G(\LL)$.
%In this case, $U_\infty = \{8\}$.}
%\label{fig:birkhoff}
%\end{figure}

\paragraph*{Submodular Function Minimization}

Let $V$ be a non-empty finite set.
A function $f\colon 2^V \to \R$ is \emph{submodular} if 
%\begin{equation}
 $ f(X) + f(Y) \geq f(X \cup Y) + f(X \cap Y)$ 
%\end{equation}
for all $X, Y \subseteq V$.
A typical example of submodular functions is the cut function $\kappa_G$ of a directed (or undirected) edge-weighted graph~$G$ as defined in~(\ref{eq:cutfnct}).
If the graph $G=(V \cup \{s,t\}, E)$ contains two distinct vertices, $s$ and~$t$, then
we can restrict the cut function to the domain of $(s,t)$-cuts in the following sense:
each $X \subseteq V$ corresponds to an $(s,t)$-cut $X \cup \{s\}$ in~$G$;
%Then an \emph{$(s,t)$-cut} is a set of vertices of the form~$X \cup \{s\}$ for some $X \subseteq V$.
%, that is, an $s, t$-cut is a cut $X$ such that $s\in X$ and $t\not\in X$.
then the function $\lambda_G :2^V \to \mathbb{R}_+ \cup \{+\infty\}$ defined by $\lambda_G (X) = \kappa_G (X\cup \{s\})$ is submodular.

%A function $f\colon 2^U \to \R$ is \emph{supermodular} if $-f$ is submodular.

%It is known that a function $f$ is submodular if and only if $f(X\cup \{e\})-f(X)\geq f(Y\cup \{e\})-f(Y)$ for any $X\subseteq Y$ and $e\not \in Y$.
%Thus a submodular function demonstrates ``diminishing marginal return.''
%This fact allows us to capture diversity or synergy effects with submodular functions in applications.
%In particular, submodular functions are recently used to model diversity in matching markets and allocation problems, see e.g.,~\cite{agrawal2018proportional,Ahmed:2017,benabbou2018diversity,golz2018migration}.

When we discuss computations on a submodular function $f\colon 2^V \to \R$, we assume that we are given a \emph{value oracle} of $f$.
%\textbf{To be removed?} For the computational purpose, we define a \emph{value oracle} of a submodular function $f\colon 2^U \to \R$.
A value oracle takes $X \subseteq V$ as an input, and returns the value~$f(X)$.
Assuming that we are given a value oracle, we can minimize a submodular function in polynomial time.
The currently fastest algorithm for  submodular function minimization
was given by Lee et al.~\cite{LeeSW15} and
runs in $\tilde{O}(|V|^3 \EO + |V|^4)$ time, where $\EO$ is the query time of a value oracle.

Let $f\colon 2^V \to \R$ be a submodular function.
A subset $Y \subseteq V$ is a \emph{minimizer} of the function~$f$ if $f(Y) \leq f(Y')$ for all $Y' \subseteq V$.
The set of minimizers of $f$ is denoted by $\arg\min f$.
The following is a well-known fact on submodular functions.
%The proof is omitted. 
\begin{lemma}[See e.g., \cite{schrijver-book}]
\label{lem:argmin-lattice}
Let $f\colon 2^V \to \R$ be a submodular function.
Then $\arg\min f$ forms a distributive lattice.
\end{lemma}

%\begin{proof}
%Let $X, Y \in \arg\min f$, and denote $\alpha = \min\{f(Z) : Z \subseteq U\}$.
%Then, 
%$f(X \cup Y) \geq \alpha$, $f(X \cap Y) \geq \alpha$, and
%\[
%  2\alpha = f(X) + f(Y) \geq f(X \cup Y) + f(X \cap Y) \geq \alpha + \alpha = 2\alpha.
%\]
%Therefore, $f(X\cup Y) = f(X\cap Y) = \alpha$, which implies that 
%$X\cup Y, X\cap Y \in \arg\min f$.
%\end{proof}

A compact representation of the distributive lattice $\arg\min f$ can be constructed
in $\tilde{O}(|V|^5 \EO + |V|^6)$ time via Orlin's submodular function minimization algorithm~\cite{DBLP:journals/mp/Orlin09}.
See~\cite[Notes~10.11--10.12]{MurotaDCA}. 
Unless otherwise stated, we will assume that the submodular functions given in our problem instances are given via their compact representation.

%\footnote{%
%  It is not known how to construct a compact representation of $\arg\min f$ with the algorithm by 
%  Lee et al.~\shortcite{LeeSW15}.
%%  It is not clear to us that the algorithm by 
%%  Lee et al.~\shortcite{LeeSW15} can be
%%  used to find a compact representation of $\arg\min f$.
%}

As a special case, consider minimum $(s,t)$-cuts in a directed graph $G=(V \cup \{s,t\},E)$ with a positive cost function~$c$ on its edges.
%For each minimizer $X$ of the cut function~$\kappa_G$, the set $X\cup \{s\}$ is called a \emph{minimum $(s,t)$-cut}.
By Lemma \ref{lem:argmin-lattice}, the family of minimum $(s,t)$-cuts 
forms a distributive lattice.
A compact representation for this lattice can be constructed from a maximum flow in the $(s,t)$-network in linear time~\cite{Picard1980}.
%More specifically, construct a flow network with $s$ as the source, $t$ as the sink, $G$ as its underlying graph, and $c$ as its cost function. Compute  a maximum flow~$f$ in this network, construct the residual network for~$f$,
%and find the strongly-connected-component decomposition
%of the residual network.
%The decomposition gives a directed graph which can be turned into a 
%compact representation of the distributive lattice of minimum $(s,t)$-cuts~\cite{Picard1980}.
%In total, when we construct a compact representation from scratch,
Thus the running time is dominated by the maximum flow computation, and
this can be done in $O(|V||E|)$ time~\cite{DBLP:conf/stoc/Orlin13}.

\paragraph*{Parameterized Complexity}
In parameterized complexity, each input instance~$I$ of a \emph{parameterized problem}~$Q$ is associated with a \emph{parameter}~$k$, usually an integer or a tuple of integers,
and we consider the running time of any algorithm solving~$Q$ as not only a function of the input length~$|I|$, but also as a function of the parameter~$k$.
An algorithm for~$Q$ is \emph{fixed-parameter tractable} or FPT, if it runs in $g(k)|I|^{O(1)}$ time for some computable function~$g$. Such an algorithm can be efficient in practice if the parameter is small.
\new{
In cases when an FPT algorithm seems out of reach, it might still be possible to find an algorithm that runs in polynomial time for each fixed value of the parameter; such an algorithm is called an \emph{$\mathsf{XP}$ algorithm}. Clearly, if $Q$ is $\NP$-hard already for some fixed value of the parameter~$k$ --- in which case the problem is said to be \emph{para-$\NP$-hard with respect to~$k$} --- then we cannot expect an $\mathsf{XP}$-algorithm for it.
} 

\new{
Analogously to the concept of $\NP$-hardness, the \emph{$\mathsf{W}[1]$-hardness} of a parameterized problem shows that it is unlikely to be fixed-parameter tractable. Proving that some parameterized problem~$Q$ is $\mathsf{W}[1]$-hard entails giving a \emph{parameterized reduction} from a parameterized problem~$Q'$ that is already known to be $\mathsf{W}[1]$-hard to~$Q$.} 
See the books~\cite{CyganEtAl2015,downey-fellows-FPC-book} for more background.

\section{Algorithms for small $k$ or~$d$}
\label{sec:alg}
In this section, we present algorithms for \SR{} for the cases when $d=0$ or~$k\leq 2$, and an FPT algorithm for parameter~$(k,d)$ that is efficient when both~$k$ and~$d$ are small integers.
We start with a construction that we will use in most of our algorithms.
Let $\InstSR=(V,f_1,\dots,f_k,d)$ be our input instance.

For each $i\in [k]$, let $\mathcal{L}_i = \arg\min f_i$ denote the set of minimizers.
By Lemma~\ref{lem:argmin-lattice}, using Birkhoff's representation theorem 
%we can construct a
we may assume that $f_i$ is given through a compact representation $G(\mathcal{L}_i)$ of $\mathcal{L}_i$, whose vertex set is $\{U^i_0, U^i_1, \dots, U^i_{b_i}\}$ with $U^i_\infty = V\setminus \bigcup_{j=0}^{b_i} U^i_j$.

We then construct a directed graph $G^i$ from $G(\mathcal{L}_i)$ by expanding each vertex in $G(\mathcal{L}_i)$ to a complete graph.
More precisely, $G^i$ has vertex set $V^i\cup\{s, t\}$ where $V^i=\{v^i: v\in V\}$ is a copy of $V$,
and its edge set~$E^i$ is defined as follows.
\begin{itemize}
\item $u^i v^i\in E^i$ if $u, v\in U^i_j$ for some $j\in\{0,1,\dots, b_i, \infty\}$.
\item $u^iv^i\in E^i$ for any $u\in U^i_j$ and $v\in U^i_{j'}$ if $G(\mathcal{L}_i)$ has an edge from $U^i_j$ to $U^i_{j'}$.
\item $u^is\in E^i$ and $su^i\in E^i$ if $u\in U^i_0$.
\item $u^it\in E^i$ and $tu^i\in E^i$ if $u\in U^i_\infty$.
\end{itemize}

Recall that $G^i$ can be computed in $(|V|^5 \EO_i + |V|^6)$ time where $\EO_i$ denotes the query time of a value oracle for function~$f_i$.

%The cost $c(e)$ of an edge $e$ is set to be $+\infty$ for all $e\in E^i$.
We define the function $\lambda_i:2^{V^i}\to \mathbb{Z}_+$ so that $\lambda_i(X)=|\delta_{G^i}(X\cup \{s\})|$ for a subset~$X\subseteq V^i$. 
%We denote the cut function of the directed graph~$G_i$ by $\kappa_i$. Note that for any subset $X\subseteq V\cup \{s,t\}$ we see either $\kappa_i(X)=0$ or $\kappa_i (X)=+\infty$.
Then it is observed below that each subset $X\subseteq V^i$ with $\lambda_i (X)=0$ corresponds to a minimizer of $f_i$.

\begin{lemma}[{\cite[Lemma~3.2]{KakimuraKKO22}}]
\label{lem:minimizers}
Let $X$ be a subset in $V$, and $X^i=\{v^i\in V^i:v \in X\}$ its copy in~$G^i$.
Then $\lambda_i (X^i)=0$ if and only if $X\in \mathcal{L}_i$.
%Then $\kappa_i (X^i \cup \{s\})=0$ if and only if $X\in \mathcal{L}_i$.
\end{lemma}

The rest of the section is organized as follows.
In Section~\ref{sec:alg-poly} we present polynomial-time algorithms for the cases $d=0$ and~$k =2$.
In Section~\ref{sec:alg-FPT} we give an FPT 
%fixed-parameter tractable 
algorithm for the combined parameter~$(k,d)$. Section~\ref{sec:alg-unique} deals with the case when some function $f_i$ has only polynomially many minimizers, allowing for an FPT algorithm
with parameter~$d$.

%We are given $k$ submodular functions $f_1, \dots, f_k$ on the ground set $V$, and an integer bound~$d$. Let $\InstSR=(V,f_1, \dots, f_k,d)$ denote our instance.

\subsection{Polynomial-time algorithms}
\label{sec:alg-poly}

We start by observing that the case $d=0$ is efficiently solvable by computing a minimizer for the function $\sum_{i \in [k]} f_i$ which is also submodular; 
for this result, it will be convenient to use value oracles. 
\begin{observation}
\label{obs:d=0}    
\SR{} for $d=0$ can be solved in time $O(|V|^3 \sum_{i=1}^k {\EO}_i +|V|^4)$ 
where ${\EO}_i$ is the query time of a value oracle for~$f_i$ for $i \in [k]$.
\end{observation}

\begin{proof}
Let $I=(V,f_1,\dots,f_k,0)$ be our input.
Since the addition of submodular functions yields a submodular function, the function $f=\sum_{i \in [k]}f_i$ is also submodular. Clearly, if some set~$X$ is contained in~$\arg\min f_i$ for each $i \in [k]$, then $X \in \arg\min f$. Conversely, 
if some $X$ is contained in~$\arg\min f$ but $X \notin \arg\min f_i$ for some~$i \in [k]$, then there is no common minimizer~$Y$ of the functions $f_1,\dots,f_k$. Indeed, such a set~$Y$ would also necessarily minimize~$f$, implying
\[
f(Y)=f_i(Y)+\sum_{j \in [k],j \neq i} f_j(Y) <
f_i(X) +\sum_{j \in [k],j \neq i} f_j(Y)=f(X)
\]
where the inequality follows from our assumption that $X$ is not a minimizer for~$f_i$. However, this means that $X \notin \arg\min f$, a contradiction.

Therefore, to solve our input instance~$I$, it suffices to compute a minimizer~$X$ for the function $f$, which can be done in polynomial time, due to the submodularity of~$f$. 
Then we check whether $X \in \arg\min f_i$ for each $i \in [k]$ (by computing any minimizer for each function~$f_i$).
If  $X \in \arg\min f_i$ holds for each~$f_i$, then we return~$X$; otherwise we output ``No.''

The claimed running time follows from the fact that the value of~$f$  on some subset of~$V$ can be computed in $\sum_{i=1}^k \EO_i$ time, and thus we can minimize~$f$ in time  $O(|V|^3 \sum_{i=1}^k {\EO}_i +|V|^4)$  using the algorithm by Lee at al~\cite{LeeSW15}.
\end{proof}

%\subsection{\texorpdfstring{Polynomial-time algorithm for $k=2$}{Polynomial-time algorithm for k=2}}

From now on, we will assume that the functions~$f_1,\dots,f_k$ are given through a compact representation.
Under this condition, we next show that the problem is polynomial-time solvable for $k=2$.
We will need the following intuitive fact.

%The proof is similar to the previous section with the following observation.

\begin{proposition}
\label{prop:halfdiff}
Let $Y_1$, $Y_2$ be two subsets of a set~$V$.
Then $|Y_1\symmdiff Y_2|\leq 2d$ if and only if there exists a set~$X \subseteq V$ such that $|X\symmdiff Y_i|\leq d$ for each $i\in \{1,2\}$.
\end{proposition}

\begin{proof}
To prove the sufficiency, suppose that $|Y_1\symmdiff Y_2|> 2d$.
Then it holds that for any set $X \subseteq V$, 
\[
2d< |Y_1\symmdiff Y_2| \leq |Y_1\symmdiff X| + |X\symmdiff Y_2|.
\]
Hence at least one of $|Y_1\symmdiff X|$ and $|X\symmdiff Y_2|$ is more than $d$.
Thus no subset $X$ satisfies  $|X\symmdiff Y_i|\leq d$ for both $i=1,2$.

For the necessity, suppose that $|Y_1\symmdiff Y_2|\leq 2d$.
We denote $d_1 =|Y_1\setminus Y_2|$ and $d_2 = |Y_2\setminus Y_1|$.
We choose arbitrary subsets $Z_1 \subseteq Y_1\setminus Y_2$ with size $\lfloor d_1/2\rfloor$ and $Z_2 \subseteq Y_2\setminus Y_1$ with size $\lfloor d_2/2\rfloor$.
Note that $|Y_1\symmdiff Y_2| = d_1 + d_2 \leq 2d$.

Define $X=(Y_1\cap Y_2)\cup Z_1\cup Z_2$.
We claim that $X$ satisfies that $|X\symmdiff Y_i|\leq d$ for both $i=1,2$.
In fact,  it holds that 
\begin{equation}
\begin{aligned}
|X\symmdiff Y_1| & = |\left(Y_1 \setminus Y_2\right) \setminus Z_1|+|Z_2| \\
& = \left(d_1 - \lfloor d_1/2\rfloor\right) + \lfloor d_2/2\rfloor
= \lceil d_1/2\rceil + \lfloor d_2/2\rfloor . 
\end{aligned}
\label{eq:k2}
\end{equation}
If $d_1$ is even, then we get $|X\symmdiff Y_1| + d_1/2 + \lfloor d_2/2 \rfloor \leq  (d_1+d_2)/2 \leq d$.
Suppose that $d_1$ is odd and $d_2$ is odd.
Then \eqref{eq:k2} implies that
\[
|X\symmdiff Y_1| 
= (d_1+1)/2 + (d_2-1)/2
= (d_1+d_2)/2 \leq d.
\]  
Finally, suppose that $d_1$ is odd and $d_2$ is even.
Then, since $d_1 + d_2 \leq 2d$, we see that $d_1 + d_2 \leq 2d-1$ also holds.
Hence, by~\eqref{eq:k2} we get
\[
|X\symmdiff Y_1| 
= (d_1+1)/2 + d_2/2
= (d_1+d_2+1)/2 
\leq d.
\]
Thus $|X\symmdiff Y_1|$, and by symmetry this shows also that $|X\symmdiff Y_i|\leq d$ for each $i\in\{1,2\}$.
\end{proof}

%The lemma above offers a way to reduce \SR{} to a max-flow problem in the case $k=2$.

\begin{theorem}
\label{thm:k=2}
\SR{} for $k=2$ can be solved in polynomial time via a maximum flow computation.
\end{theorem}
\begin{proof}
Let our instance be $\InstSR=(V,f_1, f_2,d)$.
Using the method explained at the beginning of Section~\ref{sec:alg}, we 
construct the directed graphs $G^1=(V^1 \cup \{s,t\},E^2)$ and $G^2=(V^2\cup \{s,t\},E^2)$ for $f_1$ and $f_2$. %Section~\ref{sec:reduce-to-budgetedcut}.
We then construct a directed graph $\tilde{G} = (\tilde{V}, \tilde{E})$  by identifying $s$, as well as $t$, in $G^1$ and $G^2$, and then connecting the corresponding copies of each vertex with a bidirected edge.
That is, 
$\tilde{V}=V^1 \cup V^2 \cup \{s,t\}$ and $\tilde{E} = E^1 \cup E^2\cup E'$ where
$E' = \left\{ v^1v^2,v^2 v^1 : v\in V \right\}$.
We set $c(e)=+\infty$ for all edges~$e \in E^1 \cup E^2$, and we set $c(e)=1$ for all edges~$e \in E'$.

We next compute a minimum-cost $(s,t)$-cut~$Z$ in the graph~$\tilde{G}$ with cost function~$c$ using standard flow techniques. 
Let $\kappa_{\tilde{G}}$ denote the cut function in this graph.
We will show below that $\kappa_{\tilde{G}}(Z) \leq 2d$ if and only if the answer is ``yes''.

First suppose that $\kappa_{\tilde{G}}(Z) \leq 2d$.
Let $Y_1=\{v\in V:v^1 \in Z\}$ and $Y_2=\{v\in V:v^2 \in Z\}$.
Since $\delta_{\tilde{G}}(Z)$ has no edges in $E^1\cup E^2$, we see that $\lambda_i(\{v^i\in V^i:v \in Y_i\})=0$ for both $i=1,2$, and therefore the set~$Y_i$ is in $\mathcal{L}_i$ by Lemma~\ref{lem:minimizers}.
%let $Z$ be a subset of $V\setminus \{t\}$ containing~$s$ such that $\kappa_{\tilde{G}}(X \cup \{s\})$ is finite. 
%This implies $\kappa_i(\{v^i:v \in Y_i\} \cup \{s\})=0$, and therefore 
%the set~$Y_i$ is in $\mathcal{L}_i$ by Lemma~\ref{lem:minimizers}.
Since $|Y_1 \symmdiff Y_2| =  \kappa_{\tilde{G}}(Z) \leq 2d$, 
we can compute a set~$X$ such that $|X \symmdiff Y_i| \leq d$ for both $i=1,2$ by Proposition~\ref{prop:halfdiff}.
%and compute a set~$X$ such that $|X \symmdiff Y_i| \leq d$ for both $i=1,2$, as described in Lemma~\ref{lem:halfdiff}. By , such a set~$Z$ exists.

%If $\kappa_{\tilde{G}}(Z) > 2d$, then output ``No.''
%\smallskip

%To see the correctness of this algorithm, 
%let $Z$ be a subset of $V\setminus \{t\}$ containing~$s$ such that $\kappa_{\tilde{G}}(X \cup \{s\})$ is finite. 
%This implies $\kappa_i(\{v^i:v \in Y_i\} \cup \{s\})=0$, and therefore 
%the set~$Y_i$ is in $\mathcal{L}_i$ by Lemma~\ref{lem:minimizers}.

Conversely, let $X \subseteq V$ and $Y_i \in \mathcal{L}_i$ for each $i=1,2$ such that $|X \symmdiff Y_i|\leq d$. 
 Define $Z=\{s\} \cup \{v^1\in V^1:v \in Y_1\} \cup \{v^2\in V^2:v \in Y_2\}$.
 Due to Lemma~\ref{lem:minimizers} we know that
 $\lambda_i(\{v^i\in V^i:v \in Y_i\})=0$ for both $i=1,2$.  
 This implies $\kappa_{\tilde{G}}(Z)=|Y_1\symmdiff Y_2|\leq 2d$ where the inequality follows from  Proposition~\ref{prop:halfdiff}.
%Then, for any cut $X=X_1\cup X_2$ of $\tilde{G}$, the cut value is equal to $|X_1\symmdiff X_2|$.
%Therefore, by the min-cut algorithm, we can determine whether $\tilde{G}$ has a pair $X_1$ and $X_2$ such that $|X_1\symmdiff X_2|\leq d$.
%
%If $\tilde{G}$ has, then we can find a solution $X$ by Lemma~\ref{}.
%Otherwise, there exists no solution.
\end{proof}

\subsection{\texorpdfstring{FPT algorithm for parameter $(k,d)$}{FPT algorithm for k at least 3}}
\label{sec:alg-FPT}
We propose a
fixed-parameter tractable 
algorithm for \SR{} parameterized by $k$ and $d$; let $\InstSR=(V,f_1,\dots,f_k,d)$ denote our instance.

\begin{theorem}
\label{thm:SR-fpt}
\SR{} can be solved in FPT time when parameterized by~$(k,d)$.
\end{theorem}

To this end, we reduce our problem to the \MBDC{} problem~\cite{KratschLMPW20}, defined as follows.
We are given a directed graph $D=(V, E)$ with distinct vertices $s$ and~$t$, together with
 pairwise disjoint edge sets $A_1, \dots, A_k$, and positive integers $d_1, \dots, d_k$.
The task is to decide whether $D$ has an $(s, t)$-cut $X$ such that $|\delta (X)\cap A_i|\leq d_i$ for each $i \in [k]$.

\begin{proposition}[Kratsch et al.~\cite{KratschLMPW20}]
\label{prop:MBDC}
The \MBDC{} problem can be solved in FPT time when the parameter is $\sum_{i=1}^k d_i$.
\end{proposition}

In fact, we will need to use \emph{forbidden} edges,  so let us define the \MBDCF{} problem as follows. Given an instance $I_{\textup{MBC}}$ %=(D,s,t,A_1,\dots,A_k,d_1,\dots,d_k)$ 
of \MBDC{} and a set~$F$ of forbidden edges, find a solution $X$ for~$I_{\textup{MBC}}$ such that $\delta (X)$ is disjoint from~$F$.
It is straightforward to solve this problem using the results by Kratsch et al.~\cite{KratschLMPW20}, after replacing each forbidden edge with an appropriate number of parallel edges. Hence, we get the following.

\begin{proposition}
\label{prop:MBDCF}
The \MBDCF{} problem can be solved in FPT time when the parameter is $\sum_{i=1}^k d_i$.
\end{proposition}
\begin{proof}
Let $I=(D,s,t,A_1,\dots,A_k,d_1,\dots,d_k)$ be our instance with a set~$F$ of forbidden edges in~$D$.
%Let us define $A_{\cup}=A_1 \cup \dots \cup A_k$.
First, we create a set~$A_{k+1}=F \setminus \bigcup_{i=1}^k A_i$ and let $d_{k+1}=1$.
Next, for each $i \in [k+1]$ and each forbidden edge~$f \in F \cap A_i$, we replace~$f$ with $d_i+1$ parallel copies of~$f$ in~$D$.\footnote{To avoid using parallel edges, one can alternatively replace each edge $f=uv \in F \cap A_i$ with a set of $d_i+1$ paths of length two, i.e., with vertices $p_f^1,p_f^2,\dots,p_f^{d_i+1}$ and edges $\{u p_f^j, p_f^j v: j \in [d_i+1]\}$.} 
Let $D'$ be the resulting directed graph, 
and consider the instance $I'=(D',s,t,A_1,\dots,A_{k+1},d_1,\dots,d_{k+1})$ of \MBDC{}.
Clearly, if a set~$X$ of vertices is a solution for~$I'$, then $\delta(X)$ cannot contain a forbidden edge or a copy of a forbidden edge, as then it would have to contain \emph{all} copies of that forbidden edge, and thus would violate $|\delta(X) \cap A_i| \leq d_i$ for some $i \in [k+1]$. 
Thus, a set~$X$ of vertices in~$D$ is a solution for our original instance~$(I,F)$ of \MBDCF{} if and only if $X$ is a solution for the instance~$I'$ of \MBDC{}. 
Note that this construction increases the parameter by exactly one.
Hence, the algorithm provided by Proposition~\ref{prop:MBDC} can be used to solve \MBDCF{} in FPT time.
\end{proof}

\paragraph*{Reduction to %multi-budgeted directed cut 
\MBDCF{}}
\label{sec:reduce-to-budgetedcut}

%We denote the vertex set of $D^i$ except for $s, t$ by $U^i$, and the edge set by $A^i$.
Compute the graph~$G^i$ for each $i \in [k]$, as described at the beginning of Section~\ref{sec:alg}.
We construct a large directed graph $\tilde{G}=(\tilde{
V}, \tilde{E})$ as follows. 
We identify all vertices~$s$~(and $t$, respectively) in the graphs~$G^i$ into a single vertex~$s$ (and $t$, respectively).
We further prepare another copy of $V$, which is denoted by $V^\ast = \{v^\ast : v\in V\}$.
Thus the vertex set of $\tilde{G}$ is defined by 
%\[
$
\tilde{V}=\bigcup_{i=1}^k V^i \cup V^\ast \cup \{s, t\}.
$
%\]
The edge set of $\tilde{G}$ consists of $E^i$ and bidirected edges %(\textcolor{red}{Should we use edges or arcs?}) 
connecting $v^\ast$ and the copy~$v^i$ of~$v$ in~$G^i$, for each $i \in [k]$.
That is, 
\[
\tilde{E}=\bigcup_{i=1}^k \left(E^i\cup A^i\right) \text{\quad where \ \ } A^i = \left\{v^\ast v^i, v^i v^\ast : v\in V\right\}.
\]
We also set $d_i = d$ for each $i \in [k]$.
Consider the instance $I_{\textup{MBC}}=(\tilde{G}, s, t, \{A^i\}_{i=1}^k, \{d_i\}_{i=1}^k)$ of \textsc{multi-budgeted directed cut} with $F=\bigcup_{i=1}^k E^i$ as forbidden edges; note that its parameter is $k\cdot d$.
Theorem~\ref{thm:SR-fpt} immediately follows from Proposition~\ref{prop:MBDCF} and Lemma~\ref{lem:SR-to-MBDCF-reduction} below.

%The cost $c$ on $\tilde{E}$ is defined by 
%\[
%c(e)
%=
%\begin{cases}
%+\infty & \text{if $e\in \bigcup_{i=1}^k E^i$},\\
%1  &\text{otherwise.}
%\end{cases}
%\]
%Define $\kappa_{\tilde{G}}$ be the cut function of the directed graph $\tilde{G}$. 

\begin{lemma}
\label{lem:SR-to-MBDCF-reduction}
There exists a solution for~$\InstSR$ if and only if there exists a solution for the instance~$(I_{\textup{MBC}},F)$ of \MBDCF{}.
%contain all edges with cost~$+\infty$.
%If $X \subseteq \tilde{V} \setminus \{s,t\}$ is  a solution
%for the instance $(I_{\textup{MB}},F)$ of \MBDCF{}, then $\{v:v^* \in X\}$ is  a solution for our instance~$\InstSR$ of \SR{}. 
%Conversely, if there exists a solution for~$\InstSR$, then there is a solution for~$(I_{\textup{MB}},F)$ as well.
\end{lemma}

\begin{proof}
Suppose that $(I_{\textup{MBC}},F)$ admits a solution.
That is, there exists a subset $X$ of $\tilde{V}$ containing~$s$ but not~$t$ such that $\delta_{\tilde{G}}(X)$ is disjoint from~$F$ and satisfies $|\delta_{\tilde{G}}(X)\cap A^i| \leq d_i$ for each $i \in [k]$. 
%Then $\kappa_{\tilde{G}}(X)$ is finite.
%%Define $Y_i =\{u: u^i\in X\cap U^i\}$ for $i=1,\dots, k$.
Define $Y^i =X\cap V^i$ for $i=1,\dots, k$.
Observe that all edges within~$G^i$ leaving $Y^i \cup \{s\}$ also leave $X$ in $\tilde{G}$, 
since $s \in X$ but $t \notin X$.
Since all edges in~$E^i$ are forbidden edges, we see that $\lambda_i (Y^i)=0$.
%the finiteness of~$\kappa_{\tilde{G}}(X)$ therefore implies the finiteness of~$\kappa_i(Y^i \cup \{s\})$.
%Thus we have $\kappa_i(Y^i \cup \{s\})=0$.
Let $Y_i=\{v\in V:v^i \in Y^i\}$, so that $Y^i$ contains the copy of each vertex of~$Y_i$ in~$G^i$.
Then $Y_i$ is in~$\mathcal{L}_i$ by Lemma~\ref{lem:minimizers}.

Define the subset $X^\ast=\{v:v^* \in X\}$ of~$V$.
Observe that 
\[
\delta_{\tilde{G}}(X) \cap A^i=\{v^\ast v^i: v \in X^\ast, v \notin Y_i\} \cup \{v^i v^*: v \notin X^\ast, v \in Y_i\}. 
\]
Therefore, we get that $|X^\ast \symmdiff Y_i|= |\delta_{\tilde{G}}(X) \cap A^i| \leq d_i=d$ for each $i \in [k]$ as required, so $X^\ast$ is a solution to our instance~$\InstSR$ of \SR{}.

Conversely, 
let $X \subseteq V$ and $Y_i \in \mathcal{L}_i$ for each $i \in [k]$ such that $|X \symmdiff Y_i| \leq d$.
 Define $X^\ast =\{v^\ast\in V^\ast : v\in X\}$ and $Y^i = \{v^i\in V^i: v\in Y_i\}$.
 Then the set $\tilde{X}=\{s \} \cup X^\ast \cup \bigcup_{i=1}^k Y^i$ is an $(s,t)$-cut of $\tilde{G}$ such that
\begin{align*}
\delta_{\tilde{G}} (\tilde{X})\cap A^i
&=
\{v^* v^i:v^*\in X^\ast, v\not\in Y^i\} \cup
\{v^i v^*:v^*\not\in X^\ast, v^i\in Y^i\} 
\\ 
&=
\{v^* v^i:v\in X, v\not\in Y_i\} \cup
\{v^i v^*:v\not\in X, v\in Y_i\} = X \symmdiff Y_i.
\end{align*}
Hence we obtain
 $|\delta_{\tilde{G}}(\tilde{X})\cap A^i| =|X\symmdiff Y_i| \leq d=d_i$ for each $i \in [k]$.
Since $Y_i$ is in~$\mathcal{L}_i$,  by Lemma~\ref{lem:minimizers} we know $\lambda_i (Y^i)=0$ for each $i \in [k]$.
%$\kappa_i (Y^i \cup \{s\})=0$ for each $i$.
Thus $\delta_{\tilde{G}}(\tilde{X})$ is disjoint from the set~$F$ of forbidden edges, and therefore $\tilde{X}$ is indeed a solution to our instance~$(I_{\textup{MBC}},F)$ of \MBDCF{}. 
\end{proof}

\section{$\NP$-hardness for constant values of~$k$ or~$d$}
\label{sec:hardness}

In Section~\ref{sec:NPhard-d1} we show that \SR{} is $\NP$-hard if $d$ is a constant at least~$1$, 
and in Section~\ref{sec:NPhard-k3} we prove that it is $\NP$-hard if $k$ is a constant at least~$3$.

\subsection{\texorpdfstring{$\NP$-hardness for a constant $d\geq 1$}{NP-hardness for d=1}}
\label{sec:NPhard-d1}

In this section, we prove that \SR{} is $\NP$-hard for each constant~$d \geq 1$. 
We first introduce a separation problem that we will use as an intermediary problem in our hardness proofs. Given a subset $X \subseteq V$ of some universe~$V$ that contains two distinguished elements, $s$ and~$t$, and a family~$\Pi$ of pairwise disjoint subsets of~$V$, we define the 
\emph{distance} of the set~$X$ from~$\Pi$
as $\sum_{S \in \Pi} \dist(X,S)$ where
\[ 
\dist(X,S)= \left\{
 \begin{array}{ll}
     \min\{|S \setminus X|,|S \cap X|\} & \textrm{ if } s \notin S, t \notin S; \\
    |S \setminus X| & \textrm{ if } s \in S,t \notin S;\\
    |S \cap X| & \textrm{ if } s \notin S,t \in S;\\
    +\infty  & \textrm{ if } s \in S,t \in S.\\
 \end{array}
\right.
\]
 Given a collection of set families~$\Pi_1,\dots,\Pi_k$, the goal is to find a set~$X \subseteq V$ that separates~$s$ from~$t$ in the sense that $s \in X$ but $t \notin X$, and subject to this constraint, minimizes the maximum distance of~$X$ from the given set families. Formally,  the problem is:

\begin{center}
\fbox{ 
\parbox{13.0cm}{
\begin{tabular}{l}\centralset{}:  \end{tabular} \\
\begin{tabular}{p{1cm}p{11.0cm}}
Input: & A finite set~$V$ with two elements $s,t \in V$, set families $\Pi_1,  \dots, \Pi_k$ where each~$\Pi_i$ is a collection of pairwise disjoint subsets of~$V$, and an integer~$d \in \N$. \\
Task: & Find a set $X \subseteq V$ containing~$s$ but not~$t$ such that for each $i \in [k]$ 
\begin{equation}
\label{eqn:centralset-def}
    \sum_{S\in \Pi_i} \dist(X,S) \leq d,
\end{equation}
or output ``No'' if no such set~$X$ exists.
\end{tabular}
}}
\end{center}

%We are given a ground set $V$ and $k$ families of pairwise disjoint subsets $\Pi_1, \Pi_2, \dots, \Pi_k$ with $k$ nonnegative integers $d_1,\dots, d_k$. The goal is to find a subset $X\subseteq V$ such that, for any $i\in[k]$, 

Given an instance $(V,s,t,\Pi_1\dots,\Pi_k,d)$ of \centralset{}, the reduction proving Lemma~\ref{lem:CStoSR-reduction} below constructs a graph~$G_i$ over~$V$ for each~$i \in [k]$ in which each set in~$\Pi_i$ forms a clique, and defines a submodular function~$f_i$ based on the cut function of~$G_i$.

\begin{lemma}
\label{lem:CStoSR-reduction}
\centralset{} can be reduced to \SR{} in polynomial time
via a reduction 
that preserves the values of both~$k$ and~$d$.
\end{lemma}

\begin{proof}
    Let us be given an instance $I_{\textup{RSep}}=(V,s,t,\Pi_1,\dots,\Pi_k, d)$ of \centralset{}. Clearly, we may assume that no set family~$\Pi_i$ contains a set~$S$ that contains both~$s$ and~$t$, as that would imply $\dist(X,S)=+\infty$ for all sets~$X \subseteq V$, in which case we can clearly answer ``No''.
    
    Let $\tilde{V}=V \setminus \{s,t\}$.
    For each $i \in [k]$ we construct a directed graph $G_i=(V, E_i)$ from $\Pi_i$ by setting
    \[
    E_i = 
    %\{vs: v\in V\}\cup \{tv: t\in V\} \cup 
    \{uv: \exists S\in \Pi_i, u, v\in S\}.
    \]
    We set the cost of every edge in~$E_i$ to be $+\infty$. 
    Observe that an $(s,t)$-cut $X \subseteq V$  is
    a minimum $(s,t)$-cut in~$G_i$ if and only if  for each $S \in \Pi_i$ either $X$ contains~$S$, or is disjoint from~$S$. 
    Let $\kappa_i$ denote the (weighted) cut function for~$G_i$, and define $f_i:2^{\tilde{V}} \rightarrow \R$ such that $f_i(X)=\kappa_i(X \cup\{s\})$ for each $X \subseteq \tilde{V}$.
    %as its restriction to $(s,t)$-cuts. 
    Then $f_i$ is submodular.

\smallskip
    We claim that for each set~$X \subseteq \tilde{V}$, the set~$X \cup \{s\}$ is a solution for our instance~$I_{\textup{RSep}}$ of \centralset{} if and only if $X$ is a solution for the instance $I_{\textup{RSubMin}}=(V,f_1,\dots,f_k,d)$ of \SR{}.

\smallskip    
    To see this, first assume that $X\subseteq V$ contains~$s$ but not~$t$, and fulfills (\ref{eqn:centralset-def}) for each $i \in [k]$.     For an element $v \in V$, let $\Pi_i(v)$ be the unique set in $\Pi_i$ containing~$v$, if there is such a set, otherwise set $\Pi_i(v)=\emptyset$. By our assumption that no set in~$\Pi_i$ contains both~$s$ and~$t$, we know that $\Pi_i(s) \neq \Pi_i(t)$ unless both of them are empty.
    Define 
    \begin{align*}    
    Y_i &= X \cup \Pi_i(s) \setminus \Pi_i(t)
    \setminus 
    \left( \bigcup \left\{S \in \Pi_i: |S\cap X| < |S \setminus X|, s \notin S\right\}\right)
    \\
    & \qquad \cup 
    \left( \bigcup \left\{S \in \Pi_i: |S \cap X| \geq |S \setminus X|,t \notin S\right\}\right) .
    \end{align*}
    Notice that when constructing $Y_i$, each set in~$\Pi_i$ is either fully added to~$Y_i$, or fully removed from it. Moreover, it is an $(s,t)$-cut, since no set containing~$s$ is ever removed from~$X$, and similarly, no set containing~$t$ is ever added to~$X$.
    Therefore, $Y_i$ is a minimum $(s,t)$-cut in~$G_i$ by the observations above. 
    
    Since both $X$ and~$Y_i$ are $(s,t)$-cuts, let us define $\tilde{X}=X \setminus \{s\}$ and $\tilde{Y}_i=Y_i \setminus \{s\}$; since $Y_i$ is a minimum $(s,t)$-cut, we have $\tilde{Y}_i \in \arg\min f_i$.
    Moreover, 
    \begin{align*}
    |\tilde{X} & \symmdiff \tilde{Y}_i|=|X \symmdiff Y_i|=|X \setminus Y_i|+|Y_i \setminus X|  \\
    &=  |\Pi_i(t) \cap X| 
    + \left|\bigcup \{S \cap X : S \in \Pi_i, |S\cap X| < |S \setminus X|,s \notin S,t \notin S\}\right| \\
    & \quad \phantom{a} +|\Pi_i(s) \setminus X| + 
    \left|\bigcup \{S \setminus X: S \in \Pi_i, |S \cap X| \geq |S \setminus X|,s \notin S, t \notin S\}\right| \\
    &= \dist(X,\Pi_i(s))+\dist(X,\Pi_i(t))
    +\sum_{S \in \Pi_i,s \notin S,t \notin S} \min \{|S \setminus X|,|S \cap X|\} \\
    &= \sum_{S \in \Pi_i} \dist(X,S) \leq d
    \end{align*}
    where the last inequality holds because $X$ satisfies~(\ref{eqn:centralset-def}).
    Hence, $\tilde{X}$ is indeed a solution for~$I_{\textup{RSubMin}}$.

    \smallskip
    For the other direction, let $X \subseteq \tilde{V}$ be a solution for~$I_{\textup{RSubMin}}$. Then for each $i \in [k]$ there exists a set~$Y_i \in \arg\min f_i$ that satisfies $|X \symmdiff Y_i| \leq d$; define $Y_i$ among all such sets so that $|X \symmdiff Y_i|$ is minimized. Recall that $Y_i \in \arg\min f_i$ means that $Y_i \cup \{s\}$ is a minimum $(s,t)$-cut in~$G_i$; let us use the notation  $Y'_i=Y_i \cup \{s\}$.
    Note that if some element~$v \in \tilde{V}$ is not contained in any set of~$\Pi_i$, then adding it to, or removing it from, a minimum $(s,t)$-cut in~$G_i$ also yields a minimum $(s,t)$-cut.
    This implies that $Y_i \cap (V \setminus \bigcup_{S \in \Pi_i}S )=X \cap (V \setminus \bigcup \Pi_i)$,
    as otherwise we could replace $Y_i \cap (V \setminus \bigcup_{S \in \Pi_i}S ))$ with $X \cap (V \setminus \bigcup_{S \in \Pi_i}S ))$ in~$Y_i$ to obtain another minimizer for~$f_i$ that is closer to~$X$ than~$Y_i$; we refer to this as fact~$(f1)$.
    
    Recall also that  since $Y'_i$ is a minimum $(s,t)$-cut in~$G_i$, for each $S \in \Pi_i$ it either contains~$S$ or is disjoint from it. 
    We refer to this observation as fact~$(f2)$.

    Let us define $X'=X \cup \{s\}$.
    We now prove that $X'$ is a solution for the instance~$I_{\textup{RSep}}$.
    Applying facts ($f1$) and~($f2$), for each $i \in [k]$ we get 
    {
    \allowdisplaybreaks
    \begin{align*}
        d &\geq 
        |X \symmdiff Y_i| 
        \stackrel{(f1)}{=} \left|\{v: v \in S, S \in \Pi_i, v \in X \symmdiff Y_i\}\right| \\
        & \!\!
        \stackrel{(f2)}{=} 
        \!\begin{multlined}[t][11.85cm]
        \left|\{v: v \in S, S \in \Pi_i, S \subseteq Y'_i, v \in X \symmdiff Y_i\}\right| 
        \\
        + 
        \left|\{v: v \in S, S \in \Pi_i, S \cap Y'_i=\emptyset, v \in X \symmdiff Y_i\}\right| 
        \end{multlined}
        \\
        &= 
        \!\begin{multlined}[t][12.0cm]
        \left|\{v: v \in S, S \in \Pi_i, S \subseteq Y'_i, v \in (S\cap Y_i) \setminus X \}\right| 
        \\
        + \left|\{v: v \in S, S \in \Pi_i, S \cap Y'_i=\emptyset, v \in S \cap X\}\right|     
        \end{multlined}
        \\
        &= 
        \begin{multlined}[t][12.0cm]
        \left|(\Pi_i(s) \setminus \{s\}) \setminus X \right|  
        +\left|\bigcup \{S \setminus X : s\notin S \in \Pi_i, S \subseteq Y_i\}\right| 
        \\
         + \left|\bigcup \{S \cap X : S \in \Pi_i, S \cap Y'_i=\emptyset\}\right|  
         \end{multlined}
        \\
        &= 
        \!\begin{multlined}[t][12.0cm]
        \left|\Pi_i(s)  \setminus X' \right| 
        + \left|\bigcup \{S \setminus X' : s \notin S \in \Pi_i, S \subseteq Y_i\}\right|
        \\
        + \left|\bigcup \{S \cap X' : S \in \Pi_i, S \cap Y'_i=\emptyset\}\right|     
        \end{multlined}
        \\
        &\geq 
        \dist(X',\Pi_i(s))+\sum_{S \in \Pi_i,s \notin S,t \notin S} \min\{|S \setminus X|,|S \cap X|\}+\dist(X',\Pi_i(t))
        \\
        & = \sum_{S \in \Pi_i} \dist(X',S).
    \end{align*}
}
    Hence, $X'$ is indeed a solution for the instance~$I_{\textup{RSep}}$ of \centralset{}, as promised.    
    This proves the correctness of our reduction.
\end{proof}

\begin{remark}
\label{rem:U}
Note the reduction presented in Lemma~\ref{lem:CStoSR-reduction} has the following property: given an instance~$I=(V,s,t,\Pi_1,\dots,\Pi_k,d)$ of
\centralset{} where $\Pi_i=\{K,V \setminus K\}$ for some set~$K \subseteq V$ containing~$s$ but not~$t$, then \new{for} the reduced instance
$I'=(V \setminus \{s,t\},f_1,\dots,f_k,d)$ of \SR{} it \new{holds that} $|\arg \min f_i|=1$.
\end{remark}

To prove that \SR{} is $\NP$-hard for each constant $d \geq 1$, we first prove the $\NP$-hardness of \centralset{} in the case $d=1$, and then extend this result to hold for any constant~$d \geq 1$. 

For the case $d=1$, we present a reduction from the \textsc{1-in-3 SAT} problem.
In this problem, we are given a set~$V$ of variables and a set $\mathcal{C}$ of clauses, with each clause $C \in \mathcal{C}$ containing exactly three distinct literals;  here, a \emph{literal} is either a variable~$v \in V$ or its negation~$\ol{v}$.
Given a truth assignment $\phi:V \rightarrow \{\tt,\ff\}$, 
we automatically extend it to the set $\ol{V}=\{\ol{v}:v \in V\}$ of negative literals by setting $\phi(\ol{v})=\tt$ if and only if $\phi(v)=\ff$.
We say that a truth assignment is \emph{valid}, if it maps \emph{exactly} one literal in each clause to \tt.
The task in the \textsc{1-in-3 SAT} problem is to decide whether a valid truth assignment exists.
This problem is $\NP$-complete~\cite{Schaefer78}.

\begin{theorem}
\label{thm:CS-NPhard-d-equal1}
\centralset{} is $\NP$-hard even when $d=1$.
\end{theorem}

\begin{proof}
Suppose that we are given an instance of the \textsc{1-in-3 SAT} problem with variable set~$V$ and clause set~$\CC=\{C_1\dots,C_m\}$. 
We construct an instance~$\InstCS$ of \centralset{} as follows.
In addition to the set~$V$ of variables and the set~$\ol{V}=\{\ol{v}:v \in V\}$ of negative literals, we introduce our two distinguished elements, $s$ and~$t$.
We further introduce a set $R_j=\{r_{j,1},r_{j,2},r_{j,3}\}$ together with an extra element $z_j$ for each clause~$C_j \in \CC$ to form our universe~$U$. We let $R=R_1 \cup \dots \cup R_m$ and $Z=\{z_1,\dots,z_m\}$, so that
\[
U=V \cup \ol{V} \cup \{s,t\} \cup \bigcup_{j \in [m]} (R_j \cup \{z_j\})=V \cup \ol{V} \cup \{s,t\} \cup R \cup Z .
\]

Next, for each variable, we introduce two set families, $\Pi_v$ and $\Pi_{\ol{v}}$, where
\[
\Pi_v=\{\{s,v,\ol{v}\} \cup R\} 
\qquad \textrm{ and } \qquad 
\Pi_{\ol{v}}=\{\{v,\ol{v},t\}\}.
\]
For simplicity, we write $\Pi(V) = \langle \Pi_v,\Pi_{\ol{v}}: v \in V \rangle$ to denote the $2|V|$-tuple formed by these set families.
For each clause~$C_j \in \CC$, we fix an arbitrary ordering of its literals, and we denote the first, second, and third  literals in $C_j$ as $\ell_{j,1},\ell_{j,2}$ and $\ell_{j,3}$. We define three set families (see also Figure~\ref{fig:hadrd1}):
\[
\begin{array}{r@{\hspace{2pt}}lcl@{\hspace{2pt}}l}
\Pi_{C_j} &= \{S_j\} & 
\textrm{ where } &  S_j&=C_j \cup \{t\}= \{ \ell_{j,1},\ell_{j,2},\ell_{j,3},t\}, \\[2pt]
\Pi_{C_j}^\alpha & =\{S_j^{\alpha,1},S_j^{\alpha,2}\} & \textrm{ where } & S_j^{\alpha,1} &= \{ \ell_{j,1},z_j\}, \\[2pt]
& & &  S_j^{\alpha,2} &= \{\ell_{j,2},r_{j,2}\};
\\[2pt]
\Pi_{C_j}^\beta &  =\{S_j^{\beta,1},S_j^{\beta,2}\} & \textrm{ where } & S_j^{\beta,1} &= \{ r_{j,1},z_j\}, \\[2pt]
& & &  S_j^{\beta,2} &=\{\ell_{j,3},r_{j,3} \}. 
\end{array}
\]
We also write $\Pi(\CC) = \langle \Pi_C,\Pi_C^\alpha,\Pi_{C}^\beta: C \in \CC \rangle$ to denote the $3|\CC|$-tuple formed by these set families in an arbitrarily fixed ordering. 
We set our threshold as $d=1$.
Thus, our instance of \centralset{} is $\InstCS=(U,s,t,\Pi(V), \Pi(\CC),1)$.

\begin{figure}[tbp]
    \begin{center}
    \includegraphics[scale=0.73, pagebox=cropbox,clip]{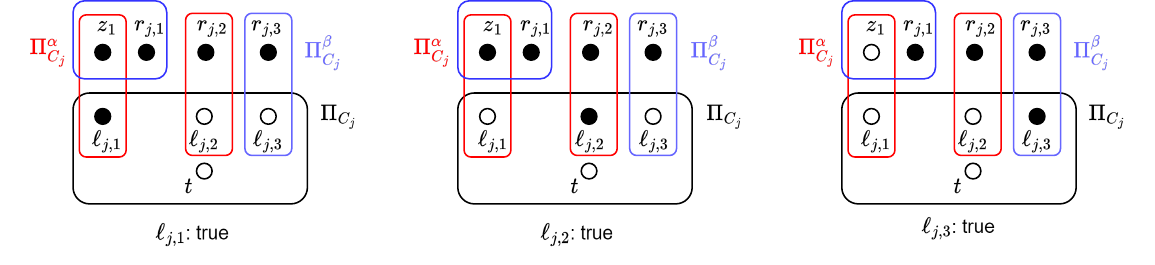}
    \caption{The set families~$\Pi_{C_j}$, $\Pi_{C_j}^\alpha$, and $\Pi_{C_j}^\beta$ defined for some clause~$C_j$. The three figures depict the construction of the set~$X$ depending on which of the literals takes $\tt$ value; 
     elements of the set~$X$ are shown as black circles.}
    \label{fig:hadrd1}
    \end{center}
\end{figure}

We will show that the constructed instance~$\InstCS$ has a solution if and only if our instance~$(V,\CC)$ of the \textsc{1-in-3 SAT} problem
is solvable.

First suppose that there is a valid truth assignment $\phi$ for $(V,\CC)$. 
Consider the set \[X=\{s\}  \cup R \cup \{ \ell:\ell \in V \cup \ol{V},\phi(\ell)=\tt \} \cup \{z_j: z_j \in Z,\phi(\ell_{j,3})=\ff\}.\]
Note that $X$ contains~$s$, but not~$t$ (see again Figure~\ref{fig:hadrd1} for an illustration).
We are going to show that $X$ is a solution for~$\InstCS$. 
Since $\phi$ maps exactly one literal in $\{v,\ol{v}\}$ to \tt{} for each $v \in V$, by $R \cup \{s\} \subseteq X$ we get that 
\begin{align*}
\sum_{S \in \Pi_v}\dist(X,S)&=
|(\{s,v,\ol{v}\} \cup R) \setminus X| =
|\{v,\ol{v}\}\setminus X|= 1 \qquad \textrm{ and} \\
\sum_{S \in \Pi_{\ol{v}}}\dist(X,S)&=
|(\{v,\ol{v},t\}) \cap X| =
|\{v,\ol{v}\}\cap X|= 1.
\end{align*}
%Here, the last inequality follows from the facts that $|C\setminus T| + |C\cap T| = 3$ and $|C\cap T|\in \{1, 2\}$.
For the distance of~$X$ from the set families associated with some clause $C_j \in \CC$, by the validity of~$\phi$ we obtain
{\allowdisplaybreaks
\begin{align*}
\sum_{S \in \Pi_{C_j}}\dist(X,S)&=
|(C_j \cup \{t\}) \cap X| = 1;
\\
%%%%%%%
\sum_{S \in \Pi_{C_j}^\alpha}\dist(X,S) & =
\min\{|S_j^{\alpha,1} \setminus X|,|S_j^{\alpha,1} \cap X|\}
+
\min\{|S_j^{\alpha,2} \setminus X|,|S_j^{\alpha,2} \cap X|\}
\\
&= 
\!\begin{multlined}[t][9cm]
\min \{ |\{ \ell_{j,1},z_j\} \setminus X|,|\{ \ell_{j,1},z_j\} \cap X| \}
\\
+ 
\min \{ | \{\ell_{j,2},r_{j,2}\} \setminus X|,|\{\ell_{j,2},r_{j,2}\} \cap X| \}
\end{multlined}
\\
& = \left\{ 
\begin{array}{ll}
   \min\{0,2\} +\min\{1,1\}=1 
   & \textrm{ if } \phi(\ell_{j,1})=\tt \\
   \min\{1,1\} +\min\{0,2\}=1 
   & \textrm{ if } \phi(\ell_{j,2})=\tt \\
   \min\{2,0\} +\min\{1,1\}=1 
   & \textrm{ if } \phi(\ell_{j,3})=\tt 
\end{array}
\right\} = 1;
%%%%%%%%%%
\\
\sum_{S \in \Pi_{C_j}^\beta}\dist(X,S) & =
\min\{|S_j^{\beta,1} \setminus X|,|S_j^{\beta,1} \cap X|\}
+
\min\{|S_j^{\beta,2} \setminus X|,|S_j^{\beta,2} \cap X|\}
\\
&= 
\!\begin{multlined}[t][9cm]
\min \{ |\{ r_{j,1},z_j\} \setminus X|,|\{ r_{j,1},z_j\} \cap X| \}
\\
+ 
\min \{ | \{\ell_{j,3},r_{j,3}\} \setminus X|,|\{\ell_{j,3},r_{j,3}\} \cap X| \}
\end{multlined}
\\
& = \left\{ 
\begin{array}{ll}
   \min\{0,2\} +\min\{1,1\}=1 
   & \textrm{ if } \phi(\ell_{j,1})=\tt \\
   \min\{0,2\} +\min\{1,1\}=1 
   & \textrm{ if } \phi(\ell_{j,2})=\tt \\
   \min\{1,1\} +\min\{0,2\}=1 
   & \textrm{ if } \phi(\ell_{j,3})=\tt 
\end{array}
\right\} = 1.
\end{align*}
}

Hence, $X$ satisfies constraint~(\ref{eqn:centralset-def}) for each set family, and thus is a solution for $I_{\textup{RS}}$.

For the other direction, suppose that there exists a subset $X\subseteq U$, containing~$s$ but not~$t$, whose distance from each set family in~$I_{\textup{RS}}$ is at most~$1$.
Let us define a truth assignment $\phi$ by setting $\phi(v)=\tt$ if and only if $v \in X$ for each variable~$v \in V$; we are going to show that $\phi$ is valid.

To this end, let us first observe that 
\begin{align}
\label{eqn:containsR}
1=d & \geq \sum_{S \in \Pi_v} \dist(X,S) =|(\{s,v,\ol{v}\} \cup R) \setminus X| \quad \textrm{ and} \\    
\notag
1=d & \geq \sum_{S \in \Pi_{\ol{v}}} \dist(X,S) =|\{v,\ol{v},t\} \cap X|. 
\end{align}
This implies that \begin{equation}
\label{eqn:consistency}
|\{v,\ol{v}\} \setminus X| \leq 1
\qquad \textrm{ and } \qquad
|\{v,\ol{v}\} \cap X| \leq 1.
\end{equation}
By $|\{v,\ol{v}\} \setminus X|+|\{v,\ol{v}\} \cap X|=|\{v,\ol{v}\}|=2$, we obtain that all inequalities in~(\ref{eqn:consistency}), and hence also in~(\ref{eqn:containsR}), must hold with equality. That is, 
for each variable~$v \in V$, exactly one of the literals~$v$ and~$\ol{v}$ is contained in~$X$. 
In other words, a literal~$\ell \in  V \cup \ol{V}$ is set to \tt{} by~$\phi$ if and only if $\ell \in X$.
Moreover, 
we must also have $R \subseteq X$ due to the equality in~(\ref{eqn:containsR}).

Consider now the set family~$\Pi_C$ for some $C \in \CC$, and recall that $t \notin X$:
\[
1 \geq d=\sum_{S \in \Pi_C}\dist(X,S)=\dist(X,S_j)=
|(C\cup \{t\})\cap X| = |C \cap X|.
\]
This means that $\phi$ sets at most one literal to \tt{} in each clause~$C \in \CC$. 
To prove that $\phi$ is valid, let us assume for the sake of contradiction that $C_j \cap X=\emptyset$ for some~$C_j \in \CC$. 
In this case we have $\{\ell_{j,1},\ell_{j,2},\ell_{j,3}\}\cap X=\emptyset$; keep in mind also that $R \subseteq X$. 
On the one hand, if $z_j \in X$, then we get
\begin{align*}
1 & \geq 
\sum_{S \in \Pi_{C_j}^\alpha}\dist(X,S)  =
\min\{|S_j^{\alpha,1} \setminus X|,|S_j^{\alpha,1} \cap X|\}
+
\min\{|S_j^{\alpha,2} \setminus X|,|S_j^{\alpha,2} \cap X|\}
\\
&= 
%\!\begin{multlined}[t][9cm]
\min \{ |\{ \ell_{j,1},z_j\} \setminus X|,|\{ \ell_{j,1},z_j\} \cap X| \}
%\\
+ 
\min \{ | \{\ell_{j,2},r_{j,2}\} \setminus X|,|\{\ell_{j,2},r_{j,2}\} \cap X| \}
%\end{multlined}
\\
& = \min\{1,1\} +\min\{1,1\}=2,
\end{align*}
a contradiction. On the other hand, if $z_j \notin X$, then we get
\begin{align*}
1 & \geq 
\sum_{S \in \Pi_{C_j}^\beta}\dist(X,S)  =
\min\{|S_j^{\beta,1} \setminus X|,|S_j^{\beta,1} \cap X|\}
+
\min\{|S_j^{\beta,2} \setminus X|,|S_j^{\beta,2} \cap X|\}
\\
&= 
%\!\begin{multlined}[t][9cm]
\min \{ |\{ r_{j,1},z_j\} \setminus X|,|\{ r_{j,1},z_j\} \cap X| \}
%\\
+ 
\min \{ | \{\ell_{j,3},r_{j,3}\} \setminus X|,|\{\ell_{j,3},r_{j,3}\} \cap X| \}
%\end{multlined}
\\
& = 
   \min\{1,1\} +\min\{1,1\}=2,
\end{align*}
which is again a contradiction. This proves that $\phi$ is indeed a valid truth assignment for our instance of \textsc{1-in-3 SAT}, and so the claim holds.
\end{proof}

Using Theorem~\ref{thm:CS-NPhard-d-equal1}, it is not hard to show that \centralset{} remains $\NP$-hard for any constant $d \geq 1$.

\begin{lemma}
\label{lem:CS-NPhard-d-atleast1}
    \centralset{} is $\NP$-hard for each constant $d \geq 1$.
\end{lemma}

\begin{proof}
    We present a reduction from \centralset{} with $d=1$ which is $\NP$-hard by Theorem~\ref{thm:CS-NPhard-d-equal1}. 
    Let $I=(V,s,t,\Pi_1,\dots,\Pi_k,1)$ be our input instance. 
    We compute a modified instance $I'$ of \centralset{} with threshold~$d \geq 1$ as follows.
    First, we introduce a set~$D$ of dummy $2d$ items and add them to our universe. Next, we add two new set families, each containing a single set: 
    namely, we let 
    $\Pi_{k+1}=\{\{D \cup \{s\} \}\}$ and $\Pi_{k+2}=\{\{D \cup \{t\} \}\}$.
    Furthermore, let $\hat{D}$ be an arbitrarily fixed subset of~$D$ of size~$2d-1$. 
    For each $i \in [k]$ we define $\Pi_i'=\Pi_i \cup \{\hat{D}\}$, that is, we add the set~$\hat{D}$ to $\Pi_i$. 
    Let $I'=(V \cup D, s,t,\Pi'_1,\dots,\Pi'_k,\Pi_{k+1},\Pi_{k+2},d)$ be our constructed instance.

    We claim $I$ admits a solution if and only if $I'$ admits one.

    First, let $X$ be a solution for~$I$. 
    Let us fix some $D' \subseteq \hat{D}$ with~$|D'|=d$, and define $X'=X \cup D'$. Then for each $i \in [k]$ and 
    $S \in \Pi_i$, since $S \cap D=\emptyset$ 
    we get $S \setminus X'=S \setminus X$ 
    and $S \cap X'=S \cap X$; this yields 
    $\dist(X',S)=\dist(X,S)$.
     Hence, for each $i \in [k]$ we obtain
    \begin{align*}
    \sum_{S \in \Pi'_i}\dist(X',S) 
    &=
    \sum_{S \in \Pi_i}\dist(X',S)+
    \dist(X',\hat{D})
    \\
    &= \sum_{S \in \Pi_i}\dist(X,S)+
    \dist(X',\hat{D})\\
    & \leq 1+\min \{|\hat{D} \setminus X'|,|\hat{D} \cap X'|\}=1+\min \{d-1,d\}=d
    \end{align*}
    where the inequality holds by our assumption on~$X$. 
    Moreover, since $s \in X \subseteq X'$, the distance of $X'$ from~$\Pi_{k+1}$ is 
    \[\sum_{S \in \Pi_{k+1}} \dist(X',S)=\dist(X',D \cup \{s\})=|(D \cup \{s\}) \setminus X'|=|D \setminus D'|=d
    \]
    and similarly, since $t \notin X$ implies $t \notin X'$, the distance of $X'$ from~$\Pi_{k+2}$ is 
    \[\sum_{S \in \Pi_{k+2}} \dist(X',S)=\dist(X',D \cup \{t\})=|(D \cup \{t\}) \cap X'|=|D \cap D'| = d.
    \]
    Hence, $X'$ is indeed a solution for~$I'$.

    For the other direction, let $X'$ be a solution for~$I'$. Since $X'$ contains~$s$ but not~$t$, we have 
    \begin{align*}
    d\geq \sum_{S \in \Pi_{k+1}} \dist(X',S) &=\dist(X',D \cup \{s\})=|(D \cup \{s\}) \setminus X'| = |D \setminus X'|, \\ %\qquad
%    \textrm{and} \\
    d\geq \sum_{S \in \Pi_{k+2}} \dist(X',S) &=\dist(X',D \cup \{t\})=|(D \cup \{t\}) \cap X'|=|D \cap X'|.
    \end{align*}
    However, using $|D \setminus X'|+|D \cap X'|=|D|=2d$, we obtain that both of the above inequalities hold with equality. Therefore, $X'$ contains exactly~$d$ dummy elements from~$D$. Recall that $|\hat{D}|=2d-1$.
    This implies that $|\hat{D} \cap X'| \in \{d-1,d\}$ 
    and $|\hat{D} \setminus X'| \in \{d-1,d\}$. From this, it follows that $\dist(X',\hat{D})=\min \{ |\hat{D} \setminus X'|,|\hat{D} \cap X'|\} \geq d-1$. 
    We obtain that for each $i \in [k]$
    \[ 
    d \geq \sum_{S \in \Pi'_i}\dist(X',S)= \sum_{S \in \Pi_i}\dist(X',S)+\dist(X',\hat{D})\geq 
    \sum_{S \in \Pi_i}\dist(X',S)+d-1
    \] which implies that the distance of $X'$ from $\Pi_i$ is $\sum_{S \in \Pi_i}\dist(X',S) \leq 1$. 
    
    Let $X =X' \cap V$. 
    Again, recall that for each $i \in [k]$ and $S \in \Pi_i$ we know $S \subseteq V$, since $S$ contains no dummies. Thus $S \cap X' =S \cap X$, and 
    $S \setminus X' = S \setminus X$, from which $\dist(X',S)=\dist(X,S)$ follows. Hence, we obtain that $\sum_{S \in \Pi_i}\dist(X,S) \leq 1$ holds for each~$i \in [k]$, and so $X$ is a solution for the original instance~$I$.
\end{proof}

\begin{corollary}
\label{cor:SR-NPhard-d-atleast1}
\SR{} is $\NP$-hard for each constant $d \geq 1$.
\end{corollary}

\subsection{\texorpdfstring{$\NP$-hardness for a constant $k\geq 3$}{NP-hardness for k=3}}
\label{sec:NPhard-k3}

In this section we prove that
%\centralset{}, and hence, 
\SR{} is $\NP$-hard even for $k=3$. 
To this end, we are going to define another intermediary problem. 
First consider the \MBMC{} problem, proved to be $\NP$-complete by Bonsma~\cite{Bonsma10}. The input of this problem is an undirected graph~$G=(V,E)$ with two distinguished vertices, $s$ and~$t$, and a parameter~$\ell$. The task is to decide whether there exists a minimum $(s,t)$-cut $X \subseteq V$ in~$G$ such that $\min\{|X|,|V \setminus X|\} \geq \ell$; recall that a set of vertices~$X \subseteq V$ is a minimum $(s,t)$-cut in the \emph{undirected} graph~$G$ if $s \in X, t \notin X$ and subject to this, the value $|\delta(X)|$, i.e., the number of edges between $X$ and $V \setminus X$, is minimized.

Instead of the \MBMC{} problem, it will be more convenient to use a variant that we call \BMC{} 
where we seek a minimum $(s,t)$-cut that contains exactly half of the vertices. Formally, its input is an undirected graph~$G=(V,E)$ with two distinguished vertices, $s$ and~$t$, and its task is to find a minimum $(s,t)$-cut~$X$ with $|X|=|V|/2$. Since \MBMC{} can be reduced to \BMC{} by simply adding a sufficient number of isolated vertices, we obtain the following.

\begin{lemma}
    \label{lem:balancedcut}
    \BMC{} is $\NP$-complete.
\end{lemma}

\begin{proof}
    We give a reduction from \MBMC{}. 
    Let $I=(G,s,t,\ell)$ with $G=(V,E)$ be our input instance. 
    We may assume $\ell \leq |V|/2$, as otherwise $I$ is clearly a ``no''-instance. 
    We construct an instance $I'=(G',s,t)$ of \BMC{} by simply adding $|V|-2\ell$ isolated vertices to~$G$, so that the number of vertices in~$G'$ is $2|V|-2\ell$. 
    Clearly, if there is a minimum $(s,t)$-cut~$X$ such that $\min\{|X|,|V|-|X|\} \geq \ell$, then adding $|V|-\ell-|X| \geq 0$ vertices yields a minimum $(s,t)$-cut in~$G'$ that contains exactly half of the vertices of~$G'$.
    Conversely, if $X'$ is a minimum $(s,t)$-cut in~$G'$ that contains half of the vertices, i.e., $|V|-\ell$ vertices, then removing the newly added isolated vertices from it yields a minimum $(s,t)$-cut~$X$ in~$G$ with $|V|-\ell -(|V|-2\ell)=\ell \leq |X| \leq |V|-\ell$.
\end{proof}

\begin{theorem}
\label{thm:k=3-almostunique}
\SR{} is $\NP$-hard for $k=3$, even on instances of the form~$I=(V,f_1,f_2,f_3,d)$ where $|\LL_1|=|\LL_2|=1$.
\end{theorem}
\begin{proof}
We present a reduction from the $\NP$-hard \BMC{} problem 
(see Lemma~\ref{lem:balancedcut}). 
Let $(G,s,t)$ be our input with undirected graph~$G=(V \cup \{s,t\},E)$; we may assume that $|V|$ is even. Let $n=|V|/2$.

We construct a set $R$ of size~$n$, and define the domain of our submodular functions as~$V \cup R$. Define $Y_1=V \cup R$ and $Y_2=R$, 
and let $f_1$ and~$f_2$ denote two submodular set functions over~$V \cup R$ whose only minimizers are $Y_1$ and~$Y_2$, respectively.
That is, $\arg\min f_1=\{Y_1\}$ and $\arg \min f_2=\{Y_2\}$. 
To define~$f_3$, we construct a graph~$G'$ by adding $R$ to the vertex set of~$G$, and connecting each~$r \in R$ to~$t$ by an edge. Let  $f_3:2^{V \cup R} \rightarrow \mathbb{N}$ be defined by $f_3(X)=|\delta_{G'}(X \cup \{s\})|$ for each $X \subseteq V$. As we have mentioned in~Section~\ref{sec:prelim}, $f_3$ is  submodular~\cite{schrijver-book}.
We finish the construction of our instance~$I=(V \cup R,f_1,f_2,f_3,d)$ of \SR{} by setting~$d=n$ as our threshold.

Suppose first that $X$ is a minimum $(s,t)$-cut in~$G$ that contains half of the vertices, i.e., $n+1$ vertices, including~$s$. We claim that $X' \cup R$ is a solution for~$I$ where $X' =X \setminus \{s\}$. Note that $X' \subseteq V$ and $|X'|=n$ and recall $|V|=2n$. 
Clearly, $|X' \symmdiff Y_1|=|(X' \cup R) \symmdiff (V \cup R)|= |X' \symmdiff V|=|V \setminus X'|=n$.
Moreover, $|X' \symmdiff Y_2|=|(X' \cup R) \symmdiff R|= |X'|=n$.
Note also that since a minimum $(s,t)$-cut in~$G$ is also a minimum $(s,t)$-cut in~$G'$, as every minimum $(s,t)$-cut in~$G$ is disjoint from~$R$.
Thus, $X' \in \arg\min f_3$, and $|(X' \cup R) \symmdiff X'|=|R|=n$. This shows that $X' \cup R$ is indeed a solution for~$I$.

Suppose now that $X$ is a solution for~$I$. Observe that if $v \in V=Y_1 \symmdiff Y_2$, then depending on whether $X$ contains~$v$ or not, $v$ is either in~$Y_2 \symmdiff X$ or in~$Y_1 \symmdiff X$. Therefore,
\begin{align*}
2n &=|V|=|Y_1 \symmdiff Y_2| \leq \left|
\left( (Y_1 \symmdiff X) \cap V \right)
\cup 
\left( (Y_2 \symmdiff X) \cap V \right) \right| 
 \\
 & \leq |(Y_1 \symmdiff X) \cap V|+ |(Y_2 \symmdiff X) \cap V| \leq |Y_1 \symmdiff X|+ |Y_2 \symmdiff X| \leq n+n=2n,
\end{align*}
and thus all of the above inequalities hold with equality. 
In particular, we get $Y_1 \symmdiff X \subseteq V$ and $Y_2 \symmdiff X \subseteq V$, which implies that $R \subseteq X$. This yields 
\begin{equation}
\label{eqn:cutsize}    
n=|X_2 \symmdiff X|=|R \symmdiff X|=|X \setminus R|=|X \cap V|.
\end{equation}
Consider now the minimizers of~$f_3$. Since $X$ is a solution, there exists some~$Y_3 \in \arg\min f_3$ with $|Y_3 \symmdiff X| \leq n$. Then $Y_3 \cup \{s\}$ is a minimum $(s,t)$-cut in~$G'$. Recall that all minimum $(s,t)$-cuts of~$G'$ are disjoint from~$R$ due to the construction of~$G'$, so $Y_3 \cap R=\emptyset$.
However, then $R \subseteq Y_3 \symmdiff X$, which implies 
\[n=|R| \leq |Y_3 \symmdiff X| \leq d=n.\]
Thus, it must hold that $R=Y_3 \symmdiff X$. 
This means that $X \cap V=Y_3 \cap V=Y_3$.

Recall that $Y_3 \cup \{s\}$ is a minimum $(s,t)$-cut in~$G'$ disjoint from~$R$, and thus it is also a minimum $(s,t)$-cut in~$G$. Furthermore, $|Y_3|=|X \cap V|=n$ by (\ref{eqn:cutsize}), so $Y_3 \cup \{s\}$ contains half of the vertices in~$G$. This proves that $Y_3 \cup \{s\}$ is a solution for our instance $(G,s,t)$ of \BMC{}, and thus the  reduction is correct.
\end{proof}

Clearly, we can increase the value of parameter~$k$ without changing the solution set of our instance by repeatedly adding a copy of, say, the first submodular function~$f_1$. Hence, we get the following consequence of Theorem~\ref{thm:k=3-almostunique}.

\begin{corollary}
\label{cor:k-atleast-3}
\SR{} is NP-hard for each constant $k\geq 3$, even if $|\LL_i|=1$ for each $i \in [k-1]$.
\end{corollary}

Note that Theorem~\ref{thm:k=3-almostunique}  is in sharp contrast with Theorem~\ref{thm:k=2}, which offers a polynomial-time algorithm for the case $k=2$.
Furthermore, Theorem~\ref{thm:k=3-almostunique} is tight also in the sense that when \emph{all} input functions have only a bounded number of minimizers, then the problem becomes FPT with parameter~$k$, as we will show in Proposition~\ref{prop:few-minimizers}.

To close this section, we remark that the \centralset{} problem can also be shown to be~$\NP$-hard for $k=3$; see the conference version of this paper for a proof~\cite{KS-SWAT}.

\section{Polynomially many minimizers}\label{sec:polynomialminimizers}
In this section we consider those variants of \SR{} when one or more among the submodular functions given as input have only a bounded number of minimizers.

When \emph{all} submodular functions in the input have unique minimizers, then the problem becomes equivalent with the \textsc{Closest String} problem. 
This problem is $\NP$-complete~\cite{DBLP:journals/mst/FrancesL97} (see also~\cite{KakimuraKKO22}), but fixed-parameter tractable with parameter~$k$, due to a result based on Integer Linear Programming by Gramm et al.~\cite{GrammNR03}. In the context of \SR{}, we obtain the following.
\begin{observation}
\label{obs:all-unique}
If $|\mathcal{L}_i|=1$ for each $i \in [k]$, then
\SR{} is $\NP$-hard and fixed-parameter tractable with respect to parameter~$k$.
\end{observation}

\begin{proof}
If $|\mathcal{L}_i|=1$ for each $i \in [k]$, then there is a unique minimizer~$Y_i \subseteq V$  for each $f_i$, and the problem is equivalent with finding a set~$X \subseteq V$ whose symmetric difference is at most~$d$ from  each of the sets~$Y_i$, $i \in [k]$. This is the \textsc{Closest String} problem over a binary alphabet, shown to be $\NP$-hard under the name \textsc{Minimum Radius} by Frances and Litman~\cite{DBLP:journals/mst/FrancesL97}. Fixed-parameter tractability for parameter~$k$ follows from the FPT algorithm for the \textsc{Closest String} problem due to Gramm et al~\cite{GrammNR03}.
\end{proof}

The FPT algorithm with parameter~$k$ in Observation~\ref{obs:all-unique} can be generalized to the case when \new{we} have a larger bound on the number $|\LL_i|$ of minimizers for each function~$f_i$.

\begin{proposition}
\label{prop:few-minimizers}
    An instance~$I=(V,f_1,\dots,f_k,d)$ of \SR{} can be solved in $|\LL_1|\cdot |\LL_2| \cdot \dots \cdot |\LL_k|\cdot  g(k)|I|^{O(1)}$ for some computable function~$g$.
    In particular: 
    \begin{itemize}
        \item If there is an integer~$\ell$ such that $|\LL_i| \leq \ell$ for each $i \in [k]$, then \SR{} is FPT with parameter $(\ell,k)$.
        \item If there is an integer~$c$ such that $|\LL_i| \leq O(|V|^c)$ for each $i \in [k]$, then \SR{} can be solved in polynomial time for each constant~$k$, i.e., it is in $\mathsf{XP}$ with parameter~$k$.
    \end{itemize}
\end{proposition}

\begin{proof}
It suffices to observe that the brute force approach of trying all possible minimizers $Y_i \in \LL_i$ for each $i \in [k]$ yields $|\LL_1|\cdot |\LL_2| \cdot \dots \cdot |\LL_k|$ possibilities, and for each possible choice of the minimizers $Y_1,\dots,Y_k$ we need to solve an instance of (binary) \textsc{Closest String}. Since the latter task can be carried out in FPT time with parameter~$k$~\cite{GrammNR03}, we obtain the claimed running time. 
\end{proof}

It should be noted that the case when $|\LL_i|$ is bounded for %any 
\new{every} $i\in [n]$ includes \SR{} when each submodular function is the cut function for an undirected graph~\cite{dinitz1976structure}.
In fact, it is known that for \new{every} undirected graph~$G$ with $n$ vertices the number of minimum cuts is bounded by $O(n^2)$~\cite{dinitz1976structure}.

Inspecting Proposition~\ref{prop:few-minimizers}, a natural question arises: assuming that each function has only polynomially many minimizers, i.e., $|\LL_i|=O(|V|^c)$ for some constant~$c$, is \SR{} solvable in FPT time when parameterized by~$k$? Or is the brute force approach behind Proposition~\ref{prop:few-minimizers} unavoidable?
In Section~\ref{sec:W1hardness}, we show 
in Theorem~\ref{thm:W1hard-k} that, even under a strong assumption on the number of minimizers for each~$f_i$, the problem is $\mathsf{W}[1]$-hard when parameterized by~$k$, which implies that no significant speed-up seems possible. 

We remark that although Proposition~\ref{prop:few-minimizers} proves that for fixed values of~$k$ \SR{} becomes polynomial-time solvable under the assumption that each input function has only a bounded number of minimizers, 
the algorithm behind Proposition~\ref{prop:few-minimizers} seems impractical: this is not only due to the combinatorial explosion in the running time of its brute force approach, but also to the fact that although \textsc{Closest String} is fixed-parameter tractable with parameter~$k$, denoting the number of strings in that context, the FPT algorithm given by Gramm et al.~\cite{GrammNR03} is based on Lenstra's famous result~\cite{Lenstra83} on Integer Linear Programs with fixed dimension, and has been reported to be impractical already for small values of~$k$ (see  \cite{GrammNR03,BoucherBD,AmirPR16} for a series of efforts to obtain efficient combinatorial algorithms for a fixed number of strings).

Contrasting these feasibility issues concerning the $\mathsf{XP}$ algorithm in Proposition~\ref{prop:few-minimizers} and the  $\mathsf{W}[1]$-hardness result of Theorem~\ref{thm:W1hard-k} for parameter~$k$, 
in Section~\ref{sec:alg-unique} we will show
that under the assumption that at least one of the input functions has only a bounded number of minimizers, \SR{} becomes fixed-parameter tractable with parameter~$d$ (Theorem~\ref{thm:unique}).

\subsection{Polynomially many minimizers for each function: $\mathsf{W}[1]$-hardness for~$k$}
\label{sec:W1hardness}

In this section we show the following result.

\begin{theorem}
\label{thm:W1hard-k}
\SR{} is $\mathsf{W}[1]$-hard with parameter~$k$ even on instances $(V,f_1,\dots,f_k,d)$ where $|\LL_i|\leq|V|$.
\end{theorem}

\begin{proof}
We are going to present a parameterized reduction from the $\mathsf{W}[1]$-hard \textsc{Multicolored Clique} problem~\cite{FellowsHRV-multicolored-hardness}.
In \textsc{Multicolored Clique}, the input is a graph $G=(V,E)$ and a parameter~$k \in \mathbb{N}$, with the vertex set of~$G$ partitioned into independent sets~$V_1, \dots, V_k$, and the task is to find a clique of size~$k$ in~$G$ that (necessarily) contains a vertex from each set~$V_i$. 
For indices $i,j \in [k]$ with $i \neq j$, let $E_{\{i,j\}}=\{uv \in E: u \in V_i,v \in V_j\}$.
We may assume without lowering the order of minimization that $|V_i|=n$ for each $i \in [k]$, and $|E_{\{i,j\}}|=m$ for each $i,j \in [k]$ with $i \neq j$, for some integers~$n$ and~$m$.\footnote{We may ensure this property by first adding the necessary number of  edges as a matching on newly introduced vertices into each edge set~$E_{\{i,j\}}$, and then further adding the necessary number of isolated vertices to each set~$V_i$.}
Hence, we write $V_i=\{v_{i,1},\dots,v_{i,n}\}$ for each $i \in [k]$, and we write $E_{\{i,j\}}=\{e_{\{i,j\},1},\dots,e_{\{i,j\},m} \}$.
For some $v \in V_i$ and $j \in [k] \setminus \{i\}$, we will also use the notation $E_j(v)=\{uv \in E: u \in V_j\}$.

To construct our universe~$U$, we add $e^+$ and $e^-$ for each edge~$e \in E$, and similarly, we add $v^+$ and $v^-$ for each vertex~$v \in V$. 
For any subset~$X \subseteq V$ of vertices, we set $X^+=\{v^+:v \in X\}$ and $X^-=\{v^-:v \in X\}$, and we define $F^+$ and $F^-$ for each edge set~$F \subseteq E$ analogously.
We further add four disjoint sets~$R, R', D$ and~$D'$ to~$U$ of sizes $|R|=|R'|=|D|=|D'|=|V|+|E|$, so that 
\[U= V^+ \cup V^-
 \cup E^+ \cup E^- \cup R \cup R' \cup D \cup D'.
\]
%We will use the notation $E^-$ ...
We set $d=2(|V|+|E|)$ as our threshold. 
We are going to construct several gadgets:
\begin{itemize}
    \item an \emph{initial gadget},
    \item $k$ \emph{vertex selection gadgets}, 
    \item $\binom{k}{2}$ \emph{edge selection gadgets},  and
    \item $k(k-1)$ \emph{incidency gadgets},    
\end{itemize} each of which comprise submodular functions given through a compact representation.
We fix an arbitrary ordering over~$R$ and~$R'$, and we denote by $R[\ell]$ the size~$\ell$ subset of~$R$
containing the first $\ell$ elements according to this ordering for some $\ell \in \mathbb{N}$; we define $R'[\ell]$ analogously. 
%We also let $\ol{R[\ell]}=R \setminus R[\ell]$ and $\ol{R'[\ell]}=R \setminus R'[\ell]$. 
Furthermore, we set $k'=k+\binom{k}{2}$.

\subparagraph{Initial gadget.} 
This gadget contains functions $f^{\init}_1, f^{\init}_2, \dots, f^{\init}_5$, 
each of which has a unique minimizer. We define $f^{\init}_h$ for $h \in [5]$ by defining their unique minimizers $Y^{\init}_h$ as follows:
\begin{align*}
    Y^\init_1 &=  R \cup V^+ \cup V^- \cup E^+ \cup E^- \cup D \cup D', \\
    %U \setminus Y^\init_2 &=  R' \cup V^+ \cup V^- \cup E^+ \cup E^- \cup D \cup D', \\
    Y^\init_2 &=  R, \\
    Y^\init_3 &=  R \cup R' \cup V^+ \cup V^- \cup E^+ \cup E^- \cup D, \\
    %U \setminus Y^\init_4 &= \{t\} \cup V^+ \cup V^- \cup E^+ \cup E^- \cup D',\\
    Y^\init_4 &=  R \cup R' \cup D,\\
    Y^\init_5 &=  R[k'] \cup (R' \setminus R'[k']) \cup V^- \cup E^- \cup D. \\
    %Y^\init_6 &=  R'[2k'] \cup V^+ \cup E^+ \cup D.
\end{align*}

\subparagraph{Vertex selection gadget for some $i \in [k]$.} 
This gadget contains functions~$f^\vsel_{i,1}$ and~$f^\vsel_{i,2}$ defined through their unique minimizers~$Y^\vsel_{i,1}$ and~$Y^\vsel_{i,2}$, as well as two functions $f^\vsel_{i,3}$ and $f^\vsel_{i,4}$. These latter two functions are defined via their compact representations, $G^\vsel_{i,3}$ and~$G^\vsel_{i,4}$.
The graph $G^\vsel_{i,3}$ (or $G^\vsel_{i,4}$) over a partition of~$U \cup \{s,t\}$ will be a simple path $P^\vsel_i$ (or $Q^\vsel_i$, respectively).
For simplicity, given a partition $(U_0,U_1,\dots,U_b,U_\infty)$ of $U \cup \{s,t\}$ with $s \in U_0$ and $t \in U_\infty$, we denote a path visiting these nodes in the reverse order as $U_0 \leftarrow U_1 \leftarrow \dots \leftarrow U_b \leftarrow U_\infty$; recall that if this path is a  compact representation of some submodular function, then its minimizers are exactly the sets of the form $U_0 \cup U_1 \cup \dots \cup U_j$ for some $0 \leq j \leq b$.
\begin{align*}
    Y^\vsel_{i,1} &=  R'[d/2-n-2k'+2] \cup V^+_i  \cup V^- \cup E^- \cup D, \\
    Y^\vsel_{i,2} &=  R'[d/2-n-2k'+2]  \cup V^- \setminus V_i^- \cup E^- \cup D, \\
    P^\vsel_i &= \{s\} \cup R'[d/2-n-2k'+3] \cup V^- \setminus V_i^- \cup E^- \cup D \cup \{v_{i,1}^+\} 
            \\
            & \qquad 
            \leftarrow \{v_{i,2}^+,v_{i,1}^-\}  
            \leftarrow \{v_{i,3}^+,v_{i,2}^-\} \leftarrow  
            \dots
            \leftarrow \{v_{i,n}^+,v_{i,n-1}^-\} 
         \\ & \qquad 
            \leftarrow \{v_{i,n}^-, t\} \cup R \cup R' \setminus R'[d/2-n-2k'+3] \cup V^+ \setminus V_i^+ \cup E^+ \cup D', \\
    Q^\vsel_i &= \{s\} \cup R'[d/2-n-2k'+3] \cup V^- \setminus V_i^- \cup E^- \cup D \cup \{v_{i,n}^+\} 
            \\
            & \qquad 
            \leftarrow \{v_{i,n-1}^+,v_{i,n}^-\} 
            \leftarrow \{v_{i,n-2}^+,v_{i,n-1}^-\} \leftarrow  
            \dots
            \leftarrow \{v_{i,1}^+,v_{i,2}^-\} 
            \\ & \qquad 
            \leftarrow \{v_{i,1}^-,t\} \cup R \cup R' \setminus R'[d/2-n-2k'+3] \cup V^+ \setminus V_i^+ \cup E^+ \cup D'.
\end{align*}

\subparagraph{Edge selection gadget for some $i,j \in [k]$ with $i < j$.} 
This gadget is analogous to the vertex selection gadgets, and contains functions~$f^\vsel_{\{i,j\},1}$ and~$f^\vsel_{\{i,j\},2}$ defined through their unique minimizers~$Y^\vsel_{\{i,j\},1}$ and~$Y^\vsel_{\{i,j\},2}$, as well as two functions $f^\vsel_{i,3}$ and $f^\vsel_{i,4}$. These latter two functions are defined via their compact representations, $G^\vsel_{\{i,j\},3}$ and~$G^\vsel_{\{i,j\},4}$.
The graph $G^\vsel_{\{i,j\},3}$ (or $G^\vsel_{\{i,j\},4}$) over a partition of~$U \cup \{s,t\}$ will be a simple path $P^\esel_{\{i,j\}}$ (or $Q^\esel_{\{i,j\}}$, respectively).
\begin{align*}
    Y^\esel_{\{i,j\},1} &=  R'[d/2-m-2k'+2] \cup E^+_{\{i,j\}}  \cup V^- \cup E^- \cup D, \\
    Y^\esel_{\{i,j\},2} &=  R'[d/2-m-2k'+2]  \cup V^- \cup E^- \setminus E_{\{i,j\}}^- \cup D, \\
    P^\esel_{\{i,j\}} &= \{s\} \cup R'[d/2-m-2k'+3] \cup V^- \cup E^- \setminus E_{\{i,j\}}^- \cup D  \cap \{e_{{\{i,j\}},1}^+\} 
             \\
            & \qquad 
            \leftarrow \{e_{{\{i,j\}},2}^+,e_{{\{i,j\}},1}^-\}  
            \leftarrow \{e_{{\{i,j\}},3}^+,e_{{\{i,j\}},2}^-\}  
            \leftarrow 
            \dots
            \leftarrow \{e_{{\{i,j\}},m}^+,e_{{\{i,j\}},m-1}^-\} 
            \\ & \qquad 
            \leftarrow \{e_{{\{i,j\}},m}^-,t\} \cup R \cup R' \setminus R'[d/2-m-2k'+3] \cup V^+ \cup E^+ \setminus E_{\{i,j\}}^+ \cup D', \\
    Q^\esel_{\{i,j\}} &= \{s\} \cup R'[d/2-m-2k'+3] \cup V^- \cup E^- \setminus E_{\{i,j\}}^- \cup D  \cup \{e_{\{i,j\},m}^+\} 
             \\
            & \qquad \leftarrow \{e_{\{i,j\},m-1}^+,e_{\{i,j\},m}^-\} 
            \leftarrow \{e_{\{i,j\},m-2}^+,e_{\{i,j\},m-1}^-\} \leftarrow  
            \dots 
            %\\ & \qquad 
            \leftarrow \{e_{\{i,j\},1}^+,e_{\{i,j\},2}^-\} 
             \\
            & \qquad 
            \leftarrow \{e_{\{i,j\},1}^-,t\} \cup R \cup R' \setminus R'[d/2-m-2k'+3] \cup V^+ \cup E^+ \setminus E_{\{i,j\}}^+ \cup D'.
\end{align*}

\subparagraph{Incidency gadget for some $i,j \in [k]$ with $i \neq j$.} 
This gadget uses a similar idea as the vertex and edge selection gadgets, but contains only a single submodular function~$f^\inc_{(i,j)}$, given via its compact representation $G^\inc_{(i,j)}$. The graph $G^\inc_{(i,j)}$ over a partition of $U \cup \{s,t\}$ will be a simple path $P^\inc_{(i,j)}$.

\begin{align*}
    \!\!\!\! P^\inc_{(i,j)} &= \{s\} \cup R'[d/2-m-n-2k'+6] \cup V^- \setminus V^-_i \cup E^- \setminus E^-_{\{i,j\}} \cup D \cup \{v_{i,1}^+\} \cup E^+_j(v_{i,1}) \\
    & \qquad
        \leftarrow \{v_{i,2}^+,v_{i,1}^-\} \cup E^+_j(v_{i,2}) \cup E^-_j(v_{i,1}) 
        \leftarrow \{v_{i,3}^+,v_{i,2}^-\} \cup E^+_j(v_{i,3}) \cup E^-_j(v_{i,2}) \leftarrow \dots \\
    & \qquad 
        \leftarrow \{v_{i,n}^+,v_{i,n-1}^-\} \cup E^+_j(v_{i,n})  \cup E^-_j(v_{i,n-1}) 
        \\
    & \qquad 
        \!\begin{multlined}[t][12cm]
        \leftarrow \{v_{i,n}^-,t\} \cup E^-_j(v_{i,n})  \cup R \cup R' \setminus R'[d/2-m-n-2k'+6] \\
        \cup V^+ \setminus V^+_i \cup E^+ \setminus E_{\{i,j\}}^+ \cup D'.\end{multlined}
\end{align*}

This finishes the construction of our instance~$I$ of \SR{}. Note that the number of submodular functions given in~$I$ is $5+4k+4\binom{k}{2}+k(k-1)$, which is a function of $k$. Therefore, the presented reduction is a parameterized reduction.
Observe also that each of the submodular functions constructed has less than~$|U|$ minimizers.

We claim that~$I$ admits a solution if and only if $G$ contains a clique of size~$k$.

\subparagraph{Direction ``$\Longrightarrow$'':}
Assume first that $I$ admits a solution~$X \subseteq U$.
As $X$ is a solution, we know $|X \symmdiff Y^\init_h| \leq d$ for each $h \in [6]$. In particular, since this holds for $h=1,2$ we get that 
\begin{align*}
2d & \geq |X \symmdiff Y^\init_1| + |X \symmdiff Y^\init_2| \geq |Y^\init_1 \symmdiff Y^\init_2|=|V^+ \cup V^- \cup E^+ \cup E^- \cup D \cup D'| \\
& = 4(|V|+|E|)=2d.
\end{align*}
Hence, all inequalities above must hold with equality. Hence, $X \symmdiff Y^\init_h \subseteq Y^\init_1 \symmdiff Y^\init_2$ for $h=1,2$, which implies that 
\begin{equation}
\label{eqn:solX-vs-RandR'}
    R \subseteq X \qquad \qquad \textrm{and} \qquad \qquad X \cap R' =\emptyset
\end{equation} 
Consider now $Y^\init_3$.
Taking into account~(\ref{eqn:solX-vs-RandR'}), we get 
{\allowdisplaybreaks
\begin{align*}
d & \geq |X \symmdiff Y^\init_3| = |X \symmdiff (R \cup R' \cup V^+ \cup V^- \cup E^+ \cup E^- \cup D)|= \\
& = \!\begin{multlined}[t][12.0cm]
|X \setminus (R \cup R' \cup V^+ \cup V^- \cup E^+ \cup E^- \cup D)| \\
+ 
|(R \cup R' \cup V^+ \cup V^- \cup E^+ \cup E^- \cup D) \setminus X| = 
\end{multlined}\\
&= |X \cap D'| + |R'| + |(V^+ \cup V^- \cup E^+ \cup E^-) \setminus X| + |D \setminus X| \\
& \geq 0+d/2+|(V^+ \cup V^- \cup E^+ \cup E^-) \setminus X|+0,
\end{align*}
}
which implies 
\begin{equation}
\label{eqn:init3}
    |(V^+ \cup V^- \cup E^+ \cup E^-) \setminus X| \leq d/2.
\end{equation} 
Arguing similarly for $Y^\init_4$ we have
\begin{align}
\notag d & \geq |X \symmdiff Y^\init_4| = |X \symmdiff (R \cup R' \cup D)|= \\
\notag  &= |X \setminus (R \cup R' \cup D)|+ |(R \cup R' \cup D) \setminus X| \\
\notag &= |X \cap D'|+|X \cap (V^+ \cup V^- \cup E^+ \cup E^-)| + |R'| + |D \setminus X|  \\
& \geq 0+|X \cap (V^+ \cup V^- \cup E^+ \cup E^-)|+d/2+0, 
\label{eqn:solX-vs-DcupD'}
\end{align}
which implies 
\begin{equation}
\label{eqn:init4}    
 |(V^+ \cup V^- \cup E^+ \cup E^-) \cap X| \leq d/2.
\end{equation}
However, as $|V^+ \cup V^- \cup E^+ \cup E^-|=2d$, both (\ref{eqn:init3}) and~(\ref{eqn:init4}) must hold with equality, which  in turn means that (\ref{eqn:solX-vs-DcupD'}) must also hold with equality. This means that 
\begin{equation}
\label{eqn:solX-vs-DandD'}
    D \subseteq X \qquad \qquad \textrm{and} \qquad \qquad X \cap D' =\emptyset.
\end{equation} 

Consider now the function~$f^\init_5$ and its unique minimizer~$Y^\init_5$. 
Using~(\ref{eqn:solX-vs-RandR'}) and~(\ref{eqn:solX-vs-DandD'}) we obtain
\begin{align*}
    d & \geq |X \symmdiff Y^\init_5|=|X \symmdiff (R[k'] \cup R' \setminus R'[k'] \cup V^- \cup E^- \cup D)| \\
    & = \!\begin{multlined}[t][12cm]
    |X \setminus (R[k'] \cup R' \setminus R'[k'] \cup V^- \cup E^-  \cup D)| \\
    +|(R[k'] \cup R' \setminus R'[k'] \cup V^- \cup E^-  \cup D) \setminus X| 
    \end{multlined}    \\
    & = |R \setminus R[k']|+ |X \cap (V^+ \cup E^+)| + 
    |R' \setminus R'[k']|+|(V^- \cup E^-) \setminus X| \\    
    &= d/2-k'+ |X \cap (V^+ \cup E^+)| + d/2-k'+|(V^- \cup E^-) \setminus X|, 
\end{align*}
which implies
\begin{equation}
\label{eqn:k1+k2}
    |X \cap (V^+ \cup E^+)| + |(V^- \cup E^-) \setminus X| \leq 2k'.
\end{equation} 
Recall now that both (\ref{eqn:init3}) and~(\ref{eqn:init4}) hold with equality: 
\begin{align*}
    d/2 &=|(V^+ \cup E^+ \cup V^- \cup E^+) \setminus X|=\\
     &= |V^+ \cup E^+| - |(V^+ \cup E^+) \cap X|+
    |(V^- \cup E^-) \setminus X| \\
    & = d/2 - |(V^+ \cup E^+) \cap X| +|(V^- \cup E^-) \setminus X|; \\
    d/2 &=|(V^+ \cup E^+ \cup V^- \cup E^+) \cap X|\\
    &= |(V^+ \cup E^+) \cap X| + 
    |V^- \cup E^-| - |(V^- \cup E^-) \setminus X| \\
    & = d/2 + |(V^+ \cup E^+) \cap X| - |(V^- \cup E^-) \setminus X|.
\end{align*}
Hence, from (\ref{eqn:k1+k2}) we obtain that
\begin{equation}
\label{eqn:bound-k'}
 |(V^+ \cup E^+) \cap X|=|(V^- \cup E^-) \setminus X|=k'.
\end{equation} 

Next, we prove a series of claims that show how $X$ intersects~$V^- \cup E^- \cup V^+ \cup E^+$.
\begin{claim}
$|X \cap V_i^+| \geq 1$  for each $i \in [k]$.
\label{clm:select-Vi+}        
\end{claim}
\begin{claimproof}
Consider some vertex selection gadget, and the function~$f^\vsel_{i,1}$ for some~$i \in [k]$ with its unique minimizer~$Y^\vsel_{i,1}$. 
Keep in mind that $R \cup D \subseteq X$ but both~$R'$ and~$D'$ are disjoint from~$X$ due to~(\ref{eqn:solX-vs-RandR'}) and~(\ref{eqn:solX-vs-DandD'}). 
Since~$X$ is a solution, we know
\begin{align*}
    d & \geq |X \symmdiff Y^\vsel_{i,1}| = |X \symmdiff (R'[d/2-n-2k'+2] \cup V^+_i  \cup V^- \cup E^- \cup D)| \\
    & = \!\begin{multlined}[t][12cm]
    |X \setminus (R'[d/2-n-2k'+2] \cup V^+_i  \cup V^- \cup E^- \cup D)| \\
    + |(R'[d/2-n-2k'+2] \cup V^+_i  \cup V^- \cup E^- \cup D) \setminus X|  
    \end{multlined} \\
    & = \!\begin{multlined}[t][12cm] 
    |R| + |X \cap (V^+ \setminus V_i^+ \cup E^+)| \\
    + |R'[d/2-n-2k'+2]| + |V^+_i \setminus X| + |(V^- \cup E^- ) \setminus X|
    \end{multlined} \\
    & = \!\begin{multlined}[t][12cm] 
    d/2 + |X \cap (V^+ \cup E^+)| - |X \cap V^+_i|  \\
    + (d/2-n-2k'+2) + n-|V_i^+ \cap X| + |(V^- \cup E^- ) \setminus X|
    \end{multlined} \\
    &= d/2 +k' - |X \cap V_i^+| + (d/2-n-2k'+2)+ n- |X \cap V^+_i| +k'
    \\
    &= d+2-2|X \cap V_i^+|
\end{align*}
where the penultimate equality follows from~(\ref{eqn:bound-k'}), proving our claim.
\end{claimproof}

\begin{claim}
$|V_i^- \setminus X| \geq 1$  for each $i \in [k]$.
\label{clm:select-Vi-}    
\end{claim}
\begin{claimproof}
Using for~$f^\vsel_{i,2}$ the same reasoning we used in Claim~\ref{clm:select-Vi+}
 yields
\begin{align*}
    d & \geq |X \symmdiff Y^\vsel_{i,2}| = |X \symmdiff (R'[d/2-n-2k'+2]  \cup V^- \setminus V_i^- \cup E^- \cup D)| \\
    & = \!\begin{multlined}[t][12cm]
    |X \setminus (R'[d/2-n-2k'+2]  \cup V^- \setminus V_i^- \cup E^- \cup D)| \\
    + |(R'[d/2-n-2k'+2]  \cup V^- \setminus V_i^- \cup E^- \cup D) \setminus X|  
    \end{multlined} \\
    & = \!\begin{multlined}[t][12cm] 
    |R| + |X \cap (V^+ \cup E^+)| + |X \cap V_i^-| \\
    + |R'[d/2-n-2k'+2]| + |(V^- \cup E^- ) \setminus X| - |V^-_i \setminus X|
    \end{multlined} \\
    & = \!\begin{multlined}[t][12cm] 
    d/2 + |X \cap (V^+ \cup E^+)| +n - |V^-_i \setminus X| \\
    + (d/2-n-2k'+2) + |(V^- \cup E^- ) \setminus X| - |V^-_i \setminus X| 
    \end{multlined} \\
    &= d/2 +k' +n- |V^-_i \setminus X| + (d/2-n-2k'+2)+k' - |V^-_i \setminus X|
    \\
    &= d+2-2|V^-_i \setminus X|.
\end{align*}
This proves the claim.
\end{claimproof}

\begin{claim}
$|X \cap E^+_{\{i,j\}}| \geq 1$ for each $i,j \in [k]$ with $i<j$.
\label{clm:select-Eij+}    
\end{claim}
\begin{claimproof}
We can argue analogously for the function~$f^\esel_{\{i,j\},1}$, contained in an edge selection gadget, as in the proof of Claim~\ref{clm:select-Vi+}, using again (\ref{eqn:solX-vs-RandR'}), (\ref{eqn:solX-vs-DandD'}), and~(\ref{eqn:bound-k'}):
{\allowdisplaybreaks
\begin{align*}
    d & \geq  |X \symmdiff Y^\esel_{\{i,j\},1}| =|X \symmdiff (R'[d/2-m-2k'+2] \cup E^+_{\{i,j\}}  \cup V^- \cup E^- \cup D)| \\
    & = \!\begin{multlined}[t][12cm]
    |X \setminus (R'[d/2-m-2k'+2] \cup E^+_{\{i,j\}}  \cup V^- \cup E^- \cup D)| \\
    + |(R'[d/2-m-2k'+2] \cup E^+_{\{i,j\}}  \cup V^- \cup E^- \cup D) \setminus X|  
    \end{multlined} \\
    & = \!\begin{multlined}[t][12cm] 
    |R| + |X \cap (V^+ \cup E^+ \setminus E^+_{\{i,j\}})| \\
    + |R'[d/2-m-2k'+2]| + | E^+_{\{i,j\}} \setminus X| + |(V^- \cup E^- ) \setminus X|
    \end{multlined} \\
    & = \!\begin{multlined}[t][12cm] 
    d/2 + |X \cap (V^+ \cup E^+)| - |X \cap E^+_{\{i,j\}}|  \\
    + (d/2-m-2k'+2) + m-|E^+_{\{i,j\}} \cap X| + |(V^- \cup E^- ) \setminus X|
    \end{multlined} \\
    &= d/2 +k' - |X \cap E^+_{\{i,j\}}| + (d/2-m-2k'+2)+ m- |X \cap E^+_{\{i,j\}}| +k'
    \\
    &= d+2-2|X \cap E^+_{\{i,j\}}|
\end{align*}
}
This proves the claim.
\end{claimproof}

\begin{claim}
$|E_{\{i,j\}}^- \setminus X| \geq 1$ for each $i,j \in [k]$ with $i<j$.
\label{clm:select-Eij-}        
\end{claim}
\begin{claimproof}
Using for~$f^\esel_{\{i,j\},2}$ the same reasoning 
as in the proof of Claim~\ref{clm:select-Vi-} yields
{\allowdisplaybreaks
\begin{align*}
    d & \geq |X \symmdiff Y^\esel_{\{i,j\},2}| = |X \symmdiff (R'[d/2-m-2k'+2]  \cup V^- \cup E^- \setminus E_{\{i,j\}}^- \cup D)| \\
    & = \!\begin{multlined}[t][12cm]
    |X \setminus (R'[d/2-m-2k'+2]  \cup V^- \cup E^- \setminus E_{\{i,j\}}^- \cup D)| \\
    + |(R'[d/2-m-2k'+2]  \cup V^- \cup E^- \setminus E_{\{i,j\}}^- \cup D) \setminus X|  
    \end{multlined} \\
    & = \!\begin{multlined}[t][12cm] 
    |R| + |X \cap (V^+ \cup E^+)| + |X \cap E_{\{i,j\}}^-| \\
    + |R'[d/2-n-2k'+2]| + |(V^- \cup E^- ) \setminus X| - |E_{\{i,j\}}^- \setminus X|
    \end{multlined} \\
    & = \!\begin{multlined}[t][12cm] 
    d/2 + |X \cap (V^+ \cup E^+)| +m - |E_{\{i,j\}}^- \setminus X|   \\
    + (d/2-n-2k'+2) + |(V^- \cup E^- ) \setminus X| - |E_{\{i,j\}}^- \setminus X| 
    \end{multlined} \\
    &= d/2 +k' +m- |E_{\{i,j\}}^- \setminus X| + (d/2-m-2k'+2)+k' - |E_{\{i,j\}}^- \setminus X|
    \\
    &= d+2-2|E_{\{i,j\}}^- \setminus X|.
\end{align*}
}
This proves the claim. 
\end{claimproof}

Summing up the inequalities in Claim~\ref{clm:select-Vi+} for each $i \in [k]$ and those in Claim~\ref{clm:select-Eij+} for each $i,j \in [k]$ with $i<j$, we get
\[\sum_{i \in [k]}|V_i^+ \cap X|+\sum_{i,j \in [k],i<j} |E_{\{i,j\}}^+ \cap X|= |(V^+ \cup E^+) \cap X| \geq k+\binom{k}{2} = k'.\]
Contrasting this with~(\ref{eqn:bound-k'}), it follows that the inequality above must hold with equality, and consequently,
Claim~\ref{clm:select-Vi+} for each $i \in [k]$ and Claim~\ref{clm:select-Eij+} for each $i,j \in [k]$ with $i<j$ must also hold with equality. 
In other words, for each~$i \in [k]$ there exists some
%$v_{i,\sigma(i)} \in V_i$ 
$\sigma(i) \in [n]$ 
such that 
\begin{equation}
\label{eqn:select-Vi+}
X \cap V_i^+=\{ v_{i,\sigma(i)}^+\},
\end{equation}
and similarly, for each  $i,j \in [k]$ with $i<j$ there exists some
%$e_{\{i,j\},\sigma(\{i,j\})} \in E_{\{i,j\}}$ 
$\sigma(\{i,j\}) \in [m]$
such that 
\begin{equation}
\label{eqn:select-Eij+}
X \cap E_{\{i,j\}}^+=\{ e_{\{i,j\},\sigma(\{i,j\})}^+\}.
\end{equation}

Similarly, summing  up the inequalities in Claim~\ref{clm:select-Vi-}) for each $i \in [k]$ and those in Claim~\ref{clm:select-Eij-}) for each $i,j \in [k]$ with $i<j$, we get
\[\sum_{i \in [k]}|V_i^- \setminus X|+\sum_{i,j \in [k],i<j} |E_{\{i,j\}}^- \setminus X|= |(V^- \cup E^-) \cap X| \geq k+\binom{k}{2} = k'.\]
Again, contrasting this with~(\ref{eqn:bound-k'}), it follows that the inequality above must hold with equality, and consequently,
Claim~\ref{clm:select-Vi-} for each $i \in [k]$ and Claim~(\ref{clm:select-Eij-} for each $i,j \in [k]$ with $i<j$ must also hold with equality. 
In other words, for each~$i \in [k]$ there exists some
%$v_{i,\sigma(i)} \in V_i$ 
$\sigma'(i) \in [n]$ 
such that 
\begin{equation}
\label{eqn:select-Vi-}
V_i^- \setminus X=\{ v_{i,\sigma'(i)}^-\},
\end{equation}
and similarly, for each  $i,j \in [k]$ with $i<j$ there exists some
%$e_{\{i,j\},\sigma(\{i,j\})} \in E_{\{i,j\}}$ 
$\sigma'(\{i,j\}) \in [m]$
such that 
\begin{equation}
\label{eqn:select-Eij-}
 E_{\{i,j\}}^- \setminus X=\{ e_{\{i,j\},\sigma'(\{i,j\})}^-\}.
 \end{equation}

 In the next four claims, we are going to show that $\sigma \equiv \sigma'$.

\begin{claim}
\label{clm:sigma-i-1}
 $\sigma(i) \leq \sigma'(i)$ for each $i \in [k]$.
\end{claim}
\begin{claimproof}
Consider the function $f^\vsel_{i,3}$ and the vertex selection gadget corresponding to~$V_i$. 
Since $X$ is a solution for~$I$, there exists some $Y^\vsel_{i,3} \in \arg\min f^\vsel_{i,3}$ such that $|X \symmdiff Y^\vsel_{i,3} | \leq d$.
Recall that the compact representation of $f^\vsel_{i,3}$ is the path~$P^\vsel_i$, and hence 
Birkhoff's theorem (Theorem~\ref{thm:Birkhoff}) implies that $Y^\vsel_{i,3} \cup \{s\}$ must be the union of  nodes in some \emph{proper postfix} of $P^\vsel_i$, i.e., a subpath of~$P^\vsel_i$ that contains its last node and is not equal to~$P^\vsel_i$.
Due to the definition of~$P^\vsel_i$, this means that there exists some~$p^\star \in [n]$ for which 
\begin{equation}
\label{eqn:def-Y-vi3}
Y^\vsel_{i,3}= R'[d/2-n-2k'+3] \cup V^- \setminus V_i^- \cup E^- \cup D 
\cup \bigcup_{h \leq p^\star} \{v_{i,h}^+\} 
\cup \bigcup_{1<h \leq p^\star} \{v_{i,h-1}^-\}. 
\end{equation}

Thus, we get
\begin{align*}
    d & \geq |X \symmdiff Y^\vsel_{i,3}| = |X \setminus Y^\vsel_{i,3}| + |Y^\vsel_{i,3} \setminus X| \\
    & = \!\begin{multlined}[t][12cm] 
    |X \setminus  (R'[d/2-n-2k'+3] \cup V^- 
    %\setminus V_i^- 
    \cup  E^- \cup D)|
    \\
    - |X \cap \{v_{i,h}^+:h \leq p^\star\}|
    + |X \cap \{v_{i,h-1}^-: h > p^\star\}| 
    \end{multlined} \\
    & \phantom{=i} \begin{multlined}[t][12cm] 
    + |(R'[d/2-n-2k'+3] \cup V^- \cup E^- \cup D) \setminus X| \\
    - |\{v_{i,h-1}^-: h > p^\star\} \setminus X| 
    + |\{v_{i,h}^+:h \leq p^\star\} \setminus X| 
    \end{multlined} \\
    & =  \!\begin{multlined}[t][12cm] 
    |R|+ |X \cap (V^+ \cup E^+)| \\
    - |X \cap \{v_{i,h}^+:h \leq p^\star\}| 
    + |\{v_{i,h'}^-:h'\geq p^\star\}| - |\{v_{i,h'}^-: h' \geq p^\star\} \setminus X| 
    \end{multlined} \\
    & \phantom{=i} \begin{multlined}[t][12cm] 
    + |R'[d/2-n-2k'+3]|+|(V^- \cup E^-) \setminus X|
    \\
    - |\{v^-_{i,h'}:h' \geq p^\star\} \setminus X| 
    + |\{v_{i,h}^+:h \leq p^\star\}| -
    |\{v_{i,h}^+:h \leq p^\star\} \cap X| 
    \end{multlined} \\
    & =  \!\begin{multlined}[t][12cm] 
    d/2+k'-2|\{v_{i,h}^+:h \leq p^\star\} \cap X| 
    + (n+1)\\
    -2 |\{v^-_{i,h'}:h' \geq p^\star\} \setminus X| +(d/2-n-2k'+3)+k' 
    \end{multlined} \\
    & = d + 4-2|\{v_{i,h}^+:h \leq p^\star\} \cap X| -2 |\{v^-_{i,h'}:h' \geq p^\star\} \setminus X| \\
    & = d+4 
    - 2 \cdot \left\{  \begin{array}{ll}
    1, & \textrm{ if }\sigma(i) \leq p^\star \\
    0, & \textrm{ otherwise} 
    \end{array}
    \right\} 
    - 2 \cdot  \left\{  \begin{array}{ll}
    1, & \textrm{ if }\sigma'(i) \geq p^\star \\
    0, & \textrm{ otherwise.} 
    \end{array}
    \right\} 
\end{align*}
From this, it follows that $\sigma(i) \leq p^\star \leq \sigma'(i)$ must hold, proving the claim.
\end{claimproof}

\begin{claim}
\label{clm:sigma-i-2}
 $\sigma(i) \geq \sigma'(i)$ for each $i \in [k]$.
\end{claim}
\begin{claimproof}
Consider now the function $f^\vsel_{i,4}$. 
Since $X$ is a solution for~$I$, there exists some $Y^\vsel_{i,4} \in \arg\min f^\vsel_{i,4}$ such that $|X \symmdiff Y^\vsel_{i,4} | \leq d$.
Recall that the compact representation of $f^\vsel_{i,4}$ is the path~$Q^\vsel_i$, which implies that $Y^\vsel_{i,4} \cup \{s\}$ must be the union of nodes in some proper postfix of $Q^\vsel_i$.
Due to the definition of~$Q^\vsel_i$ and Birkhoff's theorem, this means that there exists some~$q^\star \in [n]$ for which 
\begin{equation}
\label{eqn:def-Y-vi4}
Y^\vsel_{i,4}= R'[d/2-n-2k'+3] \cup V^- \setminus V_i^- \cup E^- \cup D \cup 
\bigcup_{h \geq q^\star} \{v_{i,h}^+\} 
\cup \bigcup_{ n > h \geq q^\star} \{v_{i,h+1}^-\}. 
\end{equation}
Thus, we get
\begin{align*}
    d & \geq |X \symmdiff Y^\vsel_{i,4}| = |X \setminus Y^\vsel_{i,4}| + |Y^\vsel_{i,4} \setminus X| \\
    & = \!\begin{multlined}[t][12cm] 
    |X \setminus  (R'[d/2-n-2k'+3] \cup V^- 
    %\setminus V_i^- 
    \cup  E^- \cup D)|
    \\
    - |X \cap \{v_{i,h}^+:h \geq q^\star\}|
    + |X \cap \{v_{i,h+1}^-: h < q^\star\}| 
    \end{multlined} \\
    & \phantom{=i} \begin{multlined}[t][12cm] 
    + |(R'[d/2-n-2k'+3] \cup V^- \cup E^- \cup D) \setminus X| \\
    - |\{v_{i,h+1}^-: h < q^\star\} \setminus X| 
    + |\{v_{i,h}^+:h \geq q^\star\} \setminus X| 
    \end{multlined} \\
    & =  \!\begin{multlined}[t][12cm] 
    |R|+ |X \cap (V^+ \cup E^+)| \\
    - |X \cap \{v_{i,h}^+:h \geq q^\star\}| 
    + |\{v_{i,h'}^-:h'\leq q^\star\}| - |\{v_{i,h'}^-: h' \leq q^\star\} \setminus X| 
    \end{multlined} \\
    & \phantom{=i} \begin{multlined}[t][12cm] 
    + |R'[d/2-n-2k'+3]|+|(V^- \cup E^-) \setminus X|
    \\
    - |\{v^-_{i,h'}:h' \leq q^\star\} \setminus X| 
    + |\{v_{i,h}^+:h \geq q^\star\}| -
    |\{v_{i,h}^+:h \geq q^\star\} \cap X| 
    \end{multlined} \\
    & =  \!\begin{multlined}[t][12cm] 
    d/2+k'-2|\{v_{i,h}^+:h \geq q^\star\} \cap X| 
    + (n+1)\\
    -2 |\{v^-_{i,h'}:h' \leq q^\star\} \setminus X| +(d/2-n-2k'+3)+k' 
    \end{multlined} \\
    & = d + 4-2|\{v_{i,h}^+:h \geq q^\star\} \cap X| -2 |\{v^-_{i,h'}:h' \leq q^\star\} \setminus X| \\
    & = d+4 
    - 2 \cdot \left\{  \begin{array}{ll}
    1, & \textrm{ if }\sigma(i) \geq q^\star \\
    0, & \textrm{ otherwise} 
    \end{array}
    \right\} 
    - 2 \cdot  \left\{  \begin{array}{ll}
    1, & \textrm{ if }\sigma'(i) \leq q^\star \\
    0, & \textrm{ otherwise.} 
    \end{array}
    \right\} 
\end{align*}
From this, it follows that $\sigma(i) \geq q^\star \geq \sigma'(i)$ must hold, proving the claim.
\end{claimproof}

\begin{claim}
\label{clm:sigma-ij-1}
 $\sigma(\{i,j\}) \leq \sigma'(\{i,j\})$ for each $i,j \in [k]$ with $i<j$.
\end{claim}
\begin{claimproof}
We use the same arguments as in the proof of Claim~\ref{clm:sigma-i-1}.
Consider the function~$f^\esel_{\{i,j\},3}$. Since $X$ is a solution, there exists some $Y^\esel_{\{i,j\},3} \in \arg\min f^\esel_{\{i,j\},3}$ such that $|X \symmdiff Y^\esel_{\{i,j\},3}| \leq d$. Due to our definition of~$f^\esel_{\{i,j\},3}$ via its compact representation, i.e., path~$P^\esel_{\{i,j\}}$, we know that 
\begin{equation}
\label{eqn:def-Y-eij3}
Y^\esel_{\{i,j\},3} = 
\begin{multlined}[t][10cm]
R'[d/2-m-2k'+3] \cup V^- \cup E^- \setminus E^-_{\{i,j\}} \cup D \\
 \cup \bigcup_{h \leq p^\dagger} \{e_{\{i,j\},h}^+\} 
\cup \bigcup_{1<h \leq p^\dagger} \{v_{\{i,j\},h-1}^-\}     
\end{multlined}
\end{equation}
for some $p^\dagger \in [m]$. Thus, we get
{\allowdisplaybreaks
\begin{align*}
    d & \geq |X \symmdiff Y^\esel_{\{i,j\},3}| = |X \setminus Y^\esel_{\{i,j\},3}| + |Y^\esel_{\{i,j\},3} \setminus X| \\
    & = \!\begin{multlined}[t][12cm] 
    |X \setminus  (R'[d/2-m-2k'+3] \cup V^- 
    \cup  E^- \cup D)|
    \\
    - |X \cap \{e_{\{i,j\},h}^+:h \leq p^\dagger\}|
    + |X \cap \{e_{\{i,j\},h-1}^-: h > p^\dagger\}| 
    \end{multlined} \\
    & \phantom{=i} \begin{multlined}[t][12cm] 
    + |(R'[d/2-m-2k'+3] \cup V^- \cup E^- \cup D) \setminus X| \\
    - |\{e_{\{i,j\},h-1}^-: h > p^\dagger\} \setminus X| 
    + |\{e_{\{i,j\},h}^+:h \leq p^\dagger\} \setminus X| 
    \end{multlined} \\
    & =  \!\begin{multlined}[t][12cm] 
    |R|+ |X \cap (V^+ \cup E^+)| \\
    - |X \cap \{e_{\{i,j\},h}^+:h \leq p^\dagger\}| 
    + |\{e_{\{i,j\},h'}^-:h'\geq p^\dagger\}| - |\{e_{\{i,j\},h'}^-: h' \geq p^\dagger\} \setminus X| 
    \end{multlined} \\
    & \phantom{=i} \begin{multlined}[t][12cm] 
    + |R'[d/2-n-2k'+3]|+|(V^- \cup E^-) \setminus X|
    \\
    - |\{e^-_{\{i,j\},h'}:h' \geq p^\dagger\} \setminus X| 
    + |\{e_{\{i,j\},h}^+:h \leq p^\dagger\}| -
    |\{e_{\{i,j\},h}^+:h \leq p^\dagger\} \cap X| 
    \end{multlined} \\
    & =  \!\begin{multlined}[t][12cm] 
    d/2+k'-2|\{e_{\{i,j\},h}^+:h \leq p^\dagger\} \cap X| 
    + (m+1)\\
    -2 |\{e^-_{\{i,j\},h'}:h' \geq p^\dagger\} \setminus X| +(d/2-m-2k'+3)+k' 
    \end{multlined} \\
    & = d + 4-2|\{e_{\{i,j\},h}^+:h \leq p^\dagger\} \cap X| -2 |\{e^-_{\{i,j\},h'}:h' \geq p^\dagger\} \setminus X| \\
    & = d+4 
    - 2 \cdot \left\{  \begin{array}{ll}
    1, & \textrm{ if }\sigma(\{i,j\}) \leq p^\dagger \\
    0, & \textrm{ otherwise} 
    \end{array}
    \right\} 
    - 2 \cdot  \left\{  \begin{array}{ll}
    1, & \textrm{ if }\sigma'(\{i,j\}) \geq p^\dagger \\
    0, & \textrm{ otherwise.} 
    \end{array}
    \right\} 
\end{align*}
}
From this, it follows that $\sigma(\{i,j\}) \leq p^\dagger \leq \sigma'(\{i,j\})$ must hold, proving the claim.
\end{claimproof}

\begin{claim}
\label{clm:sigma-ij-2}
 $\sigma(\{i,j\}) \geq \sigma'(\{i,j\})$ for each $i,j \in [k]$ with $i<j$.
\end{claim}
\begin{claimproof}
We use the same arguments as in the proof of Claim~\ref{clm:sigma-i-2}.
Consider the function~$f^\esel_{\{i,j\},4}$. Since $X$ is a solution, there exists some $Y^\esel_{\{i,j\},4} \in \arg\min f^\esel_{\{i,j\},4}$ such that $|X \symmdiff Y^\esel_{\{i,j\},4}| \leq d$. Due to our definition of~$f^\esel_{\{i,j\},4}$ via its compact representation, i.e., path~$Q^\esel_{\{i,j\}}$, we know that 
\begin{equation}
\label{eqn:def-Y-eij4}
Y^\esel_{\{i,j\},4} = 
\begin{multlined}[t][10cm]
R'[d/2-m-2k'+3] \cup V^- \cup E^- \setminus E_{\{i,j\}}^- \cup D  \\
\cup \bigcup_{h \geq q^\dagger} \{e_{\{i,j\},h}^+\} 
\cup \bigcup_{n > h \geq q^\dagger} \{e_{\{i,j\},h+1}^-\}. 
\end{multlined}
\end{equation}
for some $q^\dagger \in [m]$. Thus, we get
{
\allowdisplaybreaks
\begin{align*}
    d & \geq |X \symmdiff Y^\esel_{\{i,j\},4}| = |X \setminus Y^\esel_{\{i,j\},4}| + |Y^\esel_{\{i,j\},4} \setminus X| \\
    & = \!\begin{multlined}[t][12cm] 
    |X \setminus  (R'[d/2-m-2k'+3] \cup V^- \cup  E^- \cup D)|
    \\
    - |X \cap \{e_{\{i,j\},h}^+:h \geq q^\dagger\}|
    + |X \cap \{e_{\{i,j\},h+1}^-: h < q^\dagger\}| 
    \end{multlined} \\
    & \phantom{=i} \begin{multlined}[t][12cm] 
    + |(R'[d/2-m-2k'+3] \cup V^- \cup E^- \cup D) \setminus X| \\
    - |\{e_{\{i,j\},h+1}^-: h < q^\dagger\} \setminus X| 
    + |\{e_{\{i,j\},h}^+:h \geq q^\dagger\} \setminus X| 
    \end{multlined} \\
    & =  \!\begin{multlined}[t][12cm] 
    |R|+ |X \cap (V^+ \cup E^+)| \\
    - |X \cap \{e_{\{i,j\},h}^+:h \geq q^\dagger\}| 
    + |\{e_{\{i,j\},h'}^-:h'\leq q^\dagger\}| - |\{e_{\{i,j\},h'}^-: h' \leq q^\dagger\} \setminus X| 
    \end{multlined} \\
    & \phantom{=i} \begin{multlined}[t][12cm] 
    + |R'[d/2-m-2k'+3]|+|(V^- \cup E^-) \setminus X|
    \\
    - |\{e^-_{\{i,j\},h'}:h' \leq q^\dagger\} \setminus X| 
    + |\{e_{\{i,j\},h}^+:h \geq q^\dagger\}| -
    |\{e_{\{i,j\},h}^+:h \geq q^\dagger\} \cap X| 
    \end{multlined} \\
    & =  \!\begin{multlined}[t][12cm] 
    d/2+k'-2|\{e_{\{i,j\},h}^+:h \geq q^\dagger\} \cap X| 
    + (m+1)\\
    -2 |\{e^-_{\{i,j\},h'}:h' \leq q^\dagger\} \setminus X| +(d/2-m-2k'+3)+k' 
    \end{multlined} \\
    & = d + 4-2|\{e_{\{i,j\},h}^+:h \geq q^\dagger\} \cap X| -2 |\{e^-_{\{i,j\},h'}:h' \leq q^\dagger\} \setminus X| \\
    & = d+4 
    - 2 \cdot \left\{  \begin{array}{ll}
    1, & \textrm{ if }\sigma(\{i,j\}) \geq q^\dagger \\
    0, & \textrm{ otherwise} 
    \end{array}
    \right\} 
    - 2 \cdot  \left\{  \begin{array}{ll}
    1, & \textrm{ if }\sigma'(\{i,j\}) \leq q^\dagger \\
    0, & \textrm{ otherwise.} 
    \end{array}
    \right\} 
\end{align*}
}
From this, it follows that $\sigma(\{i,j\}) \geq q^\dagger \geq \sigma'(\{i,j\})$ must hold, proving the claim.
\end{claimproof}

From Claims~\ref{clm:sigma-i-1} and~\ref{clm:sigma-i-2}, we obtain that $\sigma(i)=\sigma'(i)$ for arbitrary~$i \in [k]$. Hence, we have 
\[ V^+_i \cap X =\{v^+_{i,\sigma(i)}\} \qquad \qquad \textrm{and} \qquad \qquad 
V^-_i \setminus X =\{v^-_{i,\sigma(i)}\};
\]
let us say that $X$ \emph{selects} the vertex~$v_{i,\sigma(i)} \in V_i$.

Similarly, from Claims~\ref{clm:sigma-ij-1} and~\ref{clm:sigma-ij-2}, we obtain that $\sigma(\{i,j\})=\sigma'(\{i,j\})$ for each~$i,j \in [k]$ with $i<j$. Hence, we have 
\[ E^+_{\{i,j\}} \cap X =\{e^+_{\{i,j\},\sigma(\{i,j\})}\} \qquad \qquad \textrm{and} \qquad \qquad 
E^-_{\{i,j\}} \setminus X =\{e^-_{\{i,j\},\sigma(\{i,j\})}\};
\]
let us say that $X$ \emph{selects} the edge~$e_{\{i,j\},\sigma(\{i,j\})} \in E_{\{i,j\}}$.
The following claim shows that the vertices and edges selected by~$X$ form a clique of size~$k$ in~$G$, finishing this direction of the proof.
\begin{claim}
\label{clm:inc}
For each $i,j \in [k]$ with $i \neq j$, the edge $e_{\{i,j\},\sigma(\{i,j\})}$ selected by~$X$ 
is incident to the vertex~$v_{i,\sigma(i)}$ selected by~$X$. 
\end{claim}
\begin{claimproof}
Consider the function~$f^\inc_{(i,j)}$, defined via its compact representation, i.e., path~$P^\inc_{(i,j)}$. Since $X$ is a solution for~$I$, there exists some $Y^\inc_{(i,j)} \in \arg\min f^\inc_{(i,j)}$ for which $|X \symmdiff Y^\inc_{(i,j)}| \leq d$.   
Due to the definition of~$P^\inc_{(i,j)}$ and Birkhoff's theorem, there must exist some~$p^\diamond \in [n]$ for which 
\begin{equation}
\label{eqn:def-Y-inc}
Y^\inc_{(i,j)} =
\begin{multlined}[t][10cm]
R'[d/2-m-n-2k'+6] \cup V^- \setminus V^-_i \cup E^- \setminus E^-_{\{i,j\}} \cup D \\
\cup \bigcup_{h \leq p^\diamond} \left( \{v_{i,h}^+\} \cup E^+_j(v_{i,h}) \right)
\cup \bigcup_{h\leq p^\diamond-1} \left( \{v_{i,h}^-\} \cup E^-_j(v_{i,h}) \right).    
\end{multlined}
\end{equation}

Thus, we get
{\allowdisplaybreaks
\begin{align}
    d & \geq |X \symmdiff Y^\inc_{(i,j)}| = |X \setminus Y^\inc_{(i,j)}| + |Y^\inc_{(i,j)} \setminus X| \notag \\
    & = \!\begin{multlined}[t][12cm] 
    |X \setminus  (R'[d/2-n-2k'+6] \cup V^- 
    \cup  E^- \cup D)|
    \\
    - \left|X \cap \bigcup_{h \leq p^\diamond} \{v_{i,h}^+\} \cup E_j^+(v_{i,h})\right|
    + \left|X \cap \bigcup_{h \geq p^\diamond}\{v_{i,h}^-\} \cup E_j^-(v_{i,h})\right| 
    \end{multlined} \notag \\
    & \phantom{=i} \begin{multlined}[t][12cm] 
    + |(R'[d/2-n-2k'+6] \cup V^- \cup E^- \cup D) \setminus X| \\
    - \left|\bigcup_{h \geq p^\diamond} 
    \left(\{v_{i,h}^-\} \cup E_j^-(v_{i,h}) \right) \setminus X\right| 
    + \left| \bigcup_{h \leq p^\diamond} 
    \left( \{v_{i,h}^+\} \cup E_j^-(v_{i,h}) \right) \setminus X\right| 
    \end{multlined} \notag \\
    %%%%%%%%%%%%
    & =  \!\begin{multlined}[t][12cm] 
    |R|+ |X \cap (V^+ \cup E^+)| - \left|X \cap \bigcup_{h \leq p^\diamond} \{v_{i,h}^+\} \cup E_j^+(v_{i,h})\right| \\
     + \left|\bigcup_{h \geq p^\diamond} 
    \{v_{i,h}^-\} \cup E_j^-(v_{i,h})  \right| 
    - \left|\bigcup_{h \geq p^\diamond} 
    \left(\{v_{i,h}^-\} \cup E_j^-(v_{i,h}) \right) \setminus X\right| 
    \end{multlined} \notag \\
    & \phantom{=i} \begin{multlined}[t][12cm] 
    + |R'[d/2-n-2k'+6]|+|(V^- \cup E^-) \setminus X|
    - \left|\bigcup_{h \geq p^\diamond} 
    \left(\{v_{i,h}^-\} \cup E_j^-(v_{i,h}) \right) \setminus X\right| 
    \\ 
    + \left| \bigcup_{h \leq p^\diamond} 
    \left( \{v_{i,h}^+\} \cup E_j^-(v_{i,h}) \right) \right|
    - \left| \bigcup_{h \leq p^\diamond} 
    \left( \{v_{i,h}^+\} \cup E_j^-(v_{i,h}) \right) \cap X\right|
    \end{multlined} \notag \\
    & =  \!\begin{multlined}[t][12cm] 
    d/2+k'-2\left| \bigcup_{h \leq p^\diamond} 
    \left( \{v_{i,h}^+\} \cup E_j^-(v_{i,h}) \right) \cap X\right| 
    + (n+m+2)     
    \\
    -2 \left|\bigcup_{h \geq p^\diamond} 
    \left(\{v_{i,h}^-\} \cup E_j^-(v_{i,h}) \right) \setminus X\right| +(d/2-n-2k'+6)+k' 
    \end{multlined} \label{eqn:inc} \\
    & = d + 8
    -2\left| \bigcup_{h \leq p^\diamond} 
    \left( \{v_{i,h}^+\} \cup E_j^-(v_{i,h}) \right) \cap X\right| 
    -2 \left|\bigcup_{h \geq p^\diamond} 
    \left(\{v_{i,h}^-\} \cup E_j^-(v_{i,h}) \right) \setminus X\right| 
    \notag \\
    & = d+8
    - 2 \cdot \left\{  \begin{array}{ll}
    1, & \textrm{ if }\sigma(i) \leq p^\diamond \\
    0, & \textrm{ otherwise} 
    \end{array}
    \right\} 
    - 2 \cdot \left\{  \begin{array}{ll}
    1, & \textrm{ if }
    e_{\{i,j\},\sigma(\{i,j\})} \in \bigcup_{ h \leq p^\diamond}E_j(v_{i,h})  \\    
    0, & \textrm{ otherwise} 
    \end{array}
    \right\} \notag
    \\
    & \phantom{ =i d+8}
    - 2 \cdot  \left\{  \begin{array}{ll}
    1, & \textrm{ if }\sigma(i) \geq p^\diamond \\
    0, & \textrm{ otherwise} 
    \end{array}
    \right\} 
    - 2 \cdot \left\{  \begin{array}{ll}
    1, & \textrm{ if }
    e_{\{i,j\},\sigma(\{i,j\})} \in \bigcup_{ h \geq p^\diamond}E_j(v_{i,h})  \\    
    0, & \textrm{ otherwise.} 
    \end{array}
    \right\} 
 \notag
\end{align}
}
We remark that to obtain equation~(\ref{eqn:inc}), we used the fact $\bigcup_{h\in [n]} E_j(v_{i,h}) = E_{\{i,j\}}$. 
Notice that the above inequalities can only hold if 
\begin{itemize} 
\item $\sigma(i)=p^\diamond$, 
\item $e_{\{i,j\},\sigma(\{i,j\})} \in \bigcup_{ h \leq p^\diamond}E_j(v_{i,h})$, and 
\item $e_{\{i,j\},\sigma(\{i,j\})} \in \bigcup_{ h \geq p^\diamond}E_j(v_{i,h})$, and 
\end{itemize}
This means that the edge $e_{\{i,j\},\sigma(\{i,j\})}$ selected by~$X$ in $E_{\{i,j\}}$ must be adjacent in~$G$  to the vertex~$v_{i,p^\diamond}=v_{i,\sigma(i)}$, that is, to the vertex selected by~$X$ 
in $V_i$. As this holds for arbitrary distinct indices $i,j \in [k]$, the claim follows.
\end{claimproof}

\subparagraph{Direction ``$\Longleftarrow$'':}
Assume now that $G$ admits a clique~$K$ of size~$k$. 
Define a set~$X \subseteq U$ by
\[X= R \cup D \cup V^- \cup E^- 
\cup  \left( \bigcup_{v \in V(K)} \{v^+\}
\cup  \bigcup_{e \in E(K)}\{e^+\} \right)
\setminus  \left(\bigcup_{v \in V(K)}\{v^-\}
\cup \bigcup_{e \in E(K)}\{e^-\}\right)
. 
\]

Let $v_{i,\sigma(i)}$ denote the vertex of~$K$ in~$V_i$ for each $i \in [k]$. Furthermore, for each $i,j \in [k]$ with $i \neq j$, define the index $\sigma(\{i,j\}) \in [m]$ such that the edge of~$K$ running between $V_i$ and~$V_j$ is the edge~$e_{\{i,j\},\sigma(\{i,j\})}=v_{i,\sigma(i)} v_{j,\sigma(j)}$. 

Let us consider first the functions $f^\init_h$ for $h \in [5]$ together with their unique minimizers~$Y^\init_h$. It is straightforward to verify that their distance from~$X$ in terms of symmetric difference is at most~$d$:
\begin{align*}
|X & \symmdiff Y^\init_1|=|X \symmdiff (R \cup V^+ \cup V^- \cup E^+ \cup E^- \cup D \cup D')| \\
&= \left| 
E^+ \cup V^+ 
\setminus  \left( \bigcup_{v \in V(K)} \{v^+\}
\cup  \bigcup_{e \in E(K)}\{e^+\} \right)
\cup  \left(\bigcup_{v \in V(K)}\{v^-\}
\cup \bigcup_{e \in E(K)}\{e^-\}\right)
\cup D' \right| \\
& = |E|+|V|-k-\binom{k}{2}+k+\binom{k}{2}+|D'|=2(|V|+|E|)=d; 
\end{align*}
%%% INIT2: 
\begin{align*}
|X & \symmdiff Y^\init_2|=|X \symmdiff R| \\
&= \left| 
E^- \cup V^- 
\setminus  
\left(\bigcup_{v \in V(K)}\{v^-\}
\cup \bigcup_{e \in E(K)}\{e^-\}\right)
\cup  
\left( \bigcup_{v \in V(K)} \{v^+\}
\cup  \bigcup_{e \in E(K)}\{e^+\} \right)
\cup D \right| \\
& = |E|+|V|-k-\binom{k}{2}+k+\binom{k}{2}+|D|=2(|V|+|E|)=d; 
\end{align*}
%%% INIT3: 
\begin{align*}
|X & \symmdiff Y^\init_3|=|X \symmdiff (R \cup R' \cup V^+ \cup V^- \cup E^+ \cup E^- \cup D) | \\
&= \left| R' \cup 
E^+ \cup V^+ 
\setminus  \left( \bigcup_{v \in V(K)} \{v^+\}
\cup  \bigcup_{e \in E(K)}\{e^+\} \right)
\cup  \left(\bigcup_{v \in V(K)}\{v^-\}
\cup \bigcup_{e \in E(K)}\{e^-\}\right) \right| \\
& = |R'|+|E|+|V|-k-\binom{k}{2}+k+\binom{k}{2}=2(|V|+|E|)=d; 
\end{align*}
%%% INIT4: 
\begin{align*}
|X & \symmdiff Y^\init_4|=|X \symmdiff (R \cup R'  \cup D) | \\
&= \left| R' \cup 
E^- \cup V^- 
\setminus  
\left(\bigcup_{v \in V(K)}\{v^-\}
\cup \bigcup_{e \in E(K)}\{e^-\}\right)
\cup  
\left( \bigcup_{v \in V(K)} \{v^+\}
\cup  \bigcup_{e \in E(K)}\{e^+\} \right)
\right| \\
& = |R'|+|E|+|V|-k-\binom{k}{2}+k+\binom{k}{2}=2(|V|+|E|)=d;
\end{align*}
%%% INIT5: 
\begin{align*}
|X & \symmdiff Y^\init_5|=|X \symmdiff (R[k'] \cup (R' \setminus R'[k']) \cup V^- \cup E^-  \cup D) | \\
&= \!\begin{multlined}[t][12cm]
    | (R \setminus R[k']) \cup 
(R' \setminus R'[k']) |
\\
+
\left| \left(\bigcup_{v \in V(K)}\{v^-\}
\cup \bigcup_{e \in E(K)}\{e^-\}\right)
\cup  
\left( \bigcup_{v \in V(K)} \{v^+\}
\cup  \bigcup_{e \in E(K)}\{e^+\} \right)
\right| 
\end{multlined}
 \\
& = |R|-k'+|R'|-k'+k+\binom{k}{2}+k+\binom{k}{2}=|R|+|R'|=d.
\end{align*}

Next, notice that the definition of~$X$ ensures that (\ref{eqn:solX-vs-RandR'}), (\ref{eqn:solX-vs-DandD'}), as well as all statements~(\ref{eqn:bound-k'})--(\ref{eqn:select-Eij-})  
are satisfied. Therefore, it is straightforward to check that the calculations within the proofs of Claims~\ref{clm:select-Vi+}--\ref{clm:select-Eij-} remain correct, and all inequalities obtained there  hold with equality.  This shows that $X$ has distance at most~$d$ from the unique minimizers of the functions~$f^\vsel_{i,1}$ and~$f^\vsel_{i,2}$ for each~$i \in [k]$, as well as from the unique minimizers of the functions~$f^\esel_{\{i,j\},1}$ and~$f^\esel_{\{i,j\},2}$ for each~$i,j \in [k]$ with $i<j$.

Consider now the functions~$f^\vsel_{i,3}$ and~$f^\vsel_{i,4}$ for some~$i \in [k]$, and ~$f^\esel_{\{i,j\},3}$ and~$f^\esel_{\{i,j\},4}$ for some~$j \in [k] \setminus \{i\}$. Let us define minimizers~$Y^\vsel_{i,3}, Y^\vsel_{i,4}, Y^\esel_{\{i,j\},3}$, and~$Y^\esel_{\{i,j\},4}$ as in~(\ref{eqn:def-Y-vi3}), (\ref{eqn:def-Y-vi4}), (\ref{eqn:def-Y-eij3}), and (\ref{eqn:def-Y-eij4}), respectively, where we set~$p^\star=q^\star=\sigma(i)$ and $p^\dagger=q^\dagger=\sigma(\{i,j\})$. 
Then it is straightforward to verify that all calculations within the proofs of Claims~\ref{clm:sigma-i-1}--\ref{clm:sigma-ij-2} remain correct, and all inequalities established there hold with equality. This shows that $X$ has distance at most~$d$ from each of the sets~$Y^\vsel_{i,3}, Y^\vsel_{i,4}, Y^\esel_{\{i,j\},3}$, and~$Y^\esel_{\{i,j\},4}$.

It remains to consider the function~$f^\inc_{(i,j)}$ for some $i,j \in [k]$ with $i \neq j$.
Let us define a minimizer~$Y^\inc_{(i,j)}$ as in~(\ref{eqn:def-Y-inc}). 
Again, it is not hard to verify that all calculations within the proof of Claim~\ref{clm:inc} remain correct if we set~$p^\diamond=\sigma(\{i,j\})$, and thus $|X \symmdiff Y^\inc_{(i,j)}| \leq d$ also holds.
Hence, $X$ is indeed a solution for~$I$, proving the correctness of our reduction.
\end{proof}

\subsection{\texorpdfstring{%
Polynomially many minimizers for one function: FPT algorithm for~$d$}{
Polynomially many minimizers for one function: FPT algorithm for d}}
\label{sec:alg-unique}

%For a given instance $\InstSR=(V, f_1,\dots, f_k, d)$, let us use the notation $\LL_i = \arg\min f_i$ for each $i \in [k]$.
In this section, we present an efficient algorithm for the case when our threshold~$d$ is small, assuming that $|\LL_1|$ can bounded by a polynomial of the input size.

\begin{theorem}
\label{thm:unique}
\SR{} can be solved in $|\mathcal{L}_1|g(d) n^c$ time where $c$ is a constant and $g$ is a computable function.
\end{theorem}

Let us consider a slightly more general version of \SR{} which we call \ASR{}.
In this problem, in addition to an instance $\InstSR=(V, f_1, \dots, f_k, d)$ of \SR{}, we are given a set~$Y_0 \subseteq V$ and integer $d_0 \leq d$, and we aim to find a subset $X$ such that
\begin{align}\label{eq:solution_Y0}
|X\symmdiff Y_0 | &\leq d_0  \qquad \textrm{and }\\
\label{eq:solution_Yall}
|X\symmdiff Y_i | &\leq d   \qquad \text{\ for some } Y_i\in\LL_i , \textrm{ for each } i\in [k].
\end{align}
%When $d_0 = d$, the problem is identical to \SR{}.
%The instance is denoted by $(V, f_1, \dots, f_k, d, Y_0,d_0)$.
Observe that we can solve our instance $\InstSR=(V, f_1, \dots, f_k, d)$ by solving the instance $(V,f_2,\dots,f_k,d, Y_0,d_0)$ of \ASR{} for each $Y_0 \in \LL_1$ and $d_0=d$.
Hence, Theorem~\ref{thm:unique} follows from Theorem~\ref{thm:asr} below.

\begin{theorem}
\label{thm:asr}
    \ASR{} can be solved in FPT time when parameterized by~$d$.
\end{theorem}

To prove Theorem~\ref{thm:asr}, we will use the technique of bounded search-trees.  Given an instance $I=(V,f_1,\dots,f_k,d,Y_0,d_0)$, after checking whether $Y_0$ itself is a solution, we search for a minimizer $Y_i \in \LL_i$ for which $d<|Y_0 \symmdiff Y_i| \leq d+d_0$.
It is not hard to see the following.
\begin{observation}
\label{obs:large-diff}
If $X$ is a solution for an instance~$I=(V,f_1,\dots,f_k,d,Y_0,d_0)$ of \textsc{An\-chored Submodular Minimizer}, and 
 $Y_i \in \LL_i$ fulfills $|X \symmdiff Y_i| \leq d$, then for all  $T \subseteq Y_0 \symmdiff Y_i$ with $|T|>d$ it holds that there exists some $v \in T$ with $v \in X  \symmdiff Y_0$.
 \iffalse
\begin{itemize}
\item either there exists $v\in T \cap (Y_0 \setminus Y_i)$ such that $v\not\in X$, 
\item or there exists $v\in T \cap (Y_i \setminus Y_0)$ such that $v\in X$.
\end{itemize}
\fi 
\end{observation}
\begin{proof}
Indeed, assuming that the claim does not hold, we have that 
$T \cap (Y_0 \setminus Y_i) \subseteq X$ 
and that 
$(T \cap (Y_i \setminus Y_0)) \cap X=\emptyset$.
From the former, $T \cap (Y_0 \setminus Y_i) \subseteq X \setminus Y_i$ follows, while the latter implies  $T \cap (Y_i \setminus Y_0) \subseteq Y_i \setminus X$.
Thus,
\[
X \symmdiff Y_i = (X \setminus Y_i) \cup (Y_i \setminus X) \supseteq  (T \cap (Y_0 \setminus Y_i)) \cup (T \cap (Y_i \setminus Y_0))  = T\cap (Y_0 \symmdiff Y_i)=T.
\]
Hence, $|X \symmdiff Y_i|\geq |T| >d$,  contradicting our assumption that $X$ is a solution for~$I$.
\end{proof}
Our algorithm will compute in $O^*(2^d)$ time\footnote{The $O^*()$ notation hides polynomial factors.} a set~$T \subseteq Y_0\setminus Y_i$ of size $d<|T|\leq d+d_0$ that contains some element~$v$ fulfilling the above conditions. Then,
by setting $Y_0 \leftarrow Y_0\symmdiff \{v\}$ and reducing the value of~$d_0$ by one, we obtain an equivalent instance~$I'$ of \ASR{} which we solve by applying recursion.

\paragraph*{Description of our algorithm.}
Our algorithm will make ``guesses''; nevertheless, it is a deterministic one, where guessing a value from a given set $U$ is interpreted as branching into $|U|$ branches. We continue the computations in each branch, and whenever a branch returns a solution for the given instance, we return it; if all branches reject the instance (by outputting ``No''), we also reject it.
See Algorithm~\ref{alg:ASR} for a pseudo-code description. 

We start by checking whether $Y_0$ is a solution for our instance~$I=(V,f_1,\dots,f_k,d,Y_0,d_0)$, that is, whether it satisfies~\eqref{eq:solution_Yall}. 
This can be done in polynomial time, since 
the set function $\gamma_i(Z)  =\min \{|Z \symmdiff Y_i| : Y_i \in \mathcal{L}_i\}$ is known to be submodular and can be computed via a maximum flow computation~\cite{KakimuraKKO22}.
   If $Y_0$ satisfies~\eqref{eq:solution_Yall}, i.e., $\gamma_i(Y_0) \leq d$ for each $i \in [k]$, then we output~$Y_0$; note that~\eqref{eq:solution_Y0} is obviously satisfied by $Y_0$, so $Y_0$ is a solution for~$I$. 
   
   Otherwise, if $d_0=0$, then we output ``No'' as in this case the only possible solution could be~$Y_0$.   
   We proceed by fixing an index~$i \in [k]$ such that $\gamma_i(Y_0)>d$, that is, $|Y_0 \symmdiff Y| > d$ for all minimizers $Y \in \LL_i$. 
       
    \begin{observation}\label{obs:unique1}
    If $X$ is a solution for~$I$ that satisfies   $|X\symmdiff Y_i|\leq d$ for some $Y_i\in \LL_i$, then  $|Y_i\symmdiff Y_0|\leq d+d_0$.
    \end{observation}
    \begin{proof}
        Since $X$ is a solution for~$I$, we have $|X\symmdiff Y_0|\leq d_0$, and thus the triangle inequality implies $|Y_i\symmdiff Y_0|\leq |X\symmdiff Y_i| + |X\symmdiff Y_0|\leq d+d_0$.
   \end{proof}

By our choice of~$i$ and Observation~\ref{obs:unique1}, we know that $d < |Y_0 \symmdiff Y_i| \leq d+d_0$. We are going to compute a set $T \subseteq Y_0 \symmdiff Y_i$ with the same bounds on its cardinality, i.e., $d<|T| \leq d+d_0$.

To this end, we compute a compact representation $G(\LL_i)$ of the distributive lattice $\LL_i$; let
$\PP = \{U_0, U_1, \dots, U_b, U_\infty\}$ be the partition of $V$ in this representation.    

Next, we proceed with an iterative procedure which also involves a set of guesses. 
We start by setting $Y=Y_0$ and $T=\emptyset$. We will maintain a family of \emph{fixed sets} from $\PP$ for which we already know whether they are in $Y_i$ or not (according to our guesses); initially, no set from~$\PP$ is fixed.

After this initialization, we start an iteration where at each step we check whether $Y \in \LL_i$ or $|T| > d$. 
If yes, then we stop the iteration. 
If not, then it can be shown that one of the following conditions holds: \begin{description}
    \item[Condition 1:] there exists a set $S \in \PP$ such that $S \cap Y \neq \emptyset$ and 
    $S \setminus Y \neq \emptyset$;
    \item[Condition 2:] there exists an edge $(S,S')$ in $G(\LL_i)$ for which $S \subseteq Y$ but $S' \cap Y=\emptyset$.
\end{description}
If Condition~1 holds for some set $S \in \PP$, then 
%first we check whether $S$ is fixed. If so, then we already know whether  $S \subseteq Y_i$; if not, then 
we guess whether $S$ is contained in~$Y_i$. If $S \subseteq Y_i$ according to our guesses, then we add $S \setminus Y$ to~$T$; otherwise, we add $S \cap Y$ to~$T$. 
In either case, we declare~$S$ as fixed,
and proceed with the next iteration. 

By contrast, if Condition~1 fails, but Condition~2 holds for some edge $(S,S')$ in $G(\LL_i)$ with endpoints $S,S' \in \PP$, then we proceed as follows. 
If both $S$ and $S'$ are fixed, then we stop and reject the current set of guesses. 
If $S$ is fixed but $S'$ is not, then we add all elements of~$S'$ to $T$. 
If $S'$ is fixed but $S$ is not, then we add $S$ to~$T$. 
If neither~$S$ nor~$S'$ is fixed, then we guess whether $S$ is contained in~$Y_i$ or not, and in the former case we add $S'$ to~$T$, while in the latter case we add~$S$ to~$T$. In all cases except for the last one, we declare both~$S$ and~$S'$ as fixed; in the last case declare only~$S$ as fixed.

Next, we modify $Y$ to reflect the current value of~$T$ by updating $Y$ to~$Y_0 \triangle T$.
If $|T|>d+d_0$, then we reject the current branch. If $d<|T| \leq d+d_0$, then we finish the iteration; otherwise, we proceed with the next iteration. 

Finally, when the iteration stops, we guess a vertex $v \in T$, define $Y'_{0,v}=Y_0 \symmdiff \{v\}$ and call the algorithm recursively on the instance $I'_v:=(V,f_1,\dots,f_k,d,Y'_{0,v},d_0-1)$.

\begin{varalgorithm}{\shortASR}
\caption{Solving \ASR{} on $I=(V,f_1,\dots,f_k,d,Y_0,d_0)$.}
\label{alg:ASR}
\begin{algorithmic}[1]
\ForAll{$j \in [k]$} %
 {compute the value $\gamma_j=\min \{|Y_0 \symmdiff Y| : Y \in \arg\min f_j\}$. }
\EndFor
\If{$\gamma_j \leq d$ for each $j \in [k]$} {\bf return} $Y_0$.
\EndIf 
\If{$d_0=0$} {\bf return} ``No''.
\EndIf
\State {Fix an index $i \in [k]$ such that $\gamma_i>d$.}
\State Compute the graph $G(\LL_i)$,  and let $\PP$ be its vertex set.
\State Set $T:= \emptyset$ and $Y := Y_0$, and 
$\fixed(S) := \ff$ for each $S \in \PP$.
\While{ $Y \notin \LL_i$ and $|T|\leq d$}
    \If{$\exists S \in \PP: S \cap Y_0 \neq \emptyset, 
                            S \setminus Y_0 \neq \emptyset$}
        %\If{$\fixed(S)=\ff$} {guess $\cont(S)$ from $\{\ff,\tt\}.$}
        %\EndIf
        \State Guess $\cont(S)$ from $\{\ff,\tt\}.$ \label{li:guess1}
        \If{$\cont(S)=\tt$} {set $T := T \cup (S \setminus Y)$.} \label{li:roundS+}
        \Else { set $T := T \cup (S \cap Y)$.} \label{li:roundS-}
        \EndIf
        \State Set $\fixed(S) := \tt$.
    \Else{ Find an edge $(S,S') \in G(\LL_i)$ such that $S \subseteq Y$ and $S' \cap Y =\emptyset$.} 
        \If{$\fixed(S)=\tt$}
            \If{$\fixed(S')=\tt$} {\bf return} ``No''. 
            \Else { set $T:=T \cup S'$ and $\fixed(S'):=\tt$.} \label{li:addS'_fixed}
            \EndIf
        \Else \Comment{$\fixed(S)=\ff$.}
            \If{$\fixed(S')=\tt$} {set $T:=T \cup S$  and $\fixed(S):=\tt$.}  \label{li:addS_fixed}
            \Else { guess $\cont(S)$ from $\{\ff,\tt\}.$} \label{li:guess2}
                \If{$\cont(S)=\tt$} {set $T:=T \cup S'$,  $\fixed(S):=\fixed(S'):=\tt$.}\label{li:addS'_guessed}%
                \Else { set $T:=T \cup S$  and $\fixed(S):=\tt$.}  \label{li:addS_guessed}
                \EndIf
            \EndIf
        \EndIf
    \EndIf
    \State Set $Y:=Y_0 \symmdiff T$. \label{li:updateY}
    \If{$|T|>d+d_0$} {\bf return} ``No''.
    \EndIf
\EndWhile
\State Guess a vertex $v$ from~$T$.
\State Set $Y'_{0,v}= Y_0 \symmdiff \{v\}$ and $I'_v=(V,f_1,\dots,f_k,d,Y'_{0,v},d_0-1)$.
\State {\bf return} $\mathsf{\shortASR}(I'_v)$.
\end{algorithmic}
\end{varalgorithm}

\begin{proof}[Proof of Theorem~\ref{thm:asr}]
We first prove the correctness of the algorithm.
Clearly, for $d_0=0$, the algorithm either correctly outputs the solution~$Y_0$, or rejects the instance. Hence, we can apply induction on~$d_0$, and assume that the algorithm works correctly when called for an instance with a smaller value for $d_0$.

We show that any set~$X$ returned by the algorithm is a solution for~$I$. First, this is clear if $X=Y_0$, as the algorithm explicitly checks whether $\gamma_i(Y_0)\leq d$ holds for each~$i \in [k]$;
second, if $X$ was returned by a recursive call on some instance~$I'_v$, then by our induction hypothesis we know that $X$ is a solution for $I'_v=(V,f_1,\dots,f_k,d,Y'_{0,v},d_0-1)$. Hence, $X$ satisfies~\eqref{eq:solution_Yall}; moreover, by $|X \symmdiff Y'_{0,v}| \leq d_0-1$, it also satisfies $|X \symmdiff Y_0| \leq d_0$, because $|Y_0 \symmdiff Y'_{0,v}|=1$.

Let us now prove that if $I$ admits a solution~$X$, then the algorithm correctly returns a solution for~$I$. Let $Y_i \in \LL_i$ be a minimizer such that $|X \symmdiff Y_i| \leq d$ where $i$ is the index fixed for which $\gamma_i(Y_0)>d$. 

\begin{claim}
\label{clm:invariant}
Assuming that all guesses made by the algorithm are correct,
in the iterative process of modifying~$T$ and $Y$ it will always hold that 
\begin{description}
    \item[(i)] $T \subseteq Y_i \symmdiff Y_0$, and 
    \item[(ii)] for each $S \in \PP$:
    \begin{description}
        \item[(a)] if $S$ is fixed, then $S \subseteq Y \Longleftrightarrow S \subseteq Y_i$, and $S \cap Y=\emptyset \Longleftrightarrow S \cap Y_i=\emptyset$, and
        \item[(b)] if $v \in S$ and $S$ is not fixed, then $v \in Y \Longleftrightarrow v \in Y_0$.
    \end{description}      
\end{description}
\end{claim}

\begin{claimproof}
The claim clearly holds initially, when $T=\emptyset$, $Y=Y_0$ and no set is fixed.
Consider now the $j$-th run of the iteration for some~$j>0$. Let $T$ and~$Y$ be as at the beginning of the iteration. 
Note that $Y=Y_0 \symmdiff T$ due to the last step of the $(j-1)$-st run of the iteration on line~\ref{li:updateY}. 

Suppose that Condition~1 holds for some set~$S \in \PP$. Using Birkhoff's representation theorem, we know that either $S\subseteq Y_i$ or $S\cap Y_i =\emptyset$. 
By Condition~1 and claim~(ii/a) of our induction hypothesis (\IH{} for short), $S$ is not fixed.
If $S\subseteq Y_i$, then by claim~(ii/b) of our~\IH{} (applicable as $S$ is not fixed), we know $S \setminus Y = S \setminus Y_0  \subseteq Y_i \symmdiff Y_0$.
Similarly,  if $S\cap Y_i=\emptyset$, then we get $S \cap Y = S \cap Y_0  \subseteq Y_i \symmdiff Y_0$. 
Assuming that the algorithm correctly guesses which among these two cases holds, the elements added to~$T$ on line~\ref{li:roundS+} or~\ref{li:roundS-} are indeed contained in~$Y_i \symmdiff Y_0$, proving that~(i) remains true in this case. 
%If $T$ denotes the resulting set, then  updating~$Y$ to $Y_0 \symmdiff T$ on line~\ref{li:updateY} ensures that  (ii) remains too as well.

Suppose now that Condition~2 holds for some edge~$(S,S')$ in~$G(\LL_i)$. 
By our inductive hypothesis, $S$ and~$S'$ cannot both be fixed (assuming correct guesses), since $S \subseteq Y_i$ and $S' \cap Y_i=\emptyset$ would then contradict Birkhoff's representation theorem, as the edge $(S,S')$ would leave the closed set $\PP(Y_i):=\{S: S \in \PP, S\subseteq Y_i\}$ in~$G(\LL_i)$. 

Now, if $S$ is fixed, then by claim~(ii/a) of our \IH, 
$S \subseteq Y$ implies $S \subseteq Y_i$. This means that $S' \subseteq Y_i$ must hold as well, since $\PP(Y_i)$ is closed. 
Moreover, if $S'$ is not fixed, then $S' \cap Y=\emptyset$ implies $S' \cap Y_0=\emptyset$ due to claim~(ii/b) of the \IH; hence $S' \subseteq Y_i \symmdiff Y_0$. 
The same arguments hold for the case when initially neither~$S$ nor~$S'$ is fixed, and the algorithm correctly guesses that $S$ is contained in~$Y_i$. 
We obtain that
the vertices added to~$T$ on line~\ref{li:addS'_fixed} or~\ref{li:addS'_guessed} are indeed contained in~$Y_i \symmdiff Y_0$.

By contrast, if $S$ is not fixed, then  claim~(ii/b) of our \IH{} and $S \subseteq Y$ imply $S \subseteq Y_0$.
Now, if $S'$ is fixed, then claim~(ii/a) of our \IH{} means that $S' \cap Y_i=\emptyset$; hence we get $S \notin \PP(Y_i)$, because $\PP(Y_i)$ is closed in~$G(\LL_i)$. Then $S \subseteq Y_i \symmdiff Y_0$ follows, so the set added to~$T$ on line~\ref{li:addS_fixed} is indeed contained in~$Y_i \symmdiff Y_0$. 
The only remaining case is when neither~$S$ nor~$S'$ is fixed initially, and the algorithm guesses $S \not\subseteq~Y_i$. 
Recall that by claim~(ii/b) of the \IH{}, we have $S \subseteq Y_0$, so a correct guess immediately yields $S \subseteq Y_i \symmdiff Y_0$, and thus the set added to~$T$ on line~\ref{li:addS_guessed} is in $Y_i \symmdiff Y_0$ as well. This proves that (i) remains true at the end of the $j$-th run of the algorithm as well. 

Observe that whenever the algorithm puts some element~$v \in V$ into~$Y$, it also declares the set in~$\PP$ containing~$v$ fixed; from this, it immediately follows that claim~(ii/b) remains true.
To see that claim~(ii/a) is maintained well, note first that the algorithm never removes or adds vertices of a fixed set from~$Y$ or to~$Y$, respectively. 
Therefore, we only need to check those sets that we declared fixed during this iteration. By the above arguments, it is not hard to verify that whenever we declare some set~$S$ as fixed, then we also ensure $S \subseteq Y \Longleftrightarrow S \subseteq Y_i$ when updating~$Y$ to~$Y_0 \symmdiff T$. 
This finishes the proof of the claim.
\end{claimproof}

Next, we show that in each run of the iteration, Condition~1 or Condition~2 holds. Indeed, if neither holds, then (1) $Y=\bigcup_{U \in \PP'} U$ for some~$\PP' \subseteq \PP$, and (2)  no edge leaves~$\PP'$ in $G(\LL_i)$. Hence, $Y \in \LL_i$ by Birkhoff's representation theorem. However, since $|T| \leq d$ holds at the beginning of each iteration,  $|Y \symmdiff Y_0| =|T| \leq d$ follows, contradicting our choice of~$i$. 

Therefore, in each run of the iteration, at least one element of~$V$ is put into~$T$. Thus, the iteration stops after at most $d+1$ runs, at which point the obtained set~$T$ has size greater than~$d$. 
Using now  statement~(i) of Claim~\ref{clm:invariant}, Observation~\ref{obs:large-diff} yields that $T$ contains at least one vertex from~$X \symmdiff Y_0$. Assuming that the algorithm guesses such a vertex~$v$ correctly, it is clear 
that our solution~$X$ for~$I$ will also be a solution for the instance~$I'_v$. 
Using our inductive hypothesis, we obtain that the recursive call returns a correct solution for~$I'_v$ which, as discussed already, will be a solution for~$I$ as well. Hence, our algorithm is correct.

Finally, let us bound the running time. Consider the search tree~$\mathcal{T}$ where each node corresponds to a  call of Algorithm~\ref{alg:ASR}.
Note that the value of~$d_0$ decreases by one in each recursive call, and the algorithm stops when $d_0=0$. Hence $\mathcal{T}$ has depth at most~$d_0$. Consider the guesses made during the execution of a single call of the algorithm (without taking into account the guesses in the recursive calls): we make at most one guess in each iteration on line~\ref{li:guess1} or on line~\ref{li:guess2}, leading to at most~$2^{d+1}$ possibilities. Then the algorithm further guesses a vertex from~$T$, leading to a total of at most~$2^{d+1}|T| \leq 2^{d+1}(d+d_0)=2^{O(d)}$ possibilities; recall that $d_0 \leq d$. 
We get that the number of nodes in our search tree 
is $2^{d_0 O(d)}$. 
Since all computations for a fixed series of guesses take polynomial time, we obtain that the running time is indeed fixed-parameter tractable with parameter $d$.
\end{proof}

\section{Conclusion}
\label{sec:conclusion}

In this paper, we studied the computational complexity of \SR{}, and provided a complete computational map of the problem with respect to the parameters~$k$ and~$d$. 
%offering dichotomies for the case when one of these parameters is a constant, and giving an FPT algorithm for the combined parameter~$(k,d)$.
We also investigated how the computational complexity of the problem changes if on instances where one or more of the functions~$f_i$ has only a bounded number of minimizers.
A summary of our results is shown in Table~\ref{tab:summary}.
We remark that our algorithmic results can be adapted in a straightforward way to a slightly generalized problem:
given $k$ submodular functions $f_1, \dots, f_k$ with non-negative integers $d_1, \dots, d_k$, we aim to find a set $X$ such that, for each $i\in [k]$, there exists some set~$Y_i \in \arg\min f_i$ with $|X \symmdiff Y_i| \leq d_i$ for each $i\in [k]$.

\begin{table}
    \centering
    \begin{tabular}{llll}
        Restriction  & Parameter & Complexity  & Reference\\ \hline
        $d=0$ & $-$ & in $\mathsf{P}$ & Obs.~\ref{obs:d=0} \\
        $d \geq 1$ constant & $-$ & $\NP$-hard & Cor.~\ref{cor:SR-NPhard-d-atleast1}\\
        $k\leq 2$ & $-$ & in $\mathsf{P}$ & Thm.~\ref{thm:k=2},~\cite{LeeSW15} \\ 
        $k \geq 3$ constant, $\ell_1=\dots=\ell_{k-1}=1$ & $-$ &  $\NP$-hard  & Cor. \ref{cor:k-atleast-3} \\
        $\ell_1=\dots = \ell_k=1$ & $-$ & $\NP$-hard & Obs.~\ref{obs:all-unique},~\cite{DBLP:journals/mst/FrancesL97} \\        
        $\ell_1=\dots = \ell_k=1$ & $k$ & FPT & Obs.~\ref{obs:all-unique}, \cite{GrammNR03} \\
        $\ell_k \leq |V|$ & $k$ & $\mathsf{W}[1]$-hard & Thm.~\ref{thm:W1hard-k} \\
        $\ell_k=|V|^{O(1)}$ & $k$ & in $\mathsf{XP}$ & Prop.~\ref{prop:few-minimizers} \\        
        $-$ & $(k, \ell_k)$ & FPT & Prop.~\ref{prop:few-minimizers} \\
        $\ell_1=|V|^{O(1)}$ & $d$ & FPT & Thm.~\ref{thm:unique} \\
        $-$ & $(k,d)$ & FPT & Thm.~\ref{thm:SR-fpt} \\[6pt]
    \end{tabular}
    \caption{A summary of our results. Recall that $k$ is the number of input functions, and $d$ the distance threshold given. 
    In this table, $\ell_i=|\LL_i|$ for some~$i \in [k]$ denotes the number of minimizers for function~$f_i$, and we assume $\ell_1 \leq \ell_2 \leq \dots \leq \ell_k$.}
    \label{tab:summary}
\end{table}

An interesting question concerns the case where one of the input functions has only polynomially many minimizers. Theorem~\ref{thm:unique} establishes that this variant of \SR{} is fixed-parameter tractable with respect to parameter~$d$; this result is based on our algorithm for \ASR{} in Theorem~\ref{thm:asr}, 
which yields a running time $2^{O(d^2)}|V|^{O(1)}$ on some instance~$I=(V,f_1,\dots,d_k,d)$.
Can we improve this to obtain an algorithm that runs in single-exponential time with respect to~$d$, i.e., in $2^{O(d)}|V|^{O(1)}$ time?
\iffalse
Specifically, we showed that \SR{} can be solved in polynomial time when $k \leq 2$, but is NP-hard if $k$ is a constant with $k \geq 3$;
\SR{} can be solved in polynomial time when $d = 0$, but is NP-hard if $d$ is a constant with $d \geq 1$;
\SR{} is fixed-parameter tractable when parameterized by $(k, d)$.
We also proposed an FPT algorithm for the case when some submodular function $f_i$ has a polynomial number of minimizers, based on the bounded search-tree technique.
\fi 

As mentioned in Section~\ref{sec:relatedwork},
\SR{} is related to recoverable robustness.
We can consider the robust recoverable variant of submodular minimization:
given submodular functions $f_0, f_1, \dots, f_k$, we aim to find a set $X$ that minimizes
\[
f_0(X) + \max_{i\in[k]} \min_{Y_i: |Y_i\symmdiff X|\leq d} f_i (X_i).
\]
The optimal value is lower-bounded by $f_0(Y_0) + \max_{i\in[k]} f_i (Y_i)$ where $Y_i\in \arg\min f_i$ for each $i\in \{0,1,\dots, k\}$.
Our results imply that we can decide efficiently whether the optimal value attains this lower bound or not, when $d$ and $k$ are parameters, or when $f_0$ has polynomially many minimizers.

\begin{figure}
    \centering
%    \vspace{36em}
    \begin{tikzpicture}[every node/.style={transform shape}]
    \node (0) at (-9,0){};
    \node[rectangle,inner sep=5pt,draw] (1) at (-6.5,0){$k$ unbounded?};
    \draw[-Stealth,very thick](0)--(1);
    \node[rectangle,inner sep=5pt,draw] (2) at (-6.5,-1.6){$\ell_1=\dots=\ell_k=1$?};
    \draw[-Stealth,very thick](1.south) to node[pos=0.3, left] {\scriptsize{no}} (2.north);  
    \node[rectangle,inner sep=5pt,draw,rounded corners,fill=FPTcolor] (3) at (-1.1,-1.6){\small{\parbox{4.5cm}{FPT w.r.t.~$k$ (Obs.~\ref{obs:all-unique})}}};
    \draw[-{Stealth},very thick] (2.east) to node[pos=0.5, above] {\scriptsize{yes}} (3.west);
    \node[rectangle,inner sep=5pt,draw] (4) at (-6.5,-3.2){$d$ parameter?};   
    \draw[-Stealth,very thick] (2.south) to node[pos=0.5, left] {\scriptsize{no}} (4.north);  
    \node[rectangle,inner sep=5pt,draw,rounded corners,fill=FPTcolor] (5) at (-1.1,-3.2){\small{\parbox{4.5cm}{FPT w.r.t.~$(k,d)$ \\ (Thm.~\ref{thm:SR-fpt})}}};
    \draw[-Stealth,very thick] (4.east) to node[pos=0.5, above] {\scriptsize{yes}} (5.west);  
    \node[rectangle,inner sep=5pt,draw] (6) at (-6.5,-4.8){$k \leq 2$?};   
    \draw[-Stealth,very thick] (4.south) to node[pos=0.5, left] {\scriptsize{no}} (6.north);  
    \node[rectangle,inner sep=5pt,draw,rounded corners,fill=inPcolor] (7) at (-1.1,-4.8){\small{\parbox{4.5cm}{in $\mathsf{P}$ (Thm.~\ref{thm:k=2})}}};
    \draw[-Stealth,very thick] (6.east) to node[pos=0.5, above] {\scriptsize{yes}} (7.west);  
    \node[rectangle,inner sep=5pt,draw] (8) at (-6.5,-6.4){$\ell_k$ parameter?};   
    \draw[-Stealth,very thick] (6.south) to node[pos=0.5, left] {\scriptsize{no}} (8.north); 
    \node[rectangle,inner sep=5pt,draw,rounded corners,fill=FPTcolor] (9) at (-1.1,-6.4){\small{\parbox{4.5cm}{FPT w.r.t. $(k,\ell_k)$ \\ (Prop.~\ref{prop:few-minimizers})}}};
    \draw[-Stealth,very thick] (8.east) to node[pos=0.5, above] {\scriptsize{yes}} (9.west);  
    \node[rectangle,inner sep=5pt,draw] (10) at (-6.5,-8.2){$\ell_k=|V|^{O(1)}$?};   
    \draw[-Stealth,very thick] (8.south) to node[pos=0.5, left] {\scriptsize{no}} (10.north); 
    \node[rectangle,inner sep=5pt,draw,rounded corners,fill=XPcolor] (11) at (-1.1,-8.2){\small{\parbox{4.5cm}{in $\mathsf{XP}$ w.r.t.~$k$ (Prop.~\ref{prop:few-minimizers}) \\[2pt] $\mathsf{W}[1]$-hard w.r.t.~$k$ even if \\ $\ell_1=1, \ell_k \leq |V|$ (Thm.~\ref{thm:W1hard-k})}}};
    \draw[-Stealth,very thick] (10.east) to node[pos=0.5, above] {\scriptsize{yes}} (11.west);  
    \node[rectangle,inner sep=5pt,draw,rounded corners,fill=NPhardcolor] (12) at (-6.5,-10.4){\small{\parbox{4.0cm}{$\NP$-hard for constant~$k$ with $k\geq 3$, 
    even if \\ $\ell_1=\dots=\ell_{k-1}=1$ \\ (Cor.~\ref{cor:k-atleast-3})}}};
    \draw[-Stealth,very thick] (10.south) to node[pos=0.5, left] {\scriptsize{no}} (12.north);     
    \node[rectangle,inner sep=5pt,draw] (13) at (-6.5,1.6){$d$ unbounded?};
    \draw[-{Stealth},very thick] (1.north) to node[pos=0.5, left] {\scriptsize{yes}}  (13.south);
    \node[rectangle,inner sep=5pt,draw,rounded corners,fill=NPhardcolor] (14) at (-1.1,1.6){\small{\parbox{4.5cm}{$\NP$-hard even if\\  $\ell_1=\dots=\ell_k=1$ \\ (Obs.~\ref{obs:all-unique})}}};
    \draw[-Stealth,very thick] (13.east) to node[pos=0.5, above] {\scriptsize{yes}} (14.west);  
    \node[rectangle,inner sep=5pt,draw] (15) at (-6.5,3.2){$d=0$?};
    \draw[-Stealth,very thick] (13.north) to node[pos=0.5, left] {\scriptsize{no}} (15.south);     
    \node[rectangle,inner sep=5pt,draw,rounded corners,fill=inPcolor] (16) at (-1.1,3.2){\small{\parbox{4.5cm}{in $\mathsf{P}$ (Obs.~\ref{obs:d=0})}}};
    \draw[-Stealth,very thick] (15.east) to node[pos=0.5, above] {\scriptsize{yes}} (16.west);  
    \node[rectangle,inner sep=5pt,draw] (17) at (-6.5,4.8){$\ell_1=|V|^{O(1)}$?};
    \draw[-Stealth,very thick] (15.north) to node[pos=0.5, left] {\scriptsize{no}} (17.south);     
    \node[rectangle,inner sep=5pt,draw,rounded corners,fill=FPTcolor] (18) at (-1.1,4.8){\small{\parbox{4.5cm}{FPT w.r.t.~$d$ (Thm.~\ref{thm:unique})}}};
    \draw[-Stealth,very thick] (17.east) to node[pos=0.5, above] {\scriptsize{yes}} (18.west);  
    \node[rectangle,inner sep=5pt,draw,rounded corners,fill=NPhardcolor] (19) at (-6.5,6.6){\small{\parbox{4.0cm}{$\NP$-hard for constant~$d$ with $d \geq 1$ (Cor.~\ref{cor:SR-NPhard-d-atleast1})}}};
    \draw[-Stealth,very thick] (17.north) to node[pos=0.5, left] {\scriptsize{no}} (19.south);  
    \end{tikzpicture}
    \caption{Decision diagram for determining the complexity of \SR{}. We use the notation $\ell_i=|\arg \min f_i|$, and assume that $\ell_1 \leq \ell_2 \leq \dots \leq \ell_k$. The diagram concerns restrictions of the problem, each obtained by designated each of~$k$ and~$d$ as (i) \emph{unbounded}, (ii) a \emph{parameter}, or (iii) a \emph{constant}, and in addition restricting each of $\ell_1$ and~$\ell_k$ either according to one of the options (i)--(iii), or restricting it to be (iv) \emph{polynomially bounded in~$|V|$}.
    }
    \label{fig-decision-diagram}
\end{figure}

%\bibliographystyle{plain}
%\bibliography{main}
%\bibliographystyle{plain}
\bibliography{mainjournal}

\begin{thebibliography}{10}

\bibitem{AmirPR16}
Amihood Amir, Haim Paryenty, and Liam Roditty.
\newblock Configurations and minority in the string consensus problem.
\newblock {\em Algorithmica}, 74:1267--1292.
\newblock \href {https://doi.org/10.1007/s00453-015-9996-7}
  {\path{doi:10.1007/s00453-015-9996-7}}.

\bibitem{Birkhoff37}
Garrett Birkhoff.
\newblock Rings of sets.
\newblock {\em Duke Math. J.}, 3(3):443--454, 1937.
\newblock \href {https://doi.org/10.1215/S0012-7094-37-00334-X}
  {\path{doi:10.1215/S0012-7094-37-00334-X}}.

\bibitem{BockenhauerFHMSS12}
Hans{-}Joachim B{\"{o}}ckenhauer, Karin Freiermuth, Juraj Hromkovic, Tobias
  M{\"{o}}mke, Andreas Sprock, and Bj{\"{o}}rn Steffen.
\newblock Steiner tree reoptimization in graphs with sharpened triangle
  inequality.
\newblock {\em J. Discrete Algorithms}, 11:73--86, 2012.
\newblock \href {https://doi.org/10.1016/J.JDA.2011.03.014}
  {\path{doi:10.1016/J.JDA.2011.03.014}}.

\bibitem{Bonsma10}
Paul Bonsma.
\newblock Most balanced minimum cuts.
\newblock {\em Discrete Applied Mathematics}, 158:261--276, 2010.
\newblock \href {https://doi.org/10.1016/j.dam.2009.09.010}
  {\path{doi:10.1016/j.dam.2009.09.010}}.

\bibitem{BoriaP10}
Nicolas Boria and Vangelis~Th. Paschos.
\newblock Fast reoptimization for the minimum spanning tree problem.
\newblock {\em J. Discrete Algorithms}, 8(3):296--310, 2010.
\newblock \href {https://doi.org/10.1016/J.JDA.2009.07.002}
  {\path{doi:10.1016/J.JDA.2009.07.002}}.

\bibitem{BoucherBD}
Christina Boucher, Daniel~G. Brown, and Stephane Durocher.
\newblock On the structure of small motif recognition instances.
\newblock In {\em 15th International Symposium on String Processing and
  Information Retrieval (SPIRE 2008)}, pages 269--281. Springer, 2009.
\newblock \href {https://doi.org/10.1007/978-3-540-89097-3_26}
  {\path{doi:10.1007/978-3-540-89097-3_26}}.

\bibitem{Busing12}
Christina B{\"{u}}sing.
\newblock Recoverable robust shortest path problems.
\newblock {\em Networks}, 59(1):181--189, 2012.
\newblock \href {https://doi.org/10.1002/NET.20487}
  {\path{doi:10.1002/NET.20487}}.

\bibitem{CyganEtAl2015}
Marek Cygan, Fedor~V. Fomin, {\L}ukasz Kowalik, Daniel Lokshtanov, D{\'a}niel
  Marx, Marcin Pilipczuk, Micha{\l} Pilipczuk, and Saket Saurabh.
\newblock {\em Parameterized algorithms}.
\newblock Springer, 2015.
\newblock \href {https://doi.org/10.1007/978-3-319-21275-3}
  {\path{doi:10.1007/978-3-319-21275-3}}.

\bibitem{dinitz1976structure}
Efim~A Dinitz, Alexander~V Karzanov, and Michael~V Lomonosov.
\newblock On the structure of the system of minimum edge cuts of a graph.
\newblock {\em Issledovaniya po Diskretnoi Optimizatsii}, pages 290--306, 1976.

\bibitem{DouradoMPRPA15}
Mitre~Costa Dourado, Dirk Meierling, Lucia~Draque Penso, Dieter Rautenbach,
  F{\'{a}}bio Protti, and Aline~Ribeiro de~Almeida.
\newblock Robust recoverable perfect matchings.
\newblock {\em Networks}, 66(3):210--213, 2015.
\newblock \href {https://doi.org/10.1002/NET.21624}
  {\path{doi:10.1002/NET.21624}}.

\bibitem{downey-fellows-FPC-book}
Rod~G. Downey and Michael~R. Fellows.
\newblock {\em Fundamentals of Parameterized Complexity}.
\newblock Texts in Computer Science. Springer, 2013.
\newblock \href {https://doi.org/10.1007/978-1-4471-5559-1}
  {\path{doi:10.1007/978-1-4471-5559-1}}.

\bibitem{FellowsHRV-multicolored-hardness}
Micheal~R. Fellows, Danny Hermelin, Frances~A. Rosamond, and St\'ephane
  Vialette.
\newblock On the parameterized complexity of multiple-interval graph problems.
\newblock {\em Theor. Comput. Sci.}, 410:53--61, 2009.
\newblock \href {https://doi.org/10.1016/j.tcs.2008.09.065}
  {\path{doi:10.1016/j.tcs.2008.09.065}}.

\bibitem{HLW21}
Dennis Fischer, Tim~A. Hartmann, Stefan Lendl, and Gerhard~J. Woeginger.
\newblock An investigation of the recoverable robust assignment problem.
\newblock In {\em 16th International Symposium on Parameterized and Exact
  Computation, {IPEC} 2021}, volume 214 of {\em LIPIcs}, pages 19:1--19:14,
  2021.
\newblock \href {https://doi.org/10.4230/LIPICS.IPEC.2021.19}
  {\path{doi:10.4230/LIPICS.IPEC.2021.19}}.

\bibitem{DBLP:journals/mst/FrancesL97}
Moti Frances and Ami Litman.
\newblock On covering problems of codes.
\newblock {\em Theory Comput. Syst.}, 30(2):113--119, 1997.
\newblock \href {https://doi.org/10.1007/s002240000044}
  {\path{doi:10.1007/s002240000044}}.

\bibitem{abs-2111-09691}
Marc Goerigk, Stefan Lendl, and Lasse Wulf.
\newblock On the recoverable traveling salesman problem.
\newblock {\em CoRR}, 2021.
\newblock \href {http://arxiv.org/abs/2111.09691} {\path{arXiv:2111.09691}}.

\bibitem{GrammNR03}
Jens Gramm, Rolf Niedermeier, and Peter Rossmanith.
\newblock Fixed-parameter algorithms for {C}losest {S}tring and related
  problems.
\newblock {\em Algorithmica}, 37:25--42, 2003.
\newblock \href {https://doi.org/10.1007/s00453-003-1028-3}
  {\path{doi:10.1007/s00453-003-1028-3}}.

\bibitem{HommelsheimMMP}
Felix Hommelsheim, Nicole Megow, Komal Muluk, and Britta Peis.
\newblock Recoverable robust optimization with commitment.
\newblock {\em CoRR}, 2023.
\newblock \href {http://arxiv.org/abs/2306.08546} {\path{arXiv:2306.08546}}.

\bibitem{HradovichKZ17}
Mikita Hradovich, Adam Kasperski, and Pawel Zielinski.
\newblock Recoverable robust spanning tree problem under interval uncertainty
  representations.
\newblock {\em J. Comb. Optim.}, 34(2):554--573, 2017.
\newblock \href {https://doi.org/10.1007/S10878-016-0089-6}
  {\path{doi:10.1007/S10878-016-0089-6}}.

\bibitem{ItoKKKO22}
Takehiro Ito, Naonori Kakimura, Naoyuki Kamiyama, Yusuke Kobayashi, and Yoshio
  Okamoto.
\newblock A parameterized view to the robust recoverable base problem of
  matroids under structural uncertainty.
\newblock {\em Oper. Res. Lett.}, 50(3):370--375, 2022.
\newblock \href {https://doi.org/10.1016/J.ORL.2022.05.001}
  {\path{doi:10.1016/J.ORL.2022.05.001}}.

\bibitem{KakimuraKKO22}
Naonori Kakimura, Naoyuki Kamiyama, Yusuke Kobayashi, and Yoshio Okamoto.
\newblock Submodular reassignment problem for reallocating agents to tasks with
  synergy effects.
\newblock {\em Discret. Optim.}, 44:100631, 2022.
\newblock \href {https://doi.org/10.1016/j.disopt.2021.100631}
  {\path{doi:10.1016/j.disopt.2021.100631}}.

\bibitem{KS-SWAT}
Naonori Kakimura and Ildik{\'{o}} Schlotter.
\newblock Parameterized complexity of submodular minimization under
  uncertainty.
\newblock In Hans~L. Bodlaender, editor, {\em 19th Scandinavian Symposium on
  Algorithm Theory, {SWAT} 2024}, volume 294 of {\em LIPIcs}, pages
  30:1--30:17, 2024.
\newblock \href {https://doi.org/10.4230/LIPICS.SWAT.2024.30}
  {\path{doi:10.4230/LIPICS.SWAT.2024.30}}.

\bibitem{KratschLMPW20}
Stefan Kratsch, Shaohua Li, D{\'{a}}niel Marx, Marcin Pilipczuk, and Magnus
  Wahlstr{\"{o}}m.
\newblock Multi-budgeted directed cuts.
\newblock {\em Algorithmica}, 82(8):2135--2155, 2020.
\newblock \href {https://doi.org/10.1007/S00453-019-00609-1}
  {\path{doi:10.1007/S00453-019-00609-1}}.

\bibitem{LachmannLW21}
Thomas Lachmann, Stefan Lendl, and Gerhard~J. Woeginger.
\newblock A linear time algorithm for the robust recoverable selection problem.
\newblock {\em Discret. Appl. Math.}, 303:94--107, 2021.
\newblock \href {https://doi.org/10.1016/J.DAM.2020.08.012}
  {\path{doi:10.1016/J.DAM.2020.08.012}}.

\bibitem{LeeSW15}
Yin~Tat Lee, Aaron Sidford, and Sam~Chiu{-}wai Wong.
\newblock A faster cutting plane method and its implications for combinatorial
  and convex optimization.
\newblock In {\em {IEEE} 56th Annual Symposium on Foundations of Computer
  Science, {FOCS} 2015}, pages 1049--1065, 2015.
\newblock \href {https://doi.org/10.1109/FOCS.2015.68}
  {\path{doi:10.1109/FOCS.2015.68}}.

\bibitem{LendlPT22}
Stefan Lendl, Britta Peis, and Veerle Timmermans.
\newblock Matroid bases with cardinality constraints on the intersection.
\newblock {\em Math. Program.}, 194(1):661--684, 2022.
\newblock \href {https://doi.org/10.1007/S10107-021-01642-1}
  {\path{doi:10.1007/S10107-021-01642-1}}.

\bibitem{Lenstra83}
Hendrik~W. Lenstra.
\newblock Integer programming with a fixed number of variables.
\newblock {\em Mathematics of Operations Research}, 8(4):538--548, 1983.
\newblock \href {https://doi.org/10.1287/moor.8.4.538}
  {\path{doi:10.1287/moor.8.4.538}}.

\bibitem{LiebchenLMS09}
Christian Liebchen, Marco~E. L{\"{u}}bbecke, Rolf~H. M{\"{o}}hring, and
  Sebastian Stiller.
\newblock The concept of recoverable robustness, linear programming recovery,
  and railway applications.
\newblock In {\em Robust and Online Large-Scale Optimization: Models and
  Techniques for Transportation Systems}, volume 5868 of {\em Lecture Notes in
  Computer Science}, pages 1--27. Springer, 2009.
\newblock \href {https://doi.org/10.1007/978-3-642-05465-5\_1}
  {\path{doi:10.1007/978-3-642-05465-5\_1}}.

\bibitem{Monnot15}
J{\'{e}}r{\^{o}}me Monnot.
\newblock A note on the traveling salesman reoptimization problem under vertex
  insertion.
\newblock {\em Inf. Process. Lett.}, 115(3):435--438, 2015.
\newblock \href {https://doi.org/10.1016/J.IPL.2014.11.003}
  {\path{doi:10.1016/J.IPL.2014.11.003}}.

\bibitem{MurotaDCA}
Kazuo Murota.
\newblock {\em Discrete Convex Analysis}.
\newblock SIAM, 2003.
\newblock \href {https://doi.org/10.1137/1.9780898718508}
  {\path{doi:10.1137/1.9780898718508}}.

\bibitem{DBLP:journals/mp/Orlin09}
James~B. Orlin.
\newblock A faster strongly polynomial time algorithm for submodular function
  minimization.
\newblock {\em Math. Program.}, 118(2):237--251, 2009.
\newblock \href {https://doi.org/10.1007/s10107-007-0189-2}
  {\path{doi:10.1007/s10107-007-0189-2}}.

\bibitem{DBLP:conf/stoc/Orlin13}
James~B. Orlin.
\newblock Max flows in $o(nm)$ time, or better.
\newblock In {\em Symposium on Theory of Computing Conference, STOC 2013},
  pages 765--774, 2013.
\newblock \href {https://doi.org/10.1145/2488608.2488705}
  {\path{doi:10.1145/2488608.2488705}}.

\bibitem{Picard1980}
Jean-Claude Picard and Maurice Queyranne.
\newblock On the structure of all minimum cuts in a network and applications.
\newblock {\em Mathematical Programming Studies}, 13:8--16, 1980.
\newblock \href {https://doi.org/10.1007/BFb0120902}
  {\path{doi:10.1007/BFb0120902}}.

\bibitem{Schaefer78}
Thomas~J. Schaefer.
\newblock The complexity of satisfiability problems.
\newblock In {\em Proceedings of the Tenth ACM Symposium on Theory of Computing
  (STOC '78)}, pages 216--226, 1978.
\newblock \href {https://doi.org/10.1145/800133.804350}
  {\path{doi:10.1145/800133.804350}}.

\bibitem{schrijver-book}
Alexander Schrijver.
\newblock {\em Combinatorial Optimization: Polyhedra and Efficiency}.
\newblock Springer, 2003.

\end{thebibliography}

\end{document}